\documentclass[12pt]{article}
\usepackage{amssymb,amsmath,amsthm,amsxtra,amsfonts}
\usepackage{graphicx,psfrag,epsf}
\usepackage{enumerate}
\usepackage[longnamesfirst]{natbib}
\usepackage{url} 
\usepackage[colorlinks=true,citecolor=Blue,urlcolor=Blue,linkcolor=Blue]{hyperref}
\usepackage{geometry}
\usepackage{aliascnt}
\usepackage{color}
\usepackage[usenames,dvipsnames,table]{xcolor}
\usepackage{subcaption}
\usepackage{bbm}
\usepackage{graphicx}
\usepackage{mathrsfs}
\usepackage{tikz}
\usetikzlibrary{plotmarks,arrows,calc,patterns}
\usepackage{threeparttable}
\usepackage{longtable}

\newcommand{\R}{\ensuremath{\mathbf{R}}}
\newcommand{\Z}{\ensuremath{\mathbf{Z}}}

\newcommand{\Nn}{\ensuremath{\mathbf{N}}}
\newcommand{\E}{\ensuremath{\mathbb{E}}}
\newcommand{\dx}{\ensuremath{\mathrm{d}}}

\newcommand{\Cov}{\ensuremath{\mathrm{Cov}}}

\newtheorem{theorem}{Theorem}[section]
\newtheorem{proposition}{Proposition}[section]
\newtheorem{lemma}{Lemma}[section]
\newtheorem{corollary}{Corollary}[section]
\newtheorem{example}{Example}[section]
\newtheorem{remark}{Remark}[section]
\newtheorem{assumption}{Assumption}[section]

\newcommand{\blind}{0}

\begin{document}

\def\spacingset#1{\renewcommand{\baselinestretch}%
{#1}\small\normalsize} \spacingset{1}


\if0\blind
{
 
 \title{\bf Machine Learning Time Series Regressions With an Application to Nowcasting}
 
 \author{Andrii Babii\thanks{Department of Economics, University of North Carolina--Chapel Hill - Gardner Hall, CB 3305
 		Chapel Hill, NC 27599-3305. Email: babii.andrii@gmail.com} \and Eric Ghysels\thanks{Department of Economics and Kenan-Flagler Business School, University of North Carolina--Chapel Hill and CEPR. Email: eghysels@unc.edu.} \and
 	Jonas Striaukas\thanks{LIDAM UC Louvain and Research Fellow of the Fonds de la Recherche Scientifique - FNRS. Email: jonas.striaukas@gmail.com.}}
 
 \maketitle
 
\thispagestyle{empty}

} \fi

\if1\blind
{
  \bigskip
  \bigskip
  \bigskip
  \begin{center}
    {\LARGE\bf Machine Learning Time Series Regressions with an Application to Nowcasting}
\end{center}
  \medskip
} \fi

\bigskip
\begin{abstract}
	\noindent This paper introduces structured machine learning regressions for high-dimensional time series data potentially sampled at different frequencies. The sparse-group LASSO estimator can take advantage of such time series data structures and outperforms the unstructured LASSO. We establish oracle inequalities for the sparse-group LASSO estimator within a framework that allows for the mixing processes and recognizes that the financial and the macroeconomic data may have heavier than exponential tails. An empirical application to nowcasting US GDP growth indicates that the estimator performs favorably compared to other alternatives and that text data can be a useful addition to more traditional numerical data.
\end{abstract}

\noindent%
\vfill

\setcounter{page}{0}

\newpage
\spacingset{1.45} 
\section{Introduction}
\label{sec:intro}

The statistical imprecision of quarterly gross domestic product (GDP) estimates, along with the fact that the first estimate is available with a delay of nearly a month, pose a significant challenge to policy makers, market participants, and other observers with an interest in monitoring the state of the economy in real time; see, e.g., \cite{ghysels2018forecasting} for a recent discussion of macroeconomic data revision and publication delays. A term originated in meteorology, nowcasting pertains to the prediction of the present and very near future.  Nowcasting is intrinsically a mixed frequency data problem as the object of interest is a low-frequency data series (e.g., quarterly GDP), whereas  the real-time information (e.g., daily, weekly, or monthly) can be used to update the state, or to put it differently, to {\it nowcast} the low-frequency series of interest. Traditional methods used for nowcasting rely on dynamic factor models that treat the underlying low frequency series of interest as a latent process with high frequency data noisy observations. These models are naturally cast in a state-space form and inference can be performed using likelihood-based methods and Kalman filtering techniques; see \cite{banbura2013now} for a recent survey.

\smallskip

So far, nowcasting has mostly relied on the so-called standard macroeconomic data releases, one of the most prominent examples being the Employment Situation report released on the first Friday of every month by the US Bureau of Labor Statistics. This report includes the data on the nonfarm payroll employment, average hourly earnings, and other summary statistics of the labor market activity. Since most sectors of the economy move together over the business cycle, good news for the labor market is usually good news for the aggregate economy. In addition to the labor market data, the nowcasting models typically also rely on construction spending, (non-)manufacturing report, retail trade, price indices, etc., which we will call the traditional macroeconomic data. One prominent example of nowcast is produced by the Federal Reserve Bank of New York relying on a dynamic factor model with thirty-six predictors of different frequencies; see \cite{bok2018macroeconomic} for more details.

Thirty-six predictors of traditional macroeconomic series may be viewed as a small number compared to hundreds of other potentially available and useful nontraditional series.  For instance, macroeconomists increasingly rely on nonstandard data such as textual analysis via machine learning, which means potentially hundreds of series. A textual analysis data set based on {\it Wall Street Journal} articles that has been recently made available features a taxonomy of 180 topics; see \cite{bybee2019structure}. Which topics are relevant? How should they be selected? \cite{thorsrud2018words} constructs a daily business cycle index based on quarterly GDP growth and textual information contained in the daily business newspapers relying on a  dynamic factor model where time-varying sparsity is enforced upon the factor loadings using a latent threshold mechanism. His work shows the feasibility of traditional state space setting, yet the challenges grow when we also start thinking about adding other potentially high-dimensional data sets, such as payment systems information or GPS tracking data. Studies for Canada (\cite{galbraith2018nowcasting}),  Denmark
(\cite{carlsen2010dankort}), India (\cite{raju2019nowcasting}), Italy (\cite{aprigliano2019using}),  Norway (\cite{aastveit2020nowcasting}), Portugal (\cite{duarte2017mixed}), and the United States (\cite{barnett2016nowcasting}) find that payment transactions can help to nowcast and to forecast GDP and private consumption in the short term; see also \cite{moriwaki2019nowcasting} for nowcasting unemployment rates with smartphone GPS data, among others. We could quickly reach  numerical complexities involved with estimating high-dimensional state space models, making the dynamic factor model approach potentially computationally prohibitively complex and slow, although some alternatives to the Kalman filter exist for the large data environments; see e.g., \cite{chan2009efficient} and \cite{delle2019efficient}.

\smallskip

In this paper, we study nowcasting a low-frequency series -- focusing on the key example of US GDP growth -- in a data-rich environment, where our data not only includes conventional high-frequency series but also nonstandard data generated by textual analysis of financial press articles. We find that our nowcasts are either superior to or at par with those posted by the Federal Reserve Bank of New York (henceforth NY Fed). This is the case when (a) we compare our approach with the NY Fed using the same data, or (b) when we compare our approach using an expanded high-dimensional data set. The former is a comparison of methods, whereas the latter pertains to the value of the additional (nonstandard) big data. To deal with such massive nontraditional data sets, instead of using the likelihood-based dynamic factor models, we rely on a different approach that involves machine learning methods based on the regularized empirical risk minimization principle and data sampled at different frequencies. We adopt the MIDAS (Mixed Data Sampling) projection approach which is more amenable to  high-dimensional data environments. Our general framework also includes the standard same frequency time series regressions.

\smallskip

Several novel contributions are required to achieve our goal. First, we argue that the high-dimensional mixed frequency time series regressions involve certain data structures that once taken into account should improve the performance of unrestricted estimators in small samples. These structures are represented by groups covering lagged dependent variables and groups of lags for a single (high-frequency) covariate. To that end, we leverage on the sparse-group LASSO (sg-LASSO) regularization that accommodates conveniently such structures; see \cite{simon2013sparse}. The attractive feature of the sg-LASSO estimator is that it allows us to combine effectively the approximately sparse and dense signals; see e.g., \cite{carrasco2016sample} for a comprehensive treatment of high-dimensional dense time series regressions as well as \cite{mogliani2020bayesian} for a complementary to ours Bayesian view of penalized MIDAS regressions.

We recognize that the economic and financial time series data are persistent and often heavy-tailed, while the bulk of the machine learning methods assumes i.i.d.\ data and/or exponential tails for covariates and regression errors; see \cite{belloni2018high} for a comprehensive review of high-dimensional econometrics with i.i.d.\ data. There have been several recent attempts to expand the asymptotic theory to settings involving time series dependent data, mostly for the LASSO estimator. For instance, \cite{kock2015oracle} and \cite{uematsu2019high} establish oracle inequalities for regressions with i.i.d.\ errors with sub-Gaussian tails; \cite{wong2017lasso} consider $\beta$-mixing series with exponential tails; \cite{wu2016performance}, \cite{han2017high}, and \cite{chernozhukov2019lasso} establish oracle inequalities for causal Bernoulli shifts with independent innovations and polynomial tails under the functional dependence measure of \cite{wu2005nonlinear}; see also \cite{medeiros2016l1} and \cite{medeiros2017adaptive} for results on the adaptive LASSO based on the triplex tail inequality for mixingales of \cite{jiang2009uniform}.

\smallskip

Despite these efforts, there is no complete estimation and prediction theory for high-dimensional time series regressions under the assumptions comparable to the classical GMM and QML estimators. For instance, the best currently available results are too restrictive for the MIDAS projection model, which is typically an example of a causal Bernoulli shift with \textit{dependent innovations}. Moreover, the \textit{mixing processes} with \textit{polynomial tails} that are especially relevant for the financial and macroeconomic time series have not been properly treated due to the fact that the sharp Fuk-Nagaev inequality was not available in the relevant literature until recently. The Fuk-Nagaev inequality, see \cite{fuk1971probability}, describes the concentration of sums of random variables with a mixture of the sub-Gaussian and the polynomial tails. It provides sharp estimates of tail probabilities unlike Markov's bound in conjunction with the Marcinkiewicz-Zygmund or Rosenthal's moment inequalities.

This paper fills these gaps in the literature relying on the Fuk-Nagaev inequality for $\tau$-mixing processes of \cite{babiietalinference} and establishes the nonasymptotic and asymptotic estimation and prediction properties of the sg-LASSO projections under weak tail conditions and potential misspecification. The class of $\tau$-mixing processes is fairly rich covering the $\alpha$-mixing processes, causal linear processes with infinitely many lags of $\beta$-mixing processes, and nonlinear Markov processes; see \cite{dedecker2004coupling,dedecker2005new} for more details, as well as \cite{carrasco2002mixing} and \cite{francq2019garch} for mixing properties of various processes encountered in time series econometrics. Our weak tail conditions require at least $4+\epsilon$ finite moments for covariates, while the number of finite moments for the error process can be as low as $2+\nu$, provided that covariates are sufficiently integrable.  From the theoretical point of view, we impose {\it approximate sparsity}, relaxing the assumption of exact  sparsity of the projection coefficients and allowing for other forms of misspecification (see \cite{giannone2018economic} for further discussion on the topic of sparsity). Lastly, we cover the LASSO and the group LASSO as special cases.

\smallskip

The rest of the paper is organized as follows. Section \ref{sec:midas} presents the setting of (potentially mixed frequency) high-dimensional time series regressions. Section \ref{sec:oracle} characterizes nonasymptotic estimation and prediction accuracy of the sg-LASSO estimator for $\tau$-mixing processes with polynomial tails. We report on a Monte Carlo study in Section \ref{sec:mc} which provides further insights regarding the validity of our theoretical analysis in small sample settings typically encountered in empirical applications. Section \ref{sec:empirical} covers the empirical application. Conclusions appear in Section \ref{sec:conclusion}.

\paragraph{Notation:} For a random variable $X\in\R$, let $\|X\|_q=(\E|X|^q)^{1/q}$ be its $L_q$ norm with $q\geq 1$. For $p\in\Nn$, put $[p] = \{1,2,\dots,p\}$. For a vector $\Delta\in\R^p$ and a subset $J\subset [p]$, let $\Delta_J$ be a vector in $\R^p$ with the same coordinates as $\Delta$ on $J$ and zero coordinates on $J^c$. Let $\mathcal{G}$ be a partition of $[p]$ defining the group structure, which is assumed to be known to the econometrician. For a vector $\beta\in\R^p$, the sparse-group structure is described by a pair $(S_0,\mathcal{G}_0)$, where $S_0=\{j\in[p]:\;\beta_j\ne 0 \}$ and  $\mathcal{G}_0 = \left\{G\in\mathcal{G}:\; \beta_{G} \ne 0\right\}$ are the support and respectively the group support of $\beta$. We also use $|S|$ to denote the cardinality of arbitrary set $S$. For $b\in\R^p$, its $\ell_q$ norm is denoted as $|b|_q = \left(\sum_{j\in[p]}|b_j|^q\right)^{1/q}$ for $q\in[1,\infty)$ and $|b|_\infty = \max_{j\in[p]}|b_j|$ for $q=\infty$. For $\mathbf{u},\mathbf{v}\in\R^T$, the empirical inner product is defined as $\langle \mathbf{u},\mathbf{v}\rangle_T = T^{-1}\sum_{t=1}^T u_tv_t$ with the induced empirical norm $\|.\|_T^2=\langle.,.\rangle_T=|.|_2^2/T$. For a symmetric $p\times p$ matrix $A$, let $\mathrm{vech}(A)\in\R^{p(p+1)/2}$ be its vectorization consisting of the lower triangular and the diagonal elements. For $a,b\in\R$, we put $a\vee b = \max\{a,b\}$ and $a\wedge b = \min\{a,b\}$. Lastly, we write $a_n\lesssim b_n$ if there exists a (sufficiently large) absolute constant $C$ such that $a_n\leq C b_n$ for all $n\geq 1$ and $a_n\sim b_n$ if $a_n\lesssim b_n$ and $b_n\lesssim a_n$.

\section{High-dimensional mixed frequency regressions}\label{sec:midas}

Let $\{y_t:t\in[T]\}$ be the target low frequency series observed at integer time points $t\in[T]$. Predictions of $y_t$ can involve its lags as well as a large set of covariates and lags thereof. In the interest of generality, but more importantly because of the empirical relevance we allow the covariates to be sampled at higher frequencies - with same frequency being a special case. More specifically, let there be $K$ covariates $\{x_{t-(j-1)/m,k},j\in[m],t\in[T],k\in[K]\}$ possibly measured at some higher frequency with $m\geq 1$ observations for every $t$  and consider the following regression model
\begin{equation*}
\label{eq:ardlmidas}
\phi(L)y_t = \rho_0 + \sum_{k=1}^K\psi(L^{1/m}; \beta_k)x_{t,k} + u_t,\qquad t\in[T],
\end{equation*}
where $\phi(L) = I - \rho_1L - \rho_2L^2 - \dots - \rho_JL^J$ is a low-frequency lag polynomial and $\psi(L^{1/m};\beta_k)x_{t,k} = 1/m \sum_{j=1}^m\beta_{j,k}x_{t-(j-1)/m,k}$ is a high-frequency lag polynomial. For $m$ = 1, we have a standard autoregressive distributed lag (ARDL) model, which is the workhorse regression model of the time series econometrics literature. Note that the polynomial $\psi(L^{1/m}; \beta_k)x_{t,k}$ involves the same $m$ number of high-frequency lags for each covariate $k\in[K]$, which is done for the sake of simplicity and can easily be relaxed; see Section~\ref{sec:empirical}.

\smallskip

The ARDL-MIDAS model (using the terminology of \cite{andreou2013should}) features $J+1+m\times K$ parameters. In the big data setting with a large number of covariates sampled at high-frequency, the total number of parameters may be large compared to the effective sample size or even exceed it. This leads to poor estimation and out-of-sample prediction accuracy in finite samples. For instance, with $m$ = 3 (quarterly/monthly setting) and 35 covariates at 4 lagged quarters, we need to estimate $m\times K=420$ parameters. At the same time, say the post-WWII quarterly GDP growth series has less than 300 observations.

\smallskip

The LASSO estimator, see \cite{tibshirani1996regression}, offers an appealing convex relaxation of a difficult nonconvex best subset selection problem. It allows increasing the precision of predictions via the selection of sparse and parsimonious models. In this paper, we focus on the structured sparsity with additional dimensionality reductions that aim to improve upon the unstructured LASSO estimator in the time series setting.

\smallskip

First, we parameterize the high-frequency lag polynomial following the MIDAS regression or the distributed lag econometric literature (see  \cite{ghysels2006predicting}) as
\begin{equation*}
\label{eq:midaspoly}
\psi(L^{1/m};\beta_k)x_{t,k} = \frac{1}{m}\sum_{j=1}^m\omega((j-1)/m;\beta_k)x_{t-(j-1)/m,k},
\end{equation*}
where $\beta_k$ is $L$-dimensional vector of coefficients with $L\leq m$ and $\omega:[0,1]\times \R^L\to \R$ is some weight function. Second, we approximate the weight function as
\begin{equation}
\label{eq:dicleg}
\omega(u;\beta_k) \approx \sum_{l=1}^L\beta_{k,l}w_l(u),\qquad u\in[0,1],
\end{equation}
where $\{w_l:\; l=1,\dots,L \}$ is a collection of functions, called the \textit{dictionary}. The simplest example of the dictionary consists of algebraic power polynomials, also known as \cite{almon1965distributed} polynomials in the time series regression analysis literature. More generally, the dictionary may consist of arbitrary approximating functions, including the classical orthogonal bases of $L_2[0,1]$; see Online Appendix Section \ref{appendix:dictionaries} for more examples. Using orthogonal polynomials typically reduces the multicollinearity and leads to better finite sample performance. It is worth mentioning that the specification with dictionaries deviates from the standard MIDAS regressions and leads to a computationally attractive convex optimization problem, cf. \cite{FMO13}.

\smallskip

The size of the dictionary $L$ and the number of covariates $K$ can still be large and the \textit{approximate sparsity} is a key assumption imposed throughout the paper. With the approximate sparsity, we recognize that assuming that most of the estimated coefficients are zero is overly restrictive and that the approximation error should be taken into account. For instance, the weight function may have an infinite series expansion,  nonetheless, most can be captured by a relatively small number of orthogonal basis functions. Similarly, there can be a large number of economically relevant predictors, nonetheless, it might be sufficient to select only a smaller number of the most relevant ones to achieve good out-of-sample forecasting performance. Both model selection goals can be achieved with the LASSO estimator. However, the LASSO does not recognize that covariates at different (high-frequency) lags are temporally related. 

\smallskip

In the baseline model, all high-frequency lags (or approximating functions once we parameterize the lag polynomial) of a single covariate constitute a group. We can also assemble all lag dependent variables into a group. Other group structures could be considered, for instance combining various covariates into a single group, but we will work with the simplest group setting of the aforementioned baseline model. The sparse-group LASSO (sg-LASSO) allows us to incorporate such structure into the estimation procedure. In contrast to the group LASSO, see \cite{yuan2006model}, the sg-LASSO promotes sparsity \textit{between} and \textit{within} groups, and allows us to capture the predictive information from each group, such as approximating functions from the dictionary or specific covariates from each group.
\begin{figure}[h]
	\centering
	\begin{subfigure}{0.33\textwidth} 
		\includegraphics[width=\textwidth]{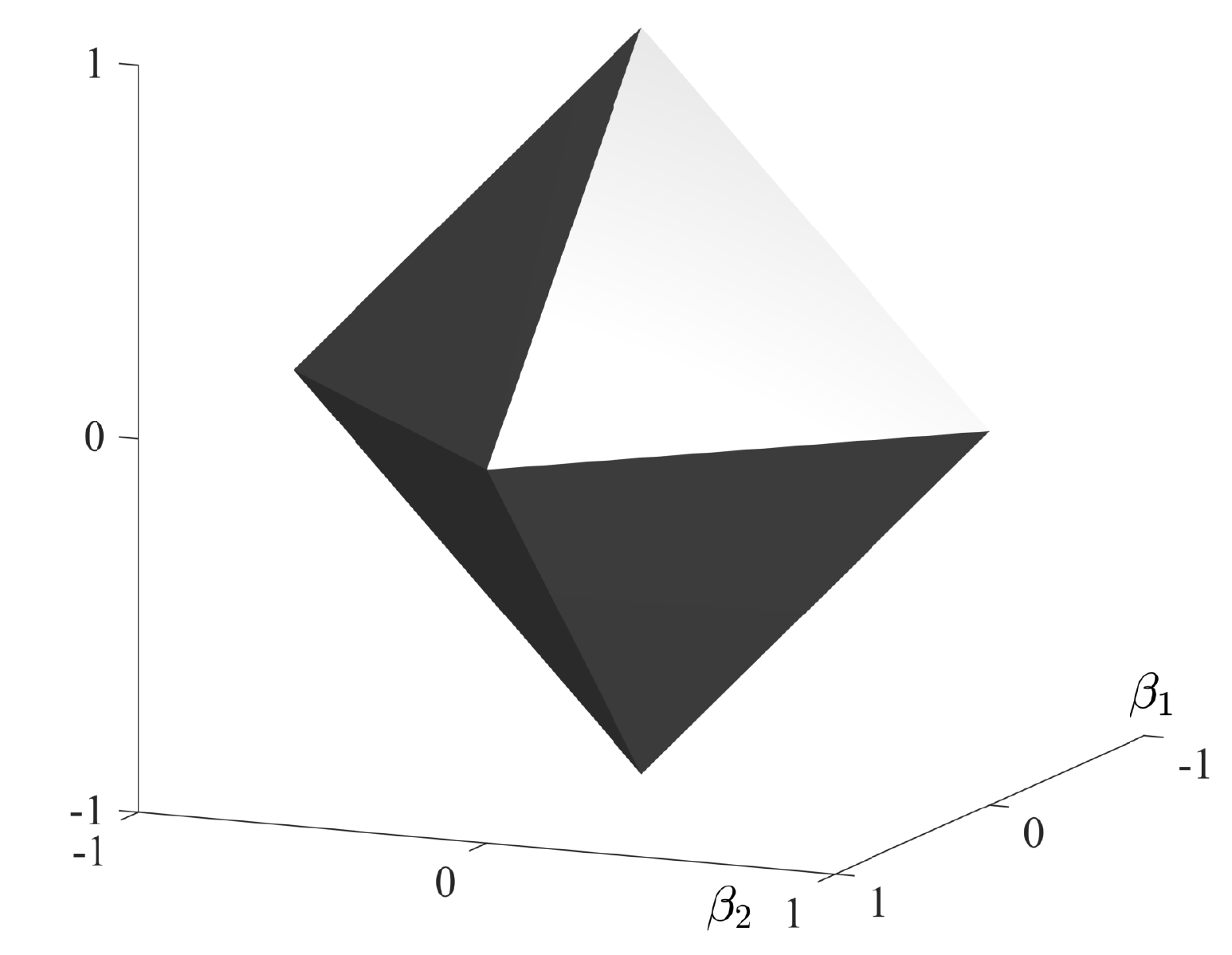}
		\caption{\footnotesize LASSO, $\alpha = 1$} 
	\end{subfigure}
	\begin{subfigure}{0.33\textwidth} 
		\includegraphics[width=\textwidth]{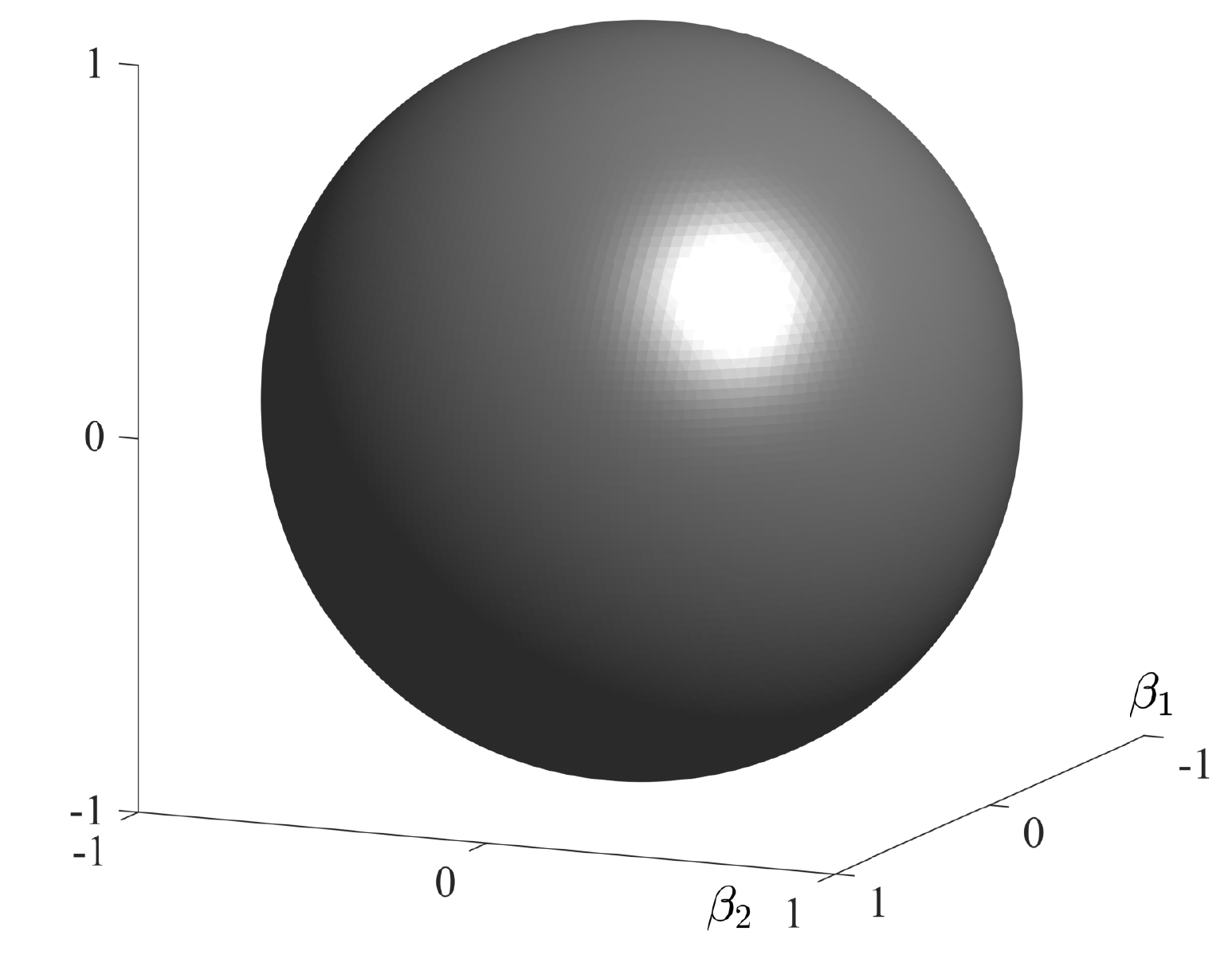}
		\caption{\footnotesize group LASSO with 1 group, $\alpha = 0$} 
	\end{subfigure}
	\begin{subfigure}{0.33\textwidth}
		\includegraphics[width=\textwidth]{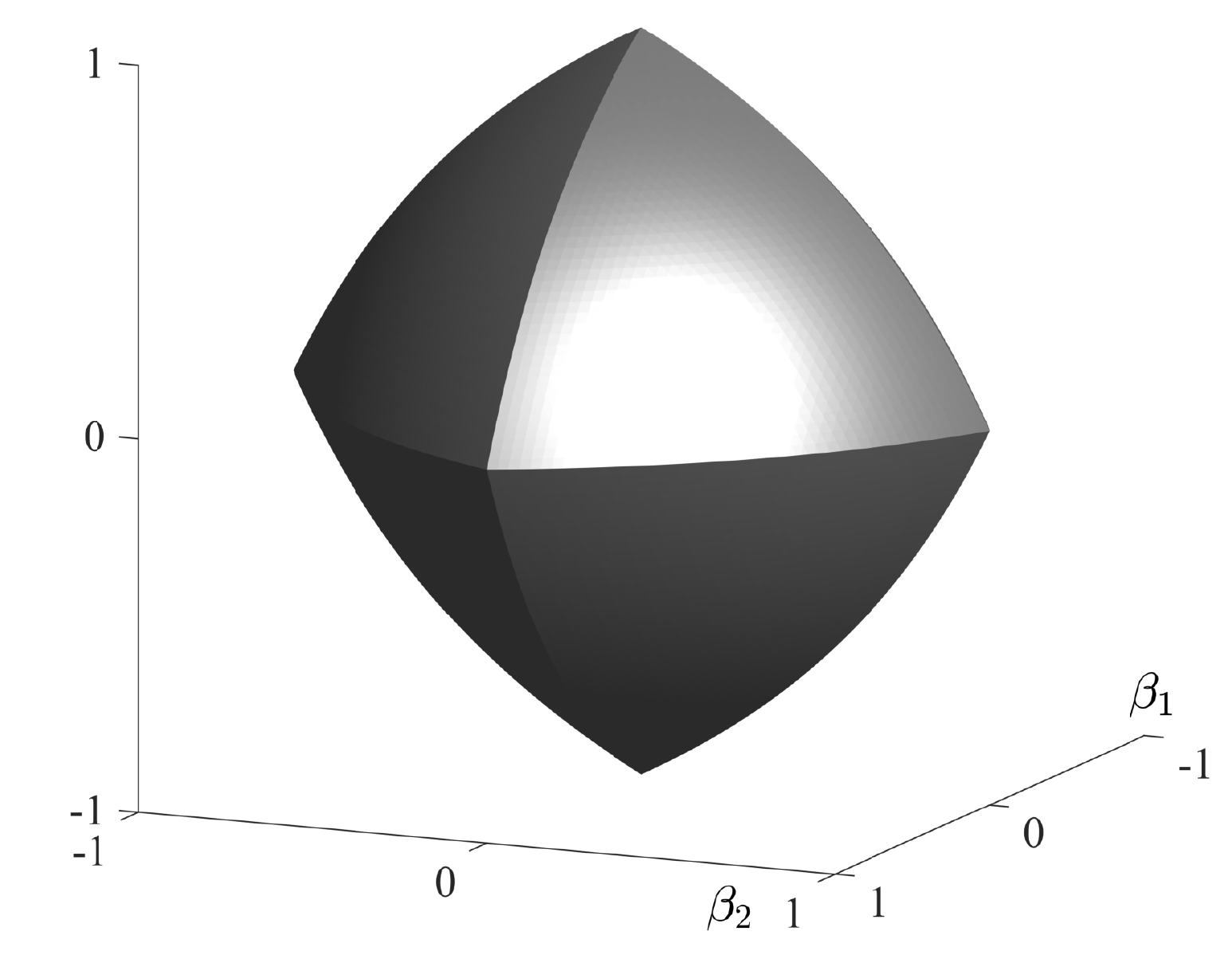}
		\caption{\footnotesize sg-LASSO with 1 group, $\alpha = 0.5$} 
	\end{subfigure}
	\begin{subfigure}{0.33\textwidth}
		\includegraphics[width=\textwidth]{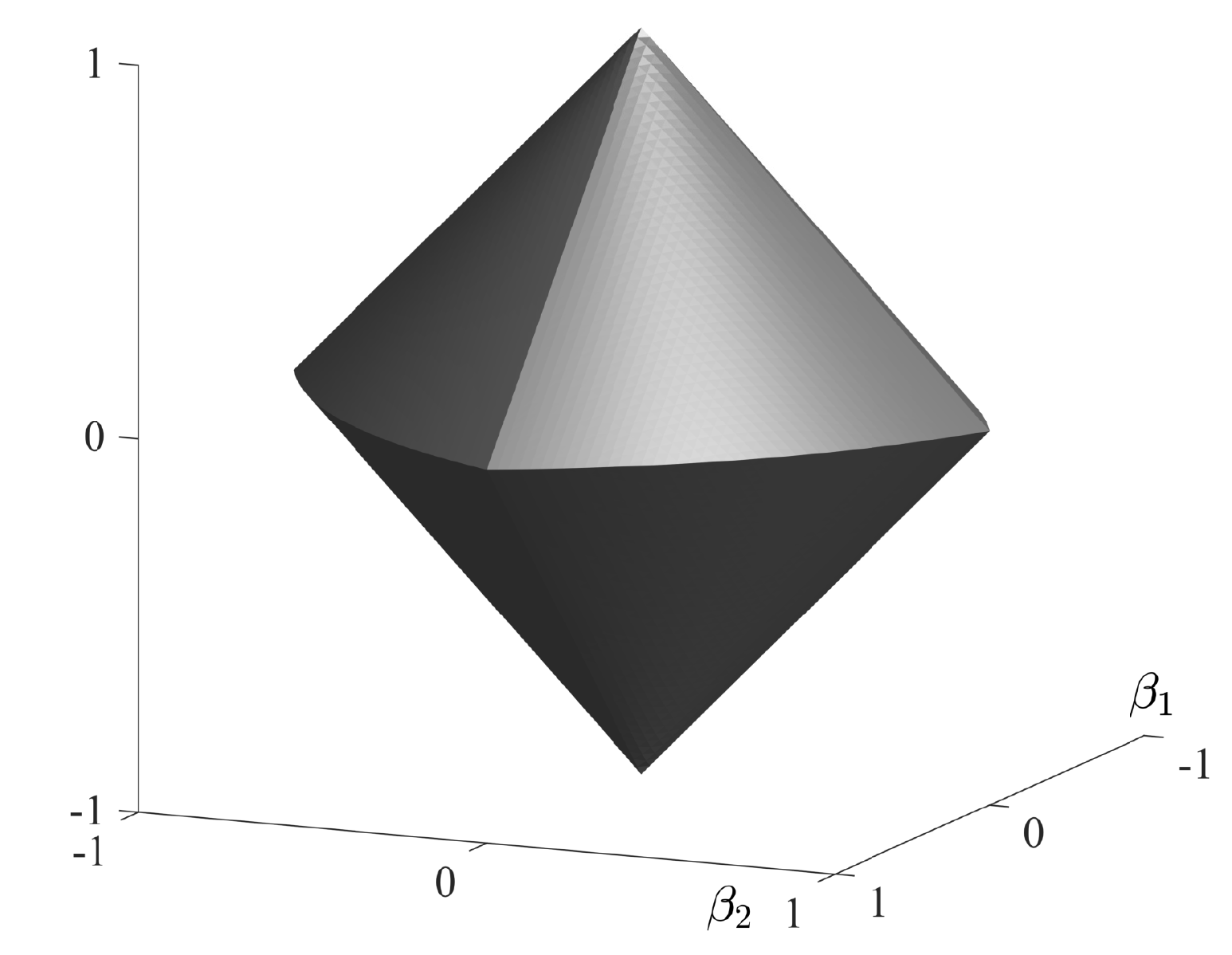}
		\caption{\footnotesize sg-LASSO with 2 groups, $\alpha = 0.5$} 
	\end{subfigure}
	\caption{Geometry of $\{\beta\in\R^2: \Omega(\beta)\leq 1 \}$ for different groupings and values of $\alpha$.} 
	\label{fig:contour_surf}
\end{figure}

To describe the estimation procedure, let
$\mathbf{y}$  = $(y_{1},\dots,y_T)^\top,$ be a vector of dependent variable and let $\mathbf{X}$ = $(\iota, \mathbf{y}_{1},\dots,\mathbf{y}_{J},Z_1W,\dots,Z_KW),$ be a design matrix, where $\iota = (1,1,\dots,1)^\top$ is a vector of ones, $\mathbf{y}_{j} = (y_{1-j},\dots, y_{T-j})^\top$, $Z_k = (x_{k,t-{(j-1)/m}})_{t\in[T],j\in[m]}$ is a $T\times m$ matrix of the covariate $k\in[K]$, and $W=\left(w_l\left((j-1)/m\right)/m\right)_{j\in[m],l\in[L]}$ is an $m\times L$ matrix of weights. In addition, put $\beta$  = $(\beta_0^\top,\beta_1^\top,\dots,\beta_K^\top)^\top$, where $\beta_0 = (\rho_0,\rho_1,\dots,\rho_J)^\top$ is a vector of parameters pertaining to the group consisting of the intercept and the autoregressive coefficients, and $\beta_k\in\R^L$ denotes parameters of the high-frequency lag polynomial pertaining to the covariate $k\geq 1$. Then, the sparse-group LASSO estimator, denoted $\hat\beta$, solves the penalized least-squares problem
\begin{equation}\label{eq:sgl}
	\min_{b\in\R^p}\|\mathbf{y} - \mathbf{X}b\|_T^2 + 2\lambda\Omega(b)
\end{equation}
with a penalty function that interpolates between the $\ell_1$ LASSO penalty and the group LASSO penalty
\begin{equation*}
	\Omega(b) = \alpha|b|_1 + (1-\alpha)\|b\|_{2,1},
\end{equation*}
where $\|b\|_{2,1} = \sum_{G\in\mathcal{G}}|b_G|_2$ is the group LASSO norm and $\mathcal{G}$ is a group structure (partition of $[p]$) specified by the econometrician. Note that estimator in equation~(\ref{eq:sgl}) is defined as a solution to the convex optimization problem and can be computed efficiently, e.g., using an appropriate coordinate descent algorithm; see \cite{simon2013sparse}.

The amount of penalization in equation~(\ref{eq:sgl}) is controlled by the regularization parameter $\lambda>0$ while $\alpha\in[0,1]$ is a weight parameter that determines the relative importance of the sparsity and the group structure. Setting $\alpha=1$, we obtain the LASSO estimator while setting $\alpha = 0$, leads to the group LASSO estimator, which is reminiscent of the elastic net. In Figure~\ref{fig:contour_surf} we illustrate the geometry of the penalty function for different groupings and different values of $\alpha$ covering (a) LASSO with $\alpha$ = 1, (b) group LASSO with one group, $\alpha$ = 0, and two sg-LASSO cases (c) one group and (d) two groups both with $\alpha$ = 0.5. In practice, groups are defined by a particular problem and are specified by the econometrician, while $\alpha$ can be fixed or selected jointly with $\lambda$ in a data-driven way such as using the cross-validation.

\section{High-dimensional time series regressions \label{sec:oracle}}
\subsection{High-dimensional regressions and $\tau$-mixing}
We focus on a generic high-dimensional linear projection model with a countable number of regressors
\begin{equation}\label{eq:hd_ts_model}
	y_t = \sum_{j=0}^\infty x_{t,j}\beta_j + u_t,\qquad\E[u_tx_{t,j}]=0,\quad\forall j\geq 1,\qquad t\in\Z,
\end{equation}
where $x_{t,0}=1$ and $m_t \triangleq \sum_{j=0}^\infty x_{t,j}\beta_j$ is a well-defined random variable. In particular, to ensure that $y_t$ is a well-defined economic quantity, we need $\beta_j\downarrow 0$ sufficiently fast, which is a form of the \textit{approximate sparsity} condition, see \cite{belloni2018high}. This setting nests the high-dimensional ARDL-MIDAS projections described in the previous section and more generally may allow for other high-dimensional time series models. In practice, given a (large) number of covariates, lags  thereof, as well as lags of the dependent variable, denoted $x_t\in\R^p$, we would approximate $m_t$ with $x_t^\top\beta \triangleq \sum_{j=0}^px_{t,j}\beta_j$, where $p<\infty$ and the regression coefficient $\beta\in\R^p$ could be sparse. Importantly, our settings allows for the approximate sparsity as well as other forms of misspecification and the main result of the following section allows for $m_t\ne x_t^\top\beta$.

Using the setting of equation (\ref{eq:sgl}), for a sample $(y_t,x_t)_{t=1}^T$, write
\begin{equation*}
	\mathbf{y} = \mathbf{m} + \mathbf{u},
\end{equation*}
where $\mathbf{y}=(y_1,\dots,y_T)^\top$, $\mathbf{m} = (m_1,\dots, m_T)^\top$, and $\mathbf{u}=(u_1,\dots,u_T)^\top$. The approximation to $\mathbf{m}$ is denoted $\mathbf{X}\beta$, where $\mathbf{X} = (x_1,\dots,x_T)^\top$ is a $T\times p$ matrix of covariates and $\beta=(\beta_1,\dots,\beta_p)^\top$ is a vector of unknown regression coefficients.

We measure the time series dependence with $\tau$-mixing coefficients. For a $\sigma$-algebra $\mathcal{M}$ and a random vector $\xi\in\R^l$, put
\begin{equation*}
	\tau(\mathcal{M},\xi) = \bigg\|\sup_{f\in\mathrm{Lip}_1}\left|\E(f(\xi)|\mathcal{M}) - \E(f(\xi))\right|\bigg\|_1,
\end{equation*}
where $\mathrm{Lip}_1=\left\{f:\R^l\to\R:\;|f(x)-f(y)|\leq |x-y|_1\right\}$ is a set of $1$-Lipschitz functions. Let $(\xi_t)_{t\in\Z}$ be a stochastic process and let $\mathcal{M}_t=\sigma(\xi_t,\xi_{t-1},\dots)$ be its canonical filtration. The \textit{$\tau$-mixing coefficient} of $(\xi_t)_{t\in\Z}$ is defined as
\begin{equation*}
	\tau_k = \sup_{j\geq 1}\frac{1}{j}\sup_{t+k\leq t_1<\dots<t_j}\tau(\mathcal{M}_t,(\xi_{t_1},\dots,\xi_{t_j})),\qquad k\geq 0.
\end{equation*}
If $\tau_k\downarrow0$ as $k\to\infty$, then the process $(\xi_t)_{t\in\Z}$ is called \textit{$\tau$-mixing}. The $\tau$-mixing coefficients were introduced in \cite{dedecker2004coupling} as dependence measures weaker than mixing. Note that the commonly used $\alpha$- and $\beta$-mixing conditions are too restrictive for the linear projection model with an ARDL-MIDAS process. Indeed, a causal linear process with dependent innovations is not necessary $\alpha$-mixing; see also \cite{andrews1984non} for an example of AR(1) process which is not $\alpha$-mixing. Roughly speaking, $\tau$-mixing processes are somewhere between mixingales and $\alpha$-mixing processes and can accommodate such counterexamples. At the same time, sharp Fuk-Nagaev inequalities are available for $\tau$-mixing processes which to the best of our knowledge is not the case for the mixingales or near-epoch dependent processes; see \cite{babiietalinference}.

\cite{dedecker2004coupling,dedecker2005new} discuss how to verify the $\tau$-mixing property for causal Bernoulli shifts with dependent innovations and nonlinear Markov processes. It is also worth comparing the $\tau$-mixing coefficient to other weak dependence coefficients. Suppose that $(\xi_t)_{t\in\Z}$ is a real-valued stationary process and let $\gamma_k = \|\E(\xi_{k}|\mathcal{M}_0) - \E(\xi_{k})\|_1$ be its $L_1$ mixingale coefficient. Then we clearly have $\gamma_k\leq \tau_k$ and it is known that
\begin{equation*}
	|\Cov(\xi_0,\xi_k)| \leq \int_0^{\gamma_k}Q\circ G(u)\dx u \leq \int_0^{\tau_k}Q\circ G(u)\dx u \leq \tau_k^\frac{q-2}{q-1}\|\xi_0\|_q^{q/(q-1)},
\end{equation*}
where $Q$ is the generalized inverse of $x\mapsto \Pr(|\xi_0|>x)$ and $G$ is the generalized inverse of $x\mapsto\int_0^x Q(u)\dx u$; see \cite{babiietalinference}, Lemma A.1.1. Therefore, the $\tau$-mixing coefficient provides a sharp control of autocovariances similarly to the $L_1$ mixingale coefficients, which in turn can be used to ensure that the long-run variance of $(\xi_t)_{t\in\Z}$ exists. The $\tau$-mixing coefficient is also bounded by the $\alpha$-mixing coefficient, denoted $\alpha_k$, as follows
\begin{equation*}
	\tau_k \leq 2\int_0^{2\alpha_k}Q(u)\dx u \leq 2\|\xi_0\|_q(2\alpha_k)^{1/r},
\end{equation*}
where the first inequality follows by \cite{dedecker2004coupling}, Lemma 7 and the second by H\"{o}lder's inequality with $q,r\geq 1$ such that $q^{-1}+r^{-1}=1$. It is worth mentioning that the mixing properties for various time series models in econometrics, including GARCH, stochastic volatility, or autoregressive conditional duration are well-known; see, e.g., \cite{carrasco2002mixing}, \cite{francq2019garch}, \cite{babii2019commercial}; see also \cite{dedecker2007weak} for more examples and a comprehensive comparison of various weak dependence coefficients.

\subsection{Estimation and prediction properties}
In this section, we introduce the main assumptions for the high-dimensional time series regressions and study the estimation and prediction properties of the sg-LASSO estimator covering the LASSO and the group LASSO estimators as special cases. The following assumption imposes some mild restrictions on the stochastic processes in the high-dimensional regression equation~(\ref{eq:hd_ts_model}).
\begin{assumption}[Data]\label{as:data}
	For every $j,k\in[p]$, the processes $(u_tx_{t,j,})_{t\in\Z}$ and $(x_{t,j}x_{t,k})_{t\in\Z}$ are stationary such that (i) $\|u_0\|_{q}<\infty$ and $\max_{j\in[p]}\|x_{0,j}\|_{r}=O(1)$ for some constants $q>2r/(r-2)$ and $r>4$; (ii) the $\tau$-mixing coefficients are $\tau_k \leq ck^{-a}$ and respectively $\tilde\tau_k\leq ck^{-b}$ for all $k\geq 0$ and some $c>0$, $a>(\varsigma-1)/(\varsigma-2)$, $b>(r-2)/(r-4)$, and $\varsigma = qr/(q+r)$.
\end{assumption}
\noindent It is worth mentioning that the stationarity condition is not essential and can be relaxed to the existence of the limiting variance of partial sums at costs of heavier notations and proofs. Condition (i) requires that covariates have at least $4$ finite moments, while the number of moments required for the error process can be as low as $2+\epsilon$, depending on the integrability of covariates. Therefore, (i) may allow for heavy-tailed distributions commonly encountered in financial and economic time series, e.g., asset returns and volatilities. Given the integrability in (i), (ii) requires that the $\tau$-mixing coefficients decrease to zero sufficiently fast; see Online Appendix, Section~\ref{sec:moments_mixing} for moments and $\tau$-mixing coefficients of ARDL-MIDAS. It is known that the $\beta$-mixing coefficients decrease geometrically fast, e.g., for geometrically ergodic Markov chains, in which case (ii) holds for every $a,b>0$. Therefore, (ii) allows for relatively persistent processes.

\smallskip
For the support $S_0$ and the group support $\mathcal{G}_0$ of $\beta$, put
\begin{equation*}
	\Omega_0(b) \triangleq \alpha|b_{S_0}|_1 + (1-\alpha)\sum_{G\in \mathcal{G}_0}|b_{G}|_2\qquad \text{and} \qquad  \Omega_1(b) \triangleq \alpha|b_{S^c_0}|_1 + (1-\alpha)\sum_{G\in \mathcal{G}_0^c}|b_{G}|_2.
\end{equation*}
For some $c_0>0$, define $\mathcal{C}(c_0) \triangleq \left\{\Delta\in\R^p:\; \Omega_1(\Delta)\leq c_0\Omega_0(\Delta) \right\}$. The following assumption generalizes the restricted eigenvalue condition of \cite{bickel2009simultaneous} to the sg-LASSO estimator and is imposed on the population covariance matrix $\Sigma = \E[\mathbf{X}^\top\mathbf{X}/T]$.
\begin{assumption}[Restricted eigenvalue]\label{as:covariance}
	There exists a universal constant $\gamma>0$ such that $\Delta^\top\Sigma\Delta \geq \gamma\sum_{G\in\mathcal{G}_0}|\Delta_{G}|_2^2$ for all $\Delta\in\mathcal{C}(c_0)$, where $c_0 = (c+1)/(c-1)$ for some $c>1$.
\end{assumption}
\noindent Recall that if $\Sigma$ is a positive definite matrix, then for all $\Delta\in\R^p$, we have $\Delta^\top\Sigma\Delta \geq \gamma|\Delta|_2^2$, where $\gamma$ is the smallest eigenvalue of $\Sigma$. Therefore, in this case Assumption~\ref{as:covariance} is trivially satisfied because $|\Delta|_2^2\geq \sum_{G\in\mathcal{G}_0}|\Delta_{G}|_2^2$. The positive definiteness of $\Sigma$ is also known as a completeness condition and Assumption~\ref{as:covariance} can be understood as its weak version; see \cite{babii2017completeness} and references therein. It is worth emphasizing that $\gamma>0$ in Assumption~\ref{as:covariance} is a universal constant independent of $p$, which is the case, e.g., when $\Sigma$ is a Toeplitz matrix or a spiked identity matrix. Alternatively, we could allow for $\gamma\downarrow 0$ as $p\to \infty$, in which case the term $\gamma^{-1}$ would appear in our nonasymptotic bounds slowing down the speed of convergence, and we may interpret $\gamma$ as a measure of ill-posedness in the spirit of econometrics literature on ill-posed inverse problems; see \cite{carrasco2007linear}.

The value of the regularization parameter is determined by the Fuk-Nagaev concentration inequality, appearing in the Online Appendix, see Theorem~\ref{appcor:fn_inequality}.

\begin{assumption}[Regularization]\label{as:tuning}
	For some $\delta\in(0,1)$
	\begin{equation*}
		\lambda \sim\left(\frac{p}{\delta T^{\kappa-1}}\right)^{1/\kappa}\vee\sqrt{\frac{\log(8p/\delta)}{T}},
	\end{equation*}
	where $\kappa = ((a+1)\varsigma-1)/(a+\varsigma-1)$ and $a,\varsigma$ are as in Assumption~\ref{as:data}.
\end{assumption}
\noindent The regularization parameter in Assumption~\ref{as:tuning} is determined by the persistence of the data, quantified by $a$, and the tails, quantified by $\varsigma = qr/(q+r)$. This dependence is reflected in the \textit{dependence-tails exponent} $\kappa$. The following result describes the nonasymptotic prediction and estimation bounds for the sg-LASSO estimator, see Online Appendix~\ref{appsec:proofs} for the proof.

\begin{theorem}\label{thm:estimation_bound}
	Suppose that Assumptions~\ref{as:data}, \ref{as:covariance}, and \ref{as:tuning} are satisfied. Then with probability at least $1 - \delta - O(p^2(T^{1-\mu}s_\alpha^\mu + \exp(-cT/s_\alpha^2)))$
	\begin{equation*}
	\|\mathbf{X}(\hat\beta - \beta)\|^2_T \lesssim s_\alpha\lambda^2 + \|\mathbf{m} - \mathbf{X}\beta\|_T^2
	\end{equation*}
	and
	\begin{equation*}
	\Omega(\hat\beta - \beta) \lesssim s_\alpha\lambda + \lambda^{-1}\|\mathbf{m} - \mathbf{X}\beta\|_T^2 + \sqrt{s_\alpha}\|\mathbf{m} - \mathbf{X}\beta\|_T
	\end{equation*}
	for some $c>0$, where $\sqrt{s_\alpha} = \alpha\sqrt{|S_0|} + (1-\alpha)\sqrt{|\mathcal{G}_0|}$ and $\mu = ((b + 1)r - 2)/(r+2(b-1))$.
\end{theorem}
Theorem~\ref{thm:estimation_bound} provides nonasymptotic guarantees for the estimation and prediction with the sg-LASSO estimator reflecting potential misspecification. In the special case of the LASSO estimator ($\alpha=1$), we obtain the counterpart to the result of \cite{belloni2012sparse} for the LASSO estimator with i.i.d.\ data taking into account that we may have $m_t\ne x_t^\top\beta$. At another extreme, when $\alpha=0$, we obtain the nonasymptotic bounds for the group LASSO allowing for misspecification which to the best of our knowledge are new, cf. \cite{negahban2012unified} and \cite{van2016estimation}. We call $s_\alpha$ the \textit{effective sparsity constant}. This constant reflects the benefits of the sparse-group structure for the sg-LASSO estimator that can not be deduced from the results currently available for the LASSO or the group LASSO.

\begin{remark}
	Since the $\ell_1$-norm is equivalent to the $\Omega$-norm whenever groups have fixed size, we deduce from Theorem~\ref{thm:estimation_bound} that
	\begin{equation*}
		|\hat\beta - \beta|_1 \lesssim s_\alpha\lambda + \lambda^{-1}\|\mathbf{m} - \mathbf{X}\beta\|_T^2 + \sqrt{s_\alpha}\|\mathbf{m} - \mathbf{X}\beta\|_T.
	\end{equation*}
\end{remark}

\smallskip

\noindent Next, we consider the asymptotic regime, in which the misspecification error vanishes when the sample size increases as described in the following assumption.
\begin{assumption}\label{as:rates}
	(i) $\|\mathbf{m} - \mathbf{X}\beta\|_T^2 = O_P\left(s_\alpha\lambda^2\right)$; and (ii) $p^2T^{1-\mu}s_\alpha^{\mu}\to 0$ and $p^2\exp(-cT/s_\alpha^2)\to 0$.
\end{assumption}

\noindent The following corollary is an immediate consequence of Theorem~\ref{thm:estimation_bound}.
\begin{corollary}\label{cor:rates_iid}
	Suppose that Assumptions~\ref{as:data}, \ref{as:covariance}, \ref{as:tuning}, and \ref{as:rates} hold. Then
	\begin{equation*}
	\|\mathbf{X}(\hat\beta - \beta)\|_T^2 = O_P\left(\frac{s_\alpha p^{2/\kappa}}{T^{2-2/\kappa}}\vee\frac{s_\alpha\log p}{T}\right)
	\end{equation*}
	and
	\begin{equation*}
		|\hat{\beta} - \beta|_1 = O_P\left(\frac{s_\alpha p^{1/\kappa}}{T^{1-1/\kappa}}\vee s_\alpha\sqrt{\frac{\log p}{T}}\right).
	\end{equation*}
\end{corollary}
\noindent If the effective sparsity constant $s_\alpha$ is fixed, then $p=o(T^{\kappa - 1})$ is a sufficient condition for the prediction and estimation errors to vanish, whenever $\mu\geq 2\kappa-1$. In this case Assumption~\ref{as:rates} (ii) is vacuous. More generally, $s_\alpha$ is allowed to increase slowly with the sample size. Convergence rates in Corollary~\ref{cor:rates_iid} quantify the effect of tails and persistence of the data on the prediction and estimation accuracies of the sg-LASSO estimator. In particular, lighter tails and less persistence allow us to handle a larger number of covariates $p$ compared to the sample size $T$. In particular $p$ can increase faster than $T$, provided that $\kappa>2$.

\begin{remark}
	In the special case of the LASSO estimator with i.i.d.\ data, Corollary 4 of \cite{fuk1971probability} leads to the convergence rate of order $O_P\left(\frac{p^{1/\varsigma}}{T^{1-1/\varsigma}}\vee \sqrt{\frac{\log p}{T}}\right)$. If the $\tau$-mixing coefficients decrease geometrically fast (e.g., stationary AR(p)), then $\kappa\approx\varsigma$ for a sufficiently large value of the dependence exponent $a$, in which case the convergence rates in Corollary~\ref{cor:rates_iid} are close to the i.i.d.\ case. In this sense these rates depend sharply on the tails exponent $\varsigma$, and we can conclude that for geometrically decreasing $\tau$-mixing coefficients, the persistence of the data should not affect the convergence rates of the LASSO.
\end{remark}

\begin{remark}\label{remark_chernozhukov}
	In the special case of the LASSO estimator, if $(u_t)_{t\in\Z}$ and $(x_t)_{t\in\Z}$ are causal Bernoulli shifts with independent innovations and at least $q=r\geq8$ finite moments, one can deduce from \cite{chernozhukov2019lasso}, Lemma 5.1 and Corollary 5.1, the convergence rate of order $O_P\left(\frac{(p\omega_T)^{1/\varsigma}}{T^{1-1/\varsigma}}\vee \sqrt{\frac{\log p}{T}}\right)$, where $\omega_T=1$ (weakly dependent case) or $\omega_T = T^{\varsigma/2-1-a\varsigma}\uparrow\infty$ (strongly dependent case), provided that the physical dependence coefficients are of size $O(k^{-a})$. Note that for causal Bernoulli shifts with independent innovations, the physical dependence coefficients are not directly comparable to $\tau$-mixing coefficients; see \cite{dedecker2007weak}, Remark 3.1 on p.32.
\end{remark}

\section{Monte Carlo experiments \label{sec:mc}}
We assess via simulations the out-of-sample predictive performance (forecasting and nowcasting), and the MIDAS weights recovery of the sg-LASSO with dictionaries. We benchmark the performance of our novel sg-LASSO setup against two alternatives: (a) unstructured, meaning standard, LASSO with MIDAS, and (b) unstructured LASSO with the unrestricted lag polynomial. The former allows us to assess the benefits of exploiting group structures, whereas the latter focuses on the advantages of using dictionaries in a high-dimensional setting.
\smallskip

\subsection{Simulation Design \label{appsec:mcdesign}}

To assess the predictive performance and the MIDAS weight recovery, we simulate the data from the following DGP:
\begin{equation*}
y_t =  \rho_1y_{t-1} + \rho_2y_{t-2} + \sum_{k=1}^K\frac{1}{m}\sum_{j=1}^m \omega((j-1)/m;\beta_k)x_{t-(j-1)/m,k} + u_t,
\end{equation*} 
where $u_t\sim_{i.i.d.}N(0,\sigma_u^2)$ and the DGP for covariates $\{x_{k,t-(j-1)/m}:j\in[m],k\in[K]\}$ is specified below. This corresponds to a target of interest $y_t$ driven by two autoregressive lags augmented with high frequency series, hence, the DGP is an ARDL-MIDAS model. We set $\sigma^2_u=1$, $\rho_1=0.3$, $\rho_2=0.01$, and take the number of relevant high frequency regressors $K$ = 3. In some scenarios we also decrease the signal-to-noise ratio by setting $\sigma^2_u$ = 5. We are interested in quarterly/monthly data, and use four quarters of data for the high frequency regressors so that $m$ = 12. We rely on a commonly used weighting scheme in the MIDAS literature, namely $\omega(s;\beta_k)$ for $k$ = 1, 2 and 3 are determined by beta densities respectively equal to $\mathrm{Beta}(1,3),\mathrm{Beta}(2,3)$, and $\mathrm{Beta}(2,2)$; see \cite{ghysels2007midas} or \cite{ghysels2019estimating}, for further details. 

\smallskip

\noindent The high frequency regressors are generated as either one of the following:
\begin{enumerate}
	\item  $K$ i.i.d.\ realizations of the univariate autoregressive (AR) process $x_h = \rho x_{h-1}+ \varepsilon_h,$ where $\rho=0.2$ or $\rho=0.7$ and either $\varepsilon_h\sim_{i.i.d.}N(0,\sigma_\varepsilon^2)$,  $\sigma_\varepsilon^2=1$, or $\varepsilon_h\sim_{i.i.d.}\text{student-}t(5)$, where $h$ denotes the high-frequency sampling.
	\item Multivariate vector autoregressive (VAR) process $X_h =  \Phi X_{h-1} + \mathbf{\varepsilon}_h,$ where $\varepsilon_{h}\sim_{i.i.d.} N(0,I_K)$ and $\Phi$ is a block diagonal matrix described below.
\end{enumerate}
For the AR simulation design, we initiate the processes as $x_0\sim N\left(0,\sigma^2/(1-\rho^2)\right)$ and $y_0\sim N\left(0,\sigma^2(1-\rho_2)/((1+\rho_2)((1-\rho_2)^2-\rho_1^2))\right).$
For the VAR, the initial value of $(y_t)$ is the same, while $X_0 \sim N(0,I_K)$. In all cases, the first 200 observations are treated as burn-in. In the estimation procedure, we add 7 noisy covariates which are generated in the same way as the relevant covariates and use 5 low-frequency lags. The empirical models use a dictionary which consists of Legendre polynomials up to degree $L = 10$ shifted to the $[0,1]$ interval with the MIDAS weight function approximated as in equation (\ref{eq:dicleg}).
The sample size is $T\in\{50, 100, 200\},$ and for all the experiments we use 5000 simulation replications. 

\smallskip

We assess the performance of different methods by modifying the assumptions on the error terms of the high-frequency process $\mathbf{\varepsilon}_h$, considering multivariate high-frequency processes, changing the degree of Legendre polynomials $L$, increasing the noise level of the low-frequency process $\sigma^2_u$, using only half of the high-frequency lags in predictive regressions, and adding a larger number of noisy covariates. In the case of VAR high-frequency process, we set $\Phi$ to be block-diagonal with the first $5\times 5$ block having entries $0.15$ and the remaining $5\times 5$ block(s) having entries $0.075$.

\smallskip

We estimate three different LASSO-type regression models. In the first model, we keep the weighting function unconstrained, and therefore we estimate 12 coefficients per high-frequency covariate using the unstructured LASSO estimator. We denote this model LASSO-U-MIDAS (inspired by the U-MIDAS of \cite{foroni2015unrestricted}). In the second model we use MIDAS weights together with the unstructured LASSO estimator; we call this model LASSO-MIDAS. In this case, we estimate $L+1$ number of coefficients per high-frequency covariate. The third model applies the sg-LASSO estimator together with MIDAS weights. Groups are defined as in Section \ref{sec:midas}; each low-frequency lag and high-frequency covariate is a group, therefore, we have $K+5$ groups. We select the value of tuning parameters $\lambda$ and $\alpha$ using the 5-fold cross-validation, defining folds as adjacent blocks over the time dimension to take into account the time series dependence. This model is denoted sg-LASSO-MIDAS. 

\smallskip

For regressions with aggregated data, we consider:
(a) Flow aggregation (FLOW): $x_{k,t}^A$ = $1/m\sum_{j=1}^mx_{k,t-(j-1)/m}$, (b) Stock aggregation (STOCK): $x_{k,t}^A$ = $x_{k,t}$, and (c) Middle high-frequency lag (MIDDLE): single middle value of the high-frequency lag with ties solved in favor of the most recent observation (i.e., we take a single $6^{\rm th}$ lag if $m=12$). In these cases, the models are estimated using the OLS estimator, which is unfeasible when the number of covariates becomes equal to the sample size and we leave results blank in this case.

\subsection{Simulation results}

Detailed results are reported in the Online Appendix.
Tables \ref{appendix:tab:forecast}--\ref{appendix:tab:nowcast}, cover the average mean squared forecast errors for one-step-ahead forecasts and nowcasts. The sg-LASSO with MIDAS weighting (sg-LASSO-MIDAS) outperforms all other methods in all simulation scenarios. Importantly, both sg-LASSO-MIDAS and unstructured LASSO-MIDAS with nonlinear weight function approximations perform much better than all other methods when the sample size is small ($T=50$). In this case, sg-LASSO-MIDAS yields the largest improvements over alternatives, in particular, with a large number of noisy covariates (bottom-right block). These findings are robust to increases in the persistence parameter of covariates $\rho$ from 0.2 to 0.7. The LASSO without MIDAS weighting has typically large forecast errors. Comparing across simulation scenarios, all methods seem to perform worse with heavy-tailed or persistent covariates. In these cases, however, the impact on the sg-LASSO-MIDAS method is lesser compared to the other methods. This simulation evidence supports our theoretical results and findings in the empirical application. Lastly, forecasts using flow-aggregated covariates seem to perform better than other simple aggregation methods in all simulation scenarios, but significantly worse than the sg-LASSO-MIDAS.

In Table \ref{appendix:tab:shape_rec1}--\ref{appendix:tab:shape_rec2} we report additional results for the estimation accuracy of the weight functions. In Figure \ref{appendix:fig:weights_beta_1}--\ref{appendix:fig:weights_beta_3}, we plot the estimated weight functions from several methods. The results indicate that the LASSO without MIDAS weighting can not accurately recover the weights in small samples and/or low signal-to-noise ratio scenarios. Using Legendre polynomials improves the performance substantially and the sg-LASSO seems to improve even more over the unstructured LASSO.

\section{Nowcasting US GDP with macro, financial and textual news data \label{sec:empirical}}
We nowcast US GDP with macroeconomic, financial, and textual news data.  Details regarding the data sources appear in the Online Appendix Section \ref{appendix:detailed_description}. Regarding the macro data, we rely on 34 series used in the Federal Reserve Bank of New York nowcast model, discarding two series ("PPI: Final demand" and "Merchant wholesalers: Inventories") due to very short samples; see \cite{bok2018macroeconomic} for more details regarding this data.

For all macro data, we use real-time vintages, which effectively means that we take all macro series with a delay. For example, if we nowcast the first quarter of GDP one month before the quarter ends, we use data up to the end of February, and therefore all macro series with a delay of one month that enter the model are available up to the end of January. We use Legendre polynomials of degree three for all macro covariates to aggregate twelve lags of monthly macro data. In particular, let $x_{t+(h+1-j)/m,k}$ be $k^{{\text{th}}}$ covariate at quarter $t$ with $m=3$ months per quarter and $h=2-1=1$ months into the quarter (2 months into the quarter minus 1 month due to publication delay), where $j = 1,2,\dots,12$ is the monthly lag. We then collect all lags in a vector
\begin{equation*}
	X_{t,k} = (x_{t+1/3,k}, x_{t+0/3,k},\dots,x_{t-10/3,k})^\top
\end{equation*}
and aggregate $X_{t,k}$ using a dictionary $W$ consisting of Legendre polynomials, so that $X_{t,k}W$ defines as a single group for the sg-LASSO estimator.

In addition to macro and financial data, we also use the textual analysis data. We take 76 news attention series from \cite{bybee2019structure} and use Legendre polynomials of degree two to aggregate three monthly lags of each news attention series. Note that the news attention series are used without a publication delay, that is, for the one-month horizon, we take the series up to the end of the second month. Moreover, the \cite{bybee2019structure} news topic models involve rolling samples, avoiding look ahead biases when used in our nowcasts.\label{newsrealtime}

\smallskip

We compute the predictions using a rolling window scheme. The first nowcast is for 2002 Q1, for which we use fifteen years (sixty quarters) of data, and the prediction is computed using 2002 January (2-month horizon) February (1-month), and March (end of the quarter) data. We calculate predictions until the sample is exhausted, which is 2017 Q2, the last date for which news attention data is available. As indicated above, we report results for the 2-month, 1-month, and the end-of-quarter horizons. Our target variable is the first release, i.e., the advance estimate of real GDP growth. We tune sg-LASSO-MIDAS regularization parameters $\lambda$ and $\alpha$ using 5-fold cross-validation, defining folds as adjacent blocks over the time dimension to take into account the time series nature of the data. Finally, we follow the literature on nowcasting real GDP and define our target variable to be the annualized growth rate.

\smallskip

Let $x_{t,k}$ be the $k$-th high-frequency covariate at time \(t\). The general ARDL-MIDAS predictive regression is
\begin{equation*}
\phi(L)y_{t+1} = \mu + \sum_{k=1}^K\psi(L^{1/m}; \beta_k)x_{t,k} + u_{t+1},\qquad t=1,\dots,T,
\end{equation*}
where \(\phi(L)\) is the low-frequency lag polynomial, \(\mu\) is the regression intercept, and $\psi(L^{1/m};\beta_k)x_{tk},k=1,\dots,K$ are lags of high-frequency covariates. Following Section \ref{sec:midas}, the high-frequency lag polynomial is defined as
\begin{equation*}
	\psi(L^{1/m};\beta_k)x_{t,k} = \frac{1}{mq_k}\sum_{j=1}^{mq_k}\omega((j-1)/mq_k;\beta_k)x_{t+(h_k+1-j)/m,k},
\end{equation*}
where for $k^{\rm th}$ covariate, \(h_k\) indicates the number of leading months of available data in the quarter \(t\), $q_k$ is the number of quarters of covariate lags, and we approximate the weight function $\omega$ with the Legendre polynomial. For example, if $h_k=1$ and $q_k=4$, then we have $1$ month of data into a quarter and use $q_km=12$ monthly lags for a covariate $k$.

\smallskip

We benchmark our predictions against the simple AR(1) model, which is considered to be a reasonable starting point for short-term GDP growth predictions. We focus on predictions of our method, sg-LASSO-MIDAS, with and without financial data combined with series based on the textual analysis. One natural comparison is with the publicly available Federal Reserve Bank of New York, denoted NY Fed, model implied nowcasts. 
\begin{table}[h]
	{\small
	\centering
	\begin{tabular}{r |ccc }
		& Rel-RMSE & DM-stat-1 & DM-stat-2 \\ 		
		&\multicolumn{3}{c}{2-month horizon}\\ 
		AR(1)                    & 2.056       & 0.612 & 2.985   \\ 
		sg-LASSO-MIDAS & 0.739       &-2.481 &            \\ 
		NY Fed                 & 0.946       &           & 2.481  \\
		&\multicolumn{3}{c}{1-month horizon}\\ 
		AR(1)                    & 2.056       & 2.025 & 2.556 \\ 
		sg-LASSO-MIDAS & 0.725       &-0.818 &           \\ 
		NY Fed                 & 0.805       &           & 0.818  \\			 							
		&\multicolumn{3}{c}{End-of-quarter}\\
		AR(1)                    & 2.056       & 2.992 & 3.000 \\ 
		sg-LASSO-MIDAS & 0.701        &-0.077 &  		  \\ 
		NY Fed                 & 0.708       &            &0.077  \\			
		&\multicolumn{3}{c}{p-value of aSPA test}\\
		&&&0.046\\		 											 
		\hline
	\end{tabular}
	\caption{\footnotesize Nowcast comparisons for models with macro and survey data only -- Nowcast horizons are 2- and 1-month ahead, as well as the end of the quarter. Column {\it Rel-RMSE} reports root mean squared forecasts error relative to the AR(1) model. Column {\it DM-stat-1} reports \cite{diebold1995comparing} test statistic of all models relative to NY Fed nowcasts, while column {\it DM-stat-2} reports the Diebold Mariano test statistic relative to sg-LASSO-MIDAS model. The last row reports the p-value of the average Superior Predictive Ability (aSPA) test, see \cite{quaedvlieg2019multi}, over the three horizons of sg-LASSO-MIDAS model compared to the NY Fed nowcasts. Out-of-sample period: 2002 Q1 to 2017 Q2. 	\label{tab:nowcasts_nyfed_data}}} 
\end{table}
We adopt the following strategy. First, we focus on the same series that are used to calculate the NY Fed nowcasts. The purpose here is to compare {\it models} since the data inputs are the same. This means that we compare the performance of dynamic factor models (NY Fed) with that of machine learning regularized regression methods (sg-LASSO-MIDAS). Next, we expand the data set to see whether additional financial and textual news series can improve the nowcast performance.

\smallskip

In Table \ref{tab:nowcasts_nyfed_data}, we report results based on real-time macro data used for the NY Fed model, see \cite{bok2018macroeconomic}. The results show that the sg-LASSO-MIDAS performs much better than the NY Fed nowcasts at the longer, i.e.\ 2-month, horizon. Our method significantly beats the benchmark AR(1) model for all the horizons, and the accuracy of the nowcasts improve with the horizon. Our end-of-quarter and 1-month horizon nowcasts are similar to the NY Fed ones, with the sg-LASSO-MIDAS being slightly better numerically but not statistically. We also report the average Superior Predictive Ability test of \cite{quaedvlieg2019multi} over all three horizons and the result reveals that the improvement of the sg-LASSO-MIDAS model versus the NY Fed nowcasts is significant at the 5\% significance level.
\begin{table}
	{\small
	\centering
	\begin{tabular}{r |ccc  cc }
				& Rel-RMSE & DM-stat-1 & DM-stat-2 & DM-stat-3 &  DM-stat-4 \\ 		
		&\multicolumn{5}{c}{2-month horizon}\\ 
		PCA-OLS 			          & 0.982      & 0.416 & 2.772 & 0.350 & 2.978 \\
		Ridge-U-MIDAS 		      &	0.918       & -0.188 & 1.073 & -1.593 & 0.281 \\
		LASSO-U-MIDAS 	        & 0.996	     & 0.275 & 1.280 & -1.983 & -0.294 \\
		Elastic Net-U-MIDAS     & 0.907      & -0.266 & 0.976 & -1.725 & 0.042 \\
		sg-LASSO-MIDAS          & 0.779      &     -2.038        &  & -2.349 &  \\
		NY Fed                          & 0.946      &   & 2.038 &  & 2.349 \\
		&\multicolumn{5}{c}{1-month horizon}\\ 
		PCA-OLS 			          &	1.028      & 2.296 & 3.668 & 2.010 & 3.399 \\
		Ridge-U-MIDAS 			  &	0.940	   & 0.927 & 2.063 & -0.184 & 1.979 \\
		LASSO-U-MIDAS 			& 1.044	      & 1.286 & 1.996 & -0.397 & 1.498 \\
		Elastic Net-U-MIDAS 	& 0.990		 & 1.341 & 2.508 & 0.444 & 2.859 \\
		sg-LASSO-MIDAS          & 0.672      & 		-1.426	 &  & -1.341 &  \\
		NY Fed                          &  0.805     & & 1.426  &  & 1.341 \\ 
		SPF (median) 				& 0.639      &  -2.317& -0.490 &  -1.743 & 0.282 \\		
		&\multicolumn{5}{c}{End-of-quarter}\\
		PCA-OLS 			          &	0.988	  & 3.414 & 3.400 & 3.113 & 3.155 \\
		Ridge-U-MIDAS 			  &	0.939     & 1.918 & 1.952 & 0.867 & 1.200 \\
		LASSO-U-MIDAS 			& 1.014	     & 1.790 & 1.773 & 0.276 & 0.517 \\
		Elastic Net-U-MIDAS     & 0.947		& 2.045 & 2.034 & 1.198 & 1.400 \\
		sg-LASSO-MIDAS          & 0.696    &     -0.156       &  & -0.159 &  \\
		NY Fed                          & 0.707     &  & 0.156 &  & 0.159 \\	
		&\multicolumn{5}{c}{p-value of aSPA test}\\
		&&&0.042&&0.056\\		 											 										 
		\hline
	\end{tabular}
	\caption{\footnotesize Nowcast comparison table -- Nowcast horizons are 2- and 1-month ahead, as well as the end of the quarter. Column {\it Rel-RMSE} reports root mean squared forecasts error relative to the AR(1) model. Column {\it DM-stat-1} reports \cite{diebold1995comparing} test statistic of all models relative to the NY FED nowcast, while column {\it DM-stat-2} reports the Diebold Mariano test statistic relative to the sg-LASSO model. Columns {\it DM-stat-3} and {\it DM-stat-4} report the Diebold Mariano test statistic for the same models, but excludes the recession period. For the 1-month horizon, the last row {\it SPF (median)} reports test statistics for the same models comparing with the SPF median nowcasts. The last row reports the p-value of the average Superior Predictive Ability (aSPA) test, see \cite{quaedvlieg2019multi}, over the three horizons of sg-LASSO-MIDAS model compared to the NY Fed nowcasts, including (left) and excluding (right) financial crisis period. Out-of-sample period: 2002 Q1 to 2017 Q2. 	\label{tab:nowcasts_hd_data}}}  
\end{table}

The comparison in Table \ref{tab:nowcasts_nyfed_data} does not fully exploit the potential of our methods, as it is easy to expand the data series beyond the small number used by the NY Fed nowcasting model. In Table \ref{tab:nowcasts_hd_data} we report results with additional sets of covariates which are financial series, advocated by \cite{andreou2013should}, and textual analysis of news. In total, the models select from 118 series -- 34 macro, 8 financial, and 76 news attention series.  For the moment we focus only on the first three columns of the table. At the longer horizon of 2 months, the method seems to produce slightly worse nowcasts compared to the results reported in Table \ref{tab:nowcasts_nyfed_data} using only macro data. However, we find significant improvements in prediction quality for the shorter 1-month and end-of-quarter horizons. In particular, a significant increase in accuracy relative to NY Fed nowcasts appears at the 1-month horizon.  We report again the average Superior Predictive Ability test of \cite{quaedvlieg2019multi} over all three horizons with the same result that the improvement of sg-LASSO-MIDAS versus the NY Fed nowcasts is significant at the 5\% significance level.  Lastly, we report results for several alternatives, namely, PCA-OLS, ridge, LASSO, and Elastic Net, using the unrestricted MIDAS scheme. Our approach produces more accurate nowcasts compared to these alternatives.
	
\smallskip

Table \ref{tab:nowcasts_hd_data} also features an entry called SPF (median), where we report results for the median survey of professional nowcasts for the 1-month horizon, and analyze how the model-based nowcasts compare with the predictions using the publicly available Survey of Professional Forecasters maintained by the Federal Reserve Bank of Philadelphia. We find that the sg-LASSO-MIDAS model-based nowcasts are similar to the SPF-implied nowcasts. We also find that the NY Fed nowcasts are significantly worse than the SPF.

\smallskip

\begin{figure}[!htbp]
	\centering\includegraphics[width=10cm]{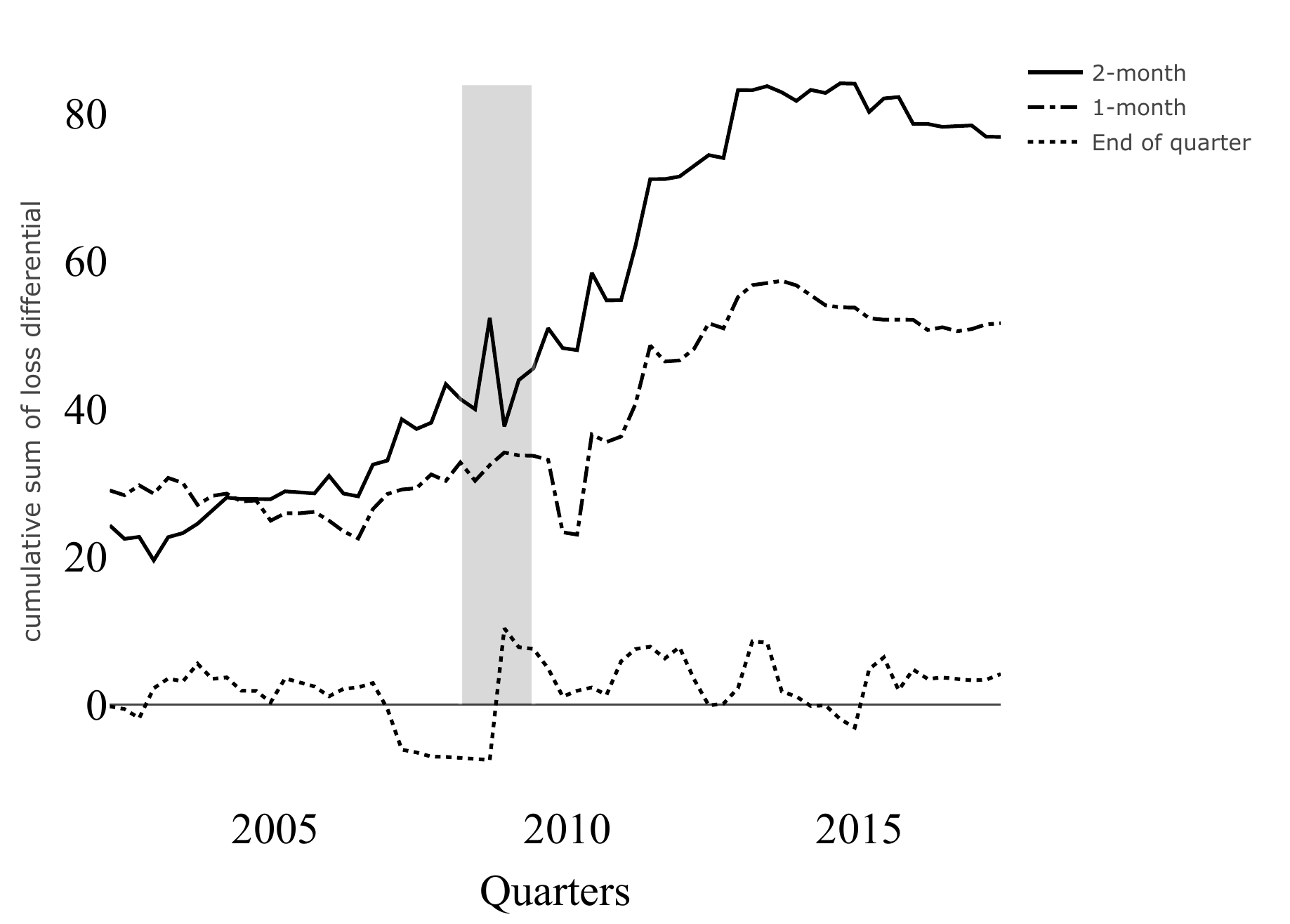}
	\caption{\footnotesize Cumulative sum of loss differentials of sg-LASSO-MIDAS model nowcasts including financial and textual data compared with the New York Fed model for three nowcasting horizons: solid black line cumsfe for the 2-months horizon, dash-dotted black line - cumsfe for the 1-month horizon, and dotted line for the end-of-quarter nowcasts. The gray shaded area corresponds to the NBER recession period. 		\label{fig:cumsfes}}
\end{figure}

In Figure \ref{fig:cumsfes} we plot the cumulative sum of squared forecast error (CUMSFE) loss differential of sg-LASSO-MIDAS versus NY Fed nowcasts for the three horizons. The CUMSFE is computed as
\begin{equation*}
	\text{CUMSFE}_{t,t+k} = \sum_{q=t}^{t+k} e_{q,M1}^2 - e_{q,M2}^2
\end{equation*}
for model $M1$ versus $M2.$ A positive value of $\text{CUMSFE}_{t,t+k}$ means that the model $M1$ has larger squared forecast errors compared to model $M2$ up to $t+k,$ and negative values imply the opposite. In our case, $M1$ is the New York Fed prediction error, while $M2$ is the sg-LASSO-MIDAS model. We observe persistent gains for the 2- and 1-month horizons throughout the out-of-sample period. When comparing the sg-LASSO-MIDAS results with additional financial and textual news series  versus those based on macro data only, we see a notable improvement at the 1-month horizon and a more modest one at the end-of-quarter horizons. In Figure \ref{fig:cumsfes2}, we plot the average CUMSFE for the 1-month and end-of-quarter horizons and observe that the largest gains of additional financial and textual news data are achieved during the financial crisis. 

The result in Figure \ref{fig:cumsfes2} prompts the question whether our results are mostly driven by this unusual period in our out-of-sample data. To assess this, we turn our attention again to the last two columns of Table \ref{tab:nowcasts_hd_data} reporting \cite{diebold1995comparing} test statistics which exclude the financial crisis period. Compared to the tests previously discussed, we find that the results largely remain the same, but some alternatives seem to slightly improve (e.g.\ LASSO or Elastic Net). Note that this also implies that our method performs better during periods with heavy-tailed observations, such as the financial crisis. It should also be noted that overall there is a slight deterioration of the average Superior Predictive Ability test over all three horizons when we remove the financial crisis. 

\smallskip

\begin{figure}[!htbp]
	\centering\includegraphics[width=10cm]{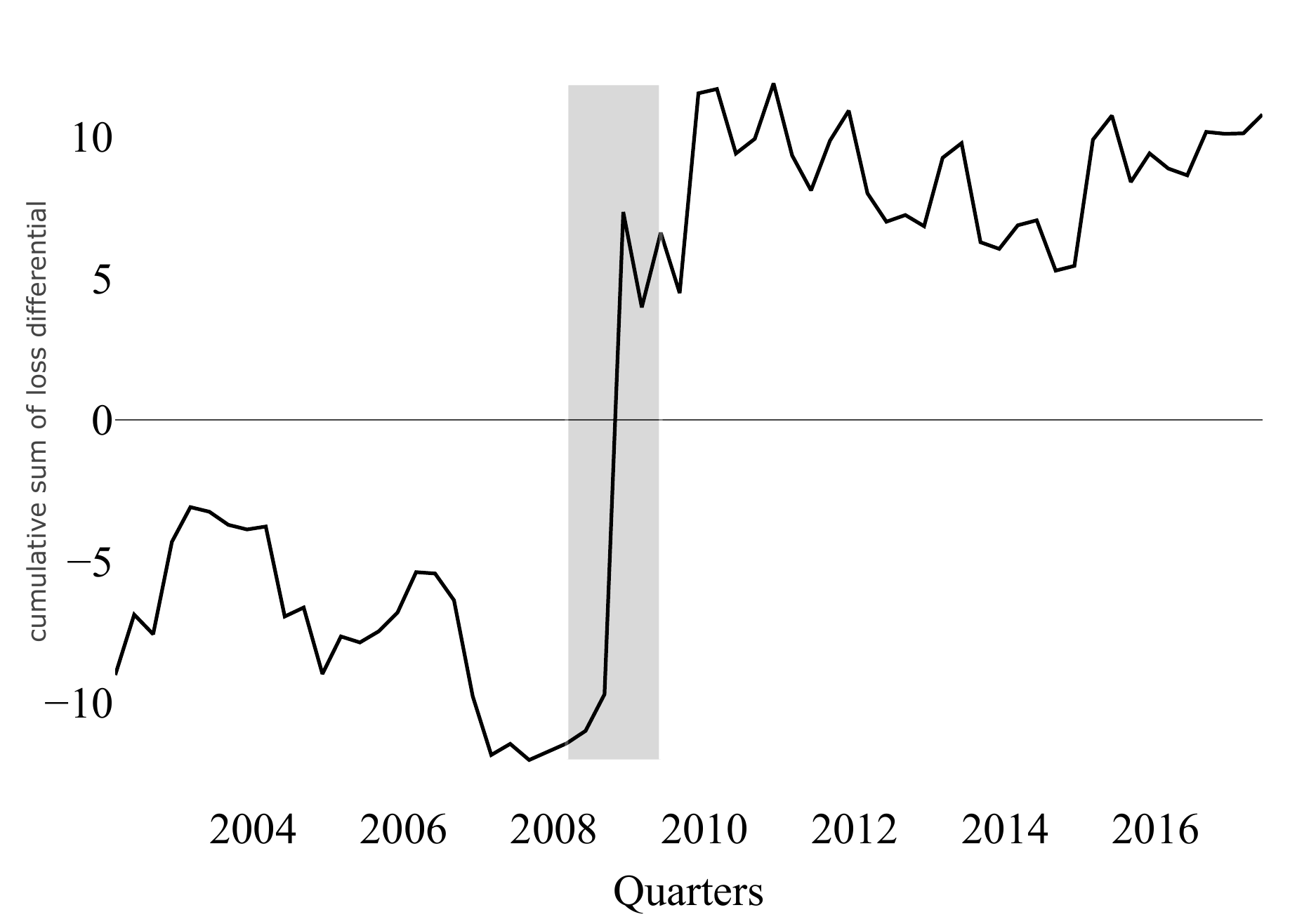}
	\caption{\footnotesize Cumulative sum of loss differentials (CUMSFE) of sg-LASSO-MIDAS nowcasts when we include vs. when we exclude the additional financial and textual news data, averaged over 1-month and the end-of-quarter horizons. The gray shaded area corresponds to the NBER recession period. 		\label{fig:cumsfes2}}
\end{figure}

In Figure \ref{fig:selected}, we plot the fraction of selected covariates by the sg-LASSO-MIDAS model when we use the macro, financial, and textual analysis data. For each reference quarter, we compute the ratio of each group of variables relative to the total number of covariates. In each subfigure, we plot the three different horizons. For all horizons, the macro series are selected more often than financial and/or textual data. The number of selected series increases with the horizon, however, the pattern of denser macro series and sparser financial and textual series is visible for all three horizons. The results are in line with the literature -- macro series tend to co-move, hence we see a denser pattern in the selection of such series, see e.g. \cite{bok2018macroeconomic}. On the other hand, the alternative textual analysis data appear to be very sparse, yet still important for nowcasting accuracy, see also \cite{thorsrud2018words}. 

%

\begin{figure}[!htbp]
	\centering\includegraphics[width=15cm]{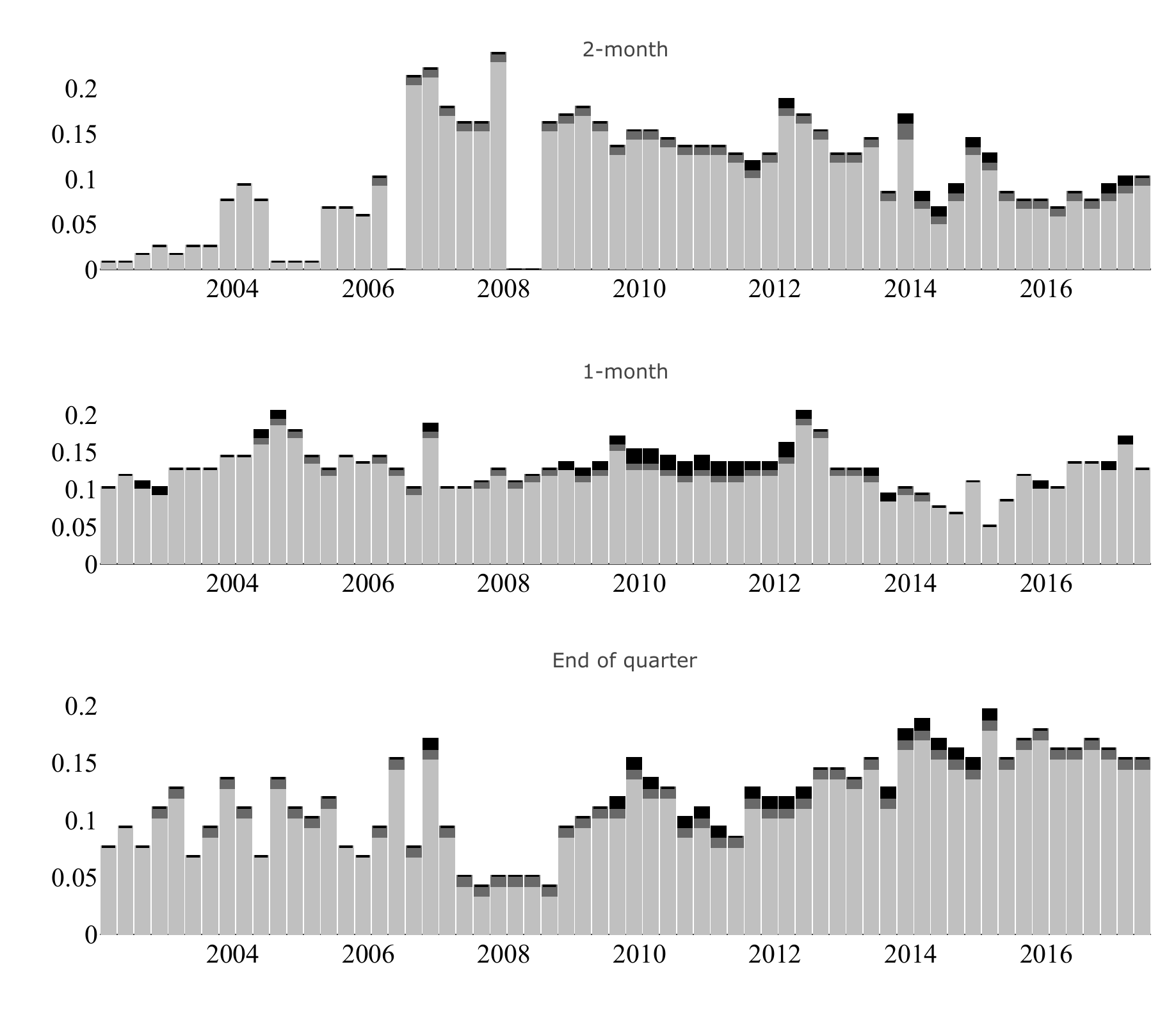}
	\caption{\footnotesize The fraction of selected covariates attributed to macro (light gray), financial (dark gray), and textual (black) data for three monthly horizons.	\label{fig:selected}}
\end{figure}



\section{Conclusion \label{sec:conclusion}}
This paper offers a new perspective on the high-dimensional time series regressions with data sampled at the same or mixed frequencies and contributes more broadly to the rapidly growing literature on the estimation, inference, forecasting, and nowcasting with regularized machine learning methods. The first contribution of the paper is to introduce the sparse-group LASSO estimator for high-dimensional time series regressions. An attractive feature of the estimator is that it recognizes time series data structures and allows us to perform the hierarchical model selection within and between groups. The classical LASSO and the group LASSO are covered as special cases. 

To recognize that the economic and financial time series have typically heavier than Gaussian tails, we use a new Fuk-Nagaev concentration inequality, from \cite{babiietalinference}, valid for a large class of $\tau$-mixing processes, including $\alpha$-mixing processes commonly used in econometrics. Building on this inequality, we establish the nonasymptotic and asymptotic properties of the sparse-group LASSO estimator.

Our empirical application provides new perspectives on applying machine learning methods to real-time forecasting, nowcasting, and monitoring with time series data, including unconventional data, sampled at different frequencies. To that end, we introduce a new class of MIDAS regressions with dictionaries linear in the parameters and based on orthogonal polynomials with lag selection performed by the sg-LASSO estimator. We find that the sg-LASSO outperforms the unstructured LASSO in small samples and conclude that incorporating specific data structures should be helpful in various applications. 

Our empirical results also show that the sg-LASSO-MIDAS using only macro data performs statistically better than NY Fed nowcasts at  2-month horizons and overall for the 1- and 2-month and end-of-quarter horizons. This is a comparison involving the same data and, therefore, pertains to models. This implies that machine learning models are, for this particular case, better than the state space dynamic factor models. When we add the financial data and the textual news data, a total of 118 series,  we find significant improvements in prediction quality for the shorter 1-month and end-of-quarter horizons. 

\section*{Acknowledgments}
We thank  participants at the Financial Econometrics Conference at the TSE Toulouse, the JRC Big Data and Forecasting Conference, the Big Data and Machine Learning in Econometrics, Finance, and Statistics Conference at the University of Chicago, the Nontraditional Data, Machine Learning, and Natural Language Processing in Macroeconomics Conference at the Board of Governors, the AI Innovations Forum organized by SAS and the Kenan Institute, the 12th World Congress of the Econometric Society, and seminar participants at the Vanderbilt University, as well as Harold Chiang, Jianqing Fan, Jonathan Hill, Michele Lenza, and Dacheng Xiu for comments. We are also grateful to the referees and the editor whose comments helped us to improve our paper significantly. All remaining errors are ours.

\clearpage
\bibliographystyle{econometrica}
\bibliography{midas_ml}

\newpage
\spacingset{1.45} 
\setcounter{page}{1}
\setcounter{section}{0}
\setcounter{equation}{0}
\setcounter{table}{0}
\setcounter{figure}{0}
\renewcommand{\theequation}{A.\arabic{equation}}
\renewcommand\thetable{A.\arabic{table}}
\renewcommand\thefigure{A.\arabic{figure}}
\renewcommand\thesection{A.\arabic{section}}
\renewcommand\thepage{Online Appendix - \arabic{page}}
\renewcommand\thetheorem{A.\arabic{theorem}}

\begin{center}
	{\LARGE\textbf{ONLINE APPENDIX}}	
\end{center}
\bigskip

\section{Dictionaries \label{appendix:dictionaries}}
In this section, we review the choice of dictionaries for the MIDAS weight function. It is possible to construct dictionaries using arbitrary sets of functions, including a mix of algebraic polynomials, trigonometric polynomials, B-splines, Haar basis, or wavelets. In this paper, we mostly focus on dictionaries generated by orthogonalized algebraic polynomials, though it might be interesting to tailor the dictionary for each particular application. The attractiveness of algebraic polynomials comes from their ability to generate a variety of shapes with a relatively low number of parameters, which is especially desirable in low signal-to-noise environments. The general family of appropriate orthogonal algebraic polynomials is given by Jacobi polynomials that nest Legendre, Gegenbauer, and Chebychev's polynomials as a special case. 
\begin{example}[Jacobi polynomials]
	Applying the Gram-Schmidt orthogonalization to the power polynomials $\{1,x,x^2,x^3,\dots \}$ with respect to the measure
	\begin{equation*}
		\dx\mu(x) = (1-x)^{\alpha}(1+x)^{\beta}\dx x,\qquad \alpha,\beta>-1,
	\end{equation*}
	on $[-1,1]$, we obtain Jacobi polynomials. In practice Jacobi polynomials can be computed through the well-known tree-term recurrence relation for $n\geq 0$
	\begin{equation*}
		P_{n+1}^{(\alpha,\beta)}(x) = axP_{n}^{(\alpha,\beta)}(x) + bP_{n}^{(\alpha,\beta)}(x) - cP_{n-1}^{(\alpha,\beta)}(x)
	\end{equation*}
	with $a = (2n+\alpha+\beta+1)(2n+\alpha+\beta+2)/2(n+1)(n+\alpha+\beta+1)$, $b=(2n+\alpha+\beta+1)(\alpha^2-\beta^2)/2(n+1)(n+\alpha+\beta+1)(2n+\alpha+\beta)$, and $c = (\alpha+n)(\beta+n)(2n+\alpha+\beta+2)/(n+1)(n+\alpha+\beta+1)(2n+\alpha+\beta)$. To obtain the orthogonal basis on $[0,1]$, we shift Jacobi polynomials with affine bijection $x\mapsto 2x-1$.
	
	For $\alpha=\beta$, we obtain Gegenbauer polynomials, for $\alpha=\beta=0$, we obtain Legendre polynomials, while for $\alpha=\beta=-1/2$ or $\alpha=\beta=1/2$, we obtain Chebychev's polynomials of two kinds.
\end{example}
In the mixed frequency setting, non-orthogonalized polynomials, $\{1,x,x^2,x^3,\dots \}$, are also called Almon polynomials. It is preferable to use orthogonal polynomials in practice due to reduced multicollinearity and better numerical properties. At the same time, orthogonal polynomials are available in Matlab, R, Python, and Julia packages. Legendre polynomials is our default recommendation, while other choices of $\alpha$ and $\beta$ are preferable if we want to accommodate MIDAS weights with other integrability/tail properties.

\smallskip

We noted in the main body of the paper that the specification in equation (\ref{eq:midaspoly}) deviates from the standard MIDAS polynomial specification as it results in a linear regression model - a subtle but key innovation as it maps MIDAS regressions in the standard regression framework. Moreover, casting the MIDAS regressions in a linear regression framework renders the optimization problem convex, something only achieved by \cite{S17} using the U-MIDAS of \cite{FMS15} which does not recognize the mixed frequency data structure, unlike our sg-LASSO.

\section{Proofs of main results}\label{appsec:proofs}
\begin{lemma}\label{lemma:dual_norm}
	Consider $\|.\| = \alpha|.|_1 + (1-\alpha)|.|_2$, where $|.|_q$ is $\ell_q$ norm on $\R^p$. Then the dual norm of $\|.\|$, denoted $\|.\|^*$, satisfies
	\begin{equation*}
		\|z\|^* \leq \alpha|z|_{1}^* + (1-\alpha)|z|_{2}^*,\qquad \forall z\in\R^p,
	\end{equation*}
	where $|.|_{1}^*$ is the dual norm of $|.|_1$ and $|.|_{2}^*$ is the dual norm of $|.|_2$.
\end{lemma}
\begin{proof}
	Clearly, $\|.\|$ is a norm. By the convexity of $x\mapsto x^{-1}$ on $(0,\infty)$
	\begin{equation*}
		\begin{aligned}
			\|z\|^* & = \sup_{b\ne 0}\frac{|\langle z,b\rangle|}{\|b\|}  \leq \sup_{b\ne 0}\left\{\alpha\frac{|\langle z,b\rangle|}{|b|_1} + (1-\alpha)\frac{|\langle z,b\rangle|}{|b|_2}  \right\} \\
			& \leq \alpha\sup_{b\ne 0}\frac{|\langle z,b\rangle|}{|b|_1} + (1-\alpha)\sup_{b\ne 0}\frac{|\langle z,b\rangle|}{|b|_2}  \\
			& = \alpha|z|_{1}^* + (1-\alpha)|z|_{2}^*.
		\end{aligned}
	\end{equation*}
\end{proof}	

\begin{proof}[Proof of Theorem~\ref{thm:estimation_bound}]
	By H\"{o}lder's inequality for every $\varsigma>0$
	\begin{equation*}
		\max_{j\in[p]}\|u_0x_{0,j}\|_\varsigma \leq \|u_0\|_{\varsigma q_1}\max_{j\in[p]}\|x_{0,j}\|_{\varsigma q_2}
	\end{equation*}
	with $q_1^{-1} + q_2^{-1}=1$ and $q_1,q_2\geq 1$. Therefore, under Assumption~\ref{as:data} (i), $\max_{j\in[p]}\|u_0x_{0,j}\|_\varsigma=O(1)$ with $\varsigma=qr/(q+r)$. Recall also that $\E[u_tx_{t,j}]=0,\forall j\in[p]$, see equation~(\ref{eq:hd_ts_model}), which in conjunction with Assumption~\ref{as:data} (ii) verifies conditions of Theorem~\ref{appcor:fn_inequality} and shows that there exists $C>0$ such that for every $\delta\in(0,1)$
	\begin{equation}\label{eq:concentration_score}
		\Pr\left(\left|\frac{1}{T}\sum_{t=1}^Tu_tX_t\right|_\infty \leq C\left(\frac{p}{\delta T^{\kappa-1}}\right)^{1/\kappa}\vee \sqrt{\frac{\log(8p/\delta)}{T}}\right) \geq 1-\delta.
	\end{equation}
	Let $G^* = \max_{G\in\mathcal{G}}|G|$ be the size of the largest group in $\mathcal{G}$. Note that the sg-LASSO penalty $\Omega$ is a norm. By Lemma~\ref{lemma:dual_norm}, its dual norm satisfies
	\begin{equation}\label{eq:dual_norm_bound}
		\begin{aligned}
			\Omega^*(\mathbf{X}^\top\mathbf{u}/T) & \leq \alpha|\mathbf{X}^\top\mathbf{u}/T|_\infty + (1-\alpha)\max_{G\in\mathcal{G}}|(\mathbf{X}^\top\mathbf{u})_G/T|_2 \\
			& \leq (\alpha + (1-\alpha)\sqrt{G^*})|\mathbf{X}^\top\mathbf{u}/T|_\infty  \\
			& \leq (\alpha + (1-\alpha)\sqrt{G^*})C\left(\frac{p}{\delta T^{\kappa-1}}\right)^{1/\kappa}\vee \sqrt{\frac{\log(8p/\delta)}{T}} \\
			& \leq \lambda/c,
		\end{aligned}
	\end{equation}
	where the first inequality follows since $|z|_1^*=|z|_\infty$ and $\left(\sum_{G\in\mathcal{G}}|z_G|_2\right)^* = \max_{G\in\mathcal{G}}|z_G|_2$, the second by elementary computations, the third by equation (\ref{eq:concentration_score}) with probability at least $1-\delta$ for every $\delta\in(0,1)$, and the last from the definition of $\lambda$ in Assumption~\ref{as:tuning}, where $c>1$ is as in Assumption~\ref{as:covariance}.
	By Fermat's rule, the sg-LASSO satisfies
	$$\mathbf{X}^\top(\mathbf{X}\hat\beta - \mathbf{y})/T + \lambda z^* = 0$$
	for some $z^*\in\partial\Omega(\hat\beta)$, where $\partial\Omega(\hat\beta)$ is the subdifferential of $b\mapsto\Omega(b)$ at $\hat\beta$. Taking the inner product with $\beta-\hat\beta$
	\begin{equation*}
		\begin{aligned}
			\langle \mathbf{X}^\top(\mathbf{y} - \mathbf{X}\hat\beta ),\beta - \hat\beta\rangle_T & = \lambda\langle z^*,\beta - \hat\beta\rangle \leq \lambda\left\{\Omega(\beta) - \Omega(\hat\beta) \right\},
		\end{aligned}
	\end{equation*}
	where the inequality follows from the definition of the subdifferential. Using $\mathbf{y}=\mathbf{m} + \mathbf{u}$ and rearranging this inequality
	\begin{equation*}
		\begin{aligned}
			\|\mathbf{X}(\hat\beta - \beta)\|^2_T - \lambda\left\{\Omega(\beta) - \Omega(\hat\beta) \right\} & \leq \langle \mathbf{X}^\top\mathbf{u},\hat\beta - \beta\rangle_T + \langle \mathbf{X}^\top(\mathbf{m}-\mathbf{X}\beta),\hat\beta - \beta\rangle_T \\
			& \leq \Omega^*\left(\mathbf{X}^\top\mathbf{u}/T\right)\Omega(\hat\beta - \beta) + \|\mathbf{X}(\hat\beta - \beta)\|_T\|\mathbf{m} - \mathbf{X}\beta\|_T \\
			& \leq c^{-1}\lambda\Omega(\hat\beta - \beta) + \|\mathbf{X}(\hat\beta - \beta)\|_T\|\mathbf{m} - \mathbf{X}\beta\|_T.
		\end{aligned}
	\end{equation*}
	where the second line follows by the dual norm inequality and the last by $\Omega^*(\mathbf{X}^\top\mathbf{u}/T)\leq \lambda/c$ as shown in equation (\ref{eq:dual_norm_bound}). Therefore,
	\begin{equation}\label{eq:prediction_bound}
		\begin{aligned}
			\|\mathbf{X}\Delta\|^2_T & \leq  c^{-1}\lambda\Omega(\Delta) + \|\mathbf{X}\Delta\|_T\|\mathbf{m} - \mathbf{X}\beta\|_T + \lambda\left\{\Omega(\beta) - \Omega(\hat\beta) \right\} \\
			& \leq (c^{-1}+1)\lambda\Omega(\Delta) + \|\mathbf{X}\Delta\|_T\|\mathbf{m} - \mathbf{X}\beta\|_T
		\end{aligned}
	\end{equation}
	with $\Delta = \hat\beta - \beta$. Note that the sg-LASSO penalty can be decomposed as a sum of two seminorms $\Omega(b) = \Omega_0(b) + \Omega_1(b),\;\forall b\in\R^p$ with
	\begin{equation*}
		\Omega_0(b) = \alpha|b_{S_0}|_1 + (1-\alpha)\sum_{G\in \mathcal{G}_0}|b_{G}|_2\qquad \text{and} \qquad  \Omega_1(b) =\alpha|b_{S^c_0}|_1 + (1-\alpha)\sum_{G\in \mathcal{G}_0^c}|b_{G}|_2.
	\end{equation*}
	Note also that $\Omega_1(\beta)=0$ and $\Omega_1(\hat\beta)=\Omega_1(\Delta)$. Then by the triangle inequality
	\begin{equation}\label{eq:norm_inequality}
		\Omega(\beta) - \Omega(\hat\beta) \leq \Omega_0(\Delta) - \Omega_1(\Delta).
	\end{equation}
	If $\|\mathbf{m} - \mathbf{X}\beta\|_T \leq 2^{-1}\|\mathbf{X}\Delta\|_T$, then it follows from the first inequality in equation (\ref{eq:prediction_bound}) and equation (\ref{eq:norm_inequality}) that
	\begin{equation*}
		\|\mathbf{X}\Delta\|^2_T \leq 2 c^{-1}\lambda\Omega(\Delta) + 2\lambda\left\{\Omega_0(\Delta) - \Omega_1(\Delta) \right\}.
	\end{equation*}
	Since the left side of this equation is positive, this shows that $\Omega_1(\Delta)\leq c_0\Omega_0(\Delta)$ with $c_0=(c+1)/(c-1)$, and whence $\Delta\in\mathcal{C}(c_0)$, cf., Assumption~\ref{as:covariance}. Then 
	\begin{equation}\label{eq:Omega_Sigma_bound}
		\begin{aligned}
			\Omega(\Delta) & \leq (1+c_0)\Omega_0(\Delta) \\
			& \leq (1+c_0)\left(\alpha\sqrt{|S_0|}|\Delta_{S_0}|_2 + (1-\alpha)\sqrt{|\mathcal{G}_0|}\sqrt{\sum_{G\in\mathcal{G}_0}|\Delta_G|_2^2}\right) \\
			& \leq (1+c_0)\sqrt{s_\alpha}\sqrt{\sum_{G\in\mathcal{G}_0}|\Delta_G|_2^2} \\
			& \leq (1+c_0)\sqrt{s_\alpha/\gamma\Delta^\top\Sigma\Delta},	
		\end{aligned}
	\end{equation}
	where we use the Jensen's inequality, Assumption~\ref{as:covariance}, and the definition of $\sqrt{s_\alpha}$. Next, note that
	\begin{equation}\label{eq:Sigma_bound}
		\begin{aligned}
			\Delta^\top\Sigma\Delta & = \|\mathbf{X}\Delta\|^2_T +  \Delta^\top(\Sigma - \hat\Sigma)\Delta \\
			& \leq 2(c^{-1}+1)\lambda\Omega(\Delta) + \Omega(\Delta)\Omega^{*}\left((\hat\Sigma - \Sigma)\Delta\right)  \\
			& \leq 2(c^{-1}+1)\lambda\Omega(\Delta) +  \Omega^2(\Delta)G^*|\mathrm{vech}(\hat\Sigma - \Sigma)|_\infty,	
		\end{aligned}
	\end{equation}
	where the first inequality follows from equation (\ref{eq:prediction_bound}) and the dual norm inequality and the second by Lemma~\ref{lemma:dual_norm} and elementary computations
	\begin{equation*}
		\begin{aligned}
			\Omega^*\left((\hat\Sigma - \Sigma)\Delta\right) & \leq \alpha|(\hat\Sigma - \Sigma)\Delta |_\infty + (1-\alpha)\max_{G\in\mathcal{G}}\left|[(\hat\Sigma - \Sigma)\Delta]_{G}\right|_2 \\
			& \leq \alpha|\Delta|_1|\mathrm{vech}(\hat\Sigma - \Sigma)|_\infty + (1-\alpha)\sqrt{G^*}|\mathrm{vech}(\hat\Sigma - \Sigma)|_\infty|\Delta|_1 \\
			& \leq G^*|\mathrm{vech}(\hat\Sigma - \Sigma)|_\infty\Omega(\Delta).
		\end{aligned}
	\end{equation*}
	Combining the inequalities obtained in equations (\ref{eq:Omega_Sigma_bound} and \ref{eq:Sigma_bound})
	\begin{equation}\label{eq:Omega_inequality}
		\begin{aligned}
			\Omega(\Delta) & \leq (1+c_0)^2\gamma^{-1}s_\alpha\left\{2(c^{-1}+1)\lambda + G^*|\mathrm{vech}(\hat\Sigma - \Sigma)|_\infty\Omega(\Delta) \right\} \\
			& \leq 2(1+c_0)^2\gamma^{-1}s_\alpha(c^{-1}+1)\lambda + (1-A^{-1})\Omega(\Delta),
		\end{aligned}
	\end{equation}
	where the second line holds on the event $E \triangleq \{|\mathrm{vech}(\hat\Sigma - \Sigma)|_\infty \leq \gamma/2G^*s_\alpha(1+2c_0)^2 \}$ with $1 - A^{-1} = (1+c_0)^2/2(1+2c_0)^2<1$. Therefore, inequalities in equation (\ref{eq:prediction_bound} and \ref{eq:Omega_inequality}) yield
	\begin{equation*}
		\begin{aligned}
			\Omega(\Delta) & \leq \frac{2A}{\gamma}(1+c_0)^2(c^{-1}+1)s_\alpha\lambda \\
			\|\mathbf{X}\Delta\|^2_T & \leq \frac{4A}{\gamma}(1+c_0)^2(c^{-1}+1)^2s_\alpha\lambda^2.
		\end{aligned}
	\end{equation*}
	
	\noindent	On the other hand, if $\|\mathbf{m} - \mathbf{X}\beta\|_T > 2^{-1}\|\mathbf{X}\Delta\|_T$, then
	\begin{equation*}
		\|\mathbf{X}\Delta\|^2_T \leq 4\|\mathbf{m} - \mathbf{X}\beta\|_T^2.
	\end{equation*}
	Therefore, on the event $E$ we always have
	\begin{equation}\label{eq:prediction}
		\|\mathbf{X}\Delta\|^2_T \leq C_1s_\alpha\lambda^2 + 4\|\mathbf{m} - \mathbf{X}\beta\|_T^2
	\end{equation}
	with $C_1 = 4A\gamma^{-1}(1+c_0)^2(c^{-1}+1)^2$. This proves the first claim of Theorem~\ref{thm:estimation_bound} if we show that $\Pr(E^c)\leq 2p(p+1)(c_1T^{1-\mu}s_\alpha^{\mu} + \exp(-c_2T/s_\alpha^2)$. To that end, by the Cauchy-Schwartz inequality under Assumptions~\ref{as:data} (i)
	\begin{equation*}
		\max_{1\leq j\leq k\leq p}\|x_{0,j}x_{0,k}\|_{r/2} \leq \max_{j\in[p]}\|x_{0,j}\|_{r}^2 = O(1).
	\end{equation*}
	This in conjunction with Assumption~\ref{as:data} (ii) verifies assumptions of \cite{babiietalinference}, Theorem 3.1 and shows that
	\begin{equation*}
		\begin{aligned}
			\Pr(E^c) & = \Pr\left(\left|\frac{1}{T}\sum_{t=1}^Tx_{t,j}x_{t,k} - \E[x_{t,j}x_{t,k}]\right|_\infty > \frac{\gamma}{2G^*s_\alpha(1+2c_0)^2}\right) \\
			& \leq c_1T^{1-\mu}s_\alpha^{\mu}p(p+1) + 2p(p+1)\exp\left(-\frac{c_2T^2}{s_\alpha^2B_T^2}\right)
		\end{aligned}
	\end{equation*}
	for some $c_1,c_2>0$ and $B_T^2 = \max_{j,k\in[p]}\sum_{t=1}^T\sum_{l=1}^T|\Cov(x_{t,j}x_{t,k}, x_{l,j}x_{l,k})|$. Lastly, under Assumption~\ref{as:data}, by \cite{babiietalinference}, Lemma~A.1.2, $B_T^2 = O(T)$.
	
	To prove the second claim of Theorem~\ref{thm:estimation_bound}, suppose first that $\Delta\in\mathcal{C}(2c_0)$. Then on the event $E$
	\begin{equation*}
		\begin{aligned}
			\Omega^2(\Delta) & = (\Omega_0(\Delta) + \Omega_1(\Delta))^2 \\
			& \leq (1+2c_0)^2\Omega_0^2(\Delta) \\
			& \leq (1+2c_0)^2\Delta^\top\Sigma\Delta s_\alpha/\gamma \\
			& = (1+2c_0)^2\left\{\|\mathbf{X}\Delta\|^2_T +  \Delta^\top(\Sigma-\hat\Sigma)\Delta\right\}s_\alpha/\gamma \\
			& \leq (1+2c_0)^2\left\{C_1s_\alpha\lambda^2 + 4\|\mathbf{m} - \mathbf{X}\beta\|_T^2 +  \Omega^2(\Delta)G^*|\mathrm{vech}(\hat\Sigma - \Sigma)|_\infty\right\}s_\alpha/\gamma \\
			& \leq (1+2c_0)^2\left\{C_1s_\alpha\lambda^2 + 4\|\mathbf{m} - \mathbf{X}\beta\|_T^2\right\}s_\alpha/\gamma + \frac{1}{2}\Omega^2(\Delta),
		\end{aligned}
	\end{equation*}
	where we use the inequality in equations (\ref{eq:Omega_Sigma_bound}, \ref{eq:Sigma_bound}, and \ref{eq:prediction}). Therefore,
	\begin{equation}\label{eq:norm_bound}
		\Omega^2(\Delta) \leq 2(1+2c_0)^2\left\{C_1s_\alpha\lambda^2 + 4\|\mathbf{m} - \mathbf{X}\beta\|_T^2\right\}s_\alpha/\gamma.
	\end{equation}
	On the other hand, if $\Delta\not\in\mathcal{C}(2c_0)$, then $\Delta\not\in\mathcal{C}(c_0)$, which as we have already shown implies $\|\mathbf{m} - \mathbf{X}\beta\|_T>2^{-1}\|\mathbf{X}\Delta\|_T$. In conjunction with equations (\ref{eq:prediction_bound} and \ref{eq:norm_inequality}), this shows that
	\begin{equation*}
		0\leq \lambda c^{-1}\Omega(\Delta) + 2\|\mathbf{m} - \mathbf{X}\beta\|_T^2 + \lambda\left\{\Omega_0(\Delta) - \Omega_1(\Delta) \right\},
	\end{equation*}
	and whence
	\begin{equation*}
		\begin{aligned}
			\Omega_1(\Delta) & \leq c_0\Omega_0(\Delta) + \frac{2c}{\lambda(c-1)}\|\mathbf{m} - \mathbf{X}\beta\|_T^2 \\
			& \leq \frac{1}{2}\Omega_1(\Delta) + \frac{2c}{\lambda(c-1)}\|\mathbf{m} - \mathbf{X}\beta\|_T^2.
		\end{aligned}
	\end{equation*}
	This shows that
	\begin{equation*}
		\begin{aligned}
			\Omega(\Delta) & \leq (1+(2c_0)^{-1})\Omega_1(\Delta) \\
			& \leq(1+(2c_0)^{-1})\frac{4c}{\lambda(c-1)}\|\mathbf{m} - \mathbf{X}\beta\|_T^2.
		\end{aligned}
	\end{equation*}
	Combining this with the inequality in equation (\ref{eq:norm_bound}), we obtain the second claim of Theorem~\ref{thm:estimation_bound}.
\end{proof}

The following result is proven in \cite{babiietalinference}, see their Theorem~A.1.

\begin{theorem}\label{appcor:fn_inequality}
	Let $(\xi_t)_{t\in\Z}$ be a centered stationary stochastic process in $\R^p$ such that (i) for some $\varsigma>2$, $\max_{j\in[p]}\|\xi_{0,j}\|_\varsigma = O(1)$; (ii)  for every $j\in[p]$, $\tau$-mixing coefficients of $\xi_{t,j}$ satisfy $\tau_k^{(j)}\leq ck^{-a}$ for some constants $c>0$ and $a>(\varsigma-1)/(\varsigma-2)$. Then there exists $C>0$ such that for every $\delta\in(0,1)$
	\begin{equation*}
		\Pr\left(\left|\frac{1}{T}\sum_{t=1}^T\xi_t\right|_\infty \leq C\left(\frac{p}{\delta T^{\kappa-1}}\right)^{1/\kappa}\vee\sqrt{\frac{\log(8p/\delta)}{T}} \right) \geq 1 - \delta
	\end{equation*}
	with $\kappa = ((a+1)\varsigma - 1)/(a+\varsigma-1)$.
\end{theorem}	

\section{ARDL-MIDAS: moments and $\tau$-mixing coefficients}\label{sec:moments_mixing}
The ARDL-MIDAS process $(y_t)_{t\in\Z}$ is defined as
\begin{equation*}
	\phi(L)y_t = \xi_t,
\end{equation*}
where $\phi(L) = I - \rho_1L - \rho_2L^2 - \dots - \rho_JL^J$ is a lag polynomial and $\xi_t = \sum_{j=0}^p x_{t,j}\gamma_j + u_t$. The process $(y_t)_{t\in\Z}$ is $\tau$-mixing and has finite moments of order $q\geq 1$ as illustrated below.
\begin{assumption}\label{as:ardl_midas}
	Suppose that $(\xi_t)_{t\in\Z}$ is a stationary process such that (i) $\|\xi_t\|_{q}<\infty$ for some $q>1$; (ii) $\beta$-mixing coefficients satisfy $\beta_k\leq Ca^{k}$ for some $a\in(0,1)$ and $C>0$; and (iii) $\phi(z)\ne 0$ for all $z\in\mathbf{C}$ such that $|z|\leq 1$.
\end{assumption}
Note that by \cite{davydov1973mixing}, (ii) holds if $(\xi_t)_{t\in\Z}$ is a geometrically ergodic Markov process and that (iii) rules out the unit root process.

\begin{proposition}\label{prop:ardl_midas_process}
	Under Assumption~\ref{as:ardl_midas}, the ARDL-MIDAS process has moments of order $q>1$ and $\tau$-mixing coefficients $\tau_k \leq C(a^{bk} + c^k)$ for some $c\in(0,1),$ $C>0$, and $b=1-1/q$.
\end{proposition}
\begin{proof}
	Under (iii) we can invert the autoregressive lag polynomial and obtain
	\begin{equation*}
		y_t = \sum_{j=0}^\infty\psi_j\xi_{t-j}
	\end{equation*}
	for some $(\psi_j)_{j=0}^\infty\in\ell_1$. Note that $(y_t)_{t\in\Z}$ has dependent innovations. Clearly, $(y_t)_{t\in\Z}$ is stationary provided that $(\xi_t)_{t\in\Z}$ is stationary, which is the case by the virtue of Assumption~\ref{as:ardl_midas}. Next, since
	\begin{equation*}
		\|y_t\|_q \leq \sum_{j=0}^\infty|\psi_j|\|\xi_0\|_q
	\end{equation*}
	and $\|\xi_0\|_q<\infty$ under (i), we verify that $\|y_t\|_q<\infty$. Let $(\xi_t')_{t\in\Z}$ be a stationary process distributed as $(\xi_{t})_{t\in\Z}$ and independent of $(\xi_t)_{t\leq 0}$. Then by \cite{dedecker2005new}, Example 1, the $\tau$-mixing coefficients of $(y_t)_{t\in\Z}$ satisfy
	\begin{equation*}
		\begin{aligned}
			\tau_k & \leq \|\xi_0 - \xi_0'\|_q\sum_{j\geq k}|\psi_j| + 2\sum_{j=0}^{k-1}|\psi_j|\int_0^{\beta_{k-j}}Q_{\xi_0}(u)\dx u \\
			& \leq 2\|\xi_0\|_q\sum_{j\geq k}|\psi_j| + 2\|\xi_0\|_q\sum_{j=0}^{k-1}|\psi_j|\beta_{k-j}^{1-1/q},
		\end{aligned}
	\end{equation*}
	where $(\beta_k)_{k\geq 1}$ are $\beta$-mixing coefficients of $(\xi_t)_{t\in\Z}$ and the second line follows by H\"{o}lder's inequality. \cite{brockwell1991time}, p.85 shows that there exist $c\in(0,1)$ and $K>0$ such that $|\psi_j|\leq Kc^{j}$. Therefore, 
	\begin{equation*}
		\sum_{j\geq k}|\psi_j| = O(c^k)
	\end{equation*}
	and under (ii) 
	\begin{equation*}
		\sum_{j=0}^{k-1}|\psi_j|\beta_{k-j}^{1-1/q} \leq CK\sum_{j=0}^{k-1}c^ja^{(k-j)(q-1)/q} \leq \begin{cases}
			CK\frac{a^{k(q-1)/q} - c^k}{1 - ca^{(1-q)/q}}	&\text{if } c\ne a^{(q-1)/q}, \\
			CKka^{k(q-1)/q} & \text{otherwise}.
		\end{cases}
	\end{equation*}
	This proves the second statement.
\end{proof}

\newpage

\section{Monte Carlo Simulations}

\begin{table}[!htbp]
	\centering
	\scriptsize
	\addtolength{\tabcolsep}{-4.5pt}
	\begin{tabular}{r cccccc p{0.1cm} cccccc}
		& \tiny{FLOW} & \tiny{STOCK} & \tiny{MIDDLE} & \tiny{LASSO-U} & \tiny{LASSO-M} & \tiny{SGL-M }& &\tiny{FLOW} & \tiny{STOCK} & \tiny{MIDDLE} & \tiny{LASSO-U} & \tiny{LASSO-M} & \tiny{SGL-M}\\\hline
		T&\multicolumn{6}{c}{\underline{Baseline scenario}}&&\multicolumn{6}{c}{\underline{\(\varepsilon_h\sim_{i.i.d.}\text{student-}t(5)\)}}\\
		50 & 1.920 & 2.086 & 2.145 & 1.848 & 1.731 & 1.537 &  & 2.081 & 2.427 & 2.702 & 2.399 & 2.038 & 1.702 \\ 
		& 0.039 & 0.042 & 0.043 & 0.038 & 0.036 & 0.031 &  & 0.042 & 0.053 & 0.062 & 0.056 & 0.050 & 0.041 \\ 
		100 & 1.423 & 1.670 & 1.791 & 1.670 & 1.517 & 1.320 &  & 1.532 & 1.933 & 2.152 & 1.831 & 1.523 & 1.315 \\ 
		& 0.029 & 0.033 & 0.036 & 0.034 & 0.031 & 0.027 &  & 0.030 & 0.039 & 0.044 & 0.037 & 0.031 & 0.027 \\ 
		200 & 1.292 & 1.502 & 1.645 & 1.407 & 1.268 & 1.170 &  & 1.410 & 1.741 & 2.017 & 1.493 & 1.278 & 1.194 \\ 
		& 0.026 & 0.030 & 0.033 & 0.028 & 0.026 & 0.024 &  & 0.029 & 0.035 & 0.043 & 0.031 & 0.026 & 0.024 \\ 
		&\multicolumn{6}{c}{\underline{High-frequency process: VAR(1)}}&&\multicolumn{6}{c}{\underline{Legendre degree $L=5$}}\\
		50 & 1.869 & 2.645 & 2.863 & 2.192 & 1.712 & 1.431 &  & 1.920 & 2.086 & 2.145 & 1.848 & 1.741 & 1.598 \\ 
		& 0.039 & 0.053 & 0.057 & 0.047 & 0.036 & 0.030 &  & 0.039 & 0.042 & 0.043 & 0.038 & 0.035 & 0.032 \\ 
		100 & 1.474 & 2.071 & 2.312 & 1.622 & 1.373 & 1.247 &  & 1.423 & 1.670 & 1.791 & 1.670 & 1.553 & 1.368 \\ 
		& 0.030 & 0.042 & 0.048 & 0.033 & 0.028 & 0.026 &  & 0.029 & 0.033 & 0.036 & 0.034 & 0.032 & 0.028 \\ 
		200 & 1.335 & 1.919 & 2.080 & 1.369 & 1.239 & 1.216 &  & 1.292 & 1.502 & 1.645 & 1.407 & 1.298 & 1.187 \\ 
		& 0.026 & 0.039 & 0.042 & 0.029 & 0.025 & 0.025 &  & 0.026 & 0.030 & 0.033 & 0.028 & 0.026 & 0.024 \\ 
		&\multicolumn{6}{c}{\underline{Legendre degree $L=10$}}&&\multicolumn{6}{c}{\underline{Low frequency noise level $\sigma^2_u$=5}}\\
		50 & 1.920 & 2.086 & 2.145 & 1.848 & 1.778 & 1.661 &  & 8.927 & 9.048 & 9.020 & 7.714 & 7.308 & 6.929 \\ 
		& 0.039 & 0.042 & 0.043 & 0.038 & 0.037 & 0.034 &  & 0.182 & 0.184 & 0.181 & 0.155 & 0.149 & 0.140 \\ 
		100 & 1.423 & 1.670 & 1.791 & 1.670 & 1.617 & 1.446 &  & 6.643 & 7.300 & 7.536 & 7.510 & 6.953 & 6.305 \\ 
		& 0.029 & 0.033 & 0.036 & 0.034 & 0.033 & 0.029 &  & 0.135 & 0.144 & 0.153 & 0.154 & 0.144 & 0.128 \\ 
		200 & 1.292 & 1.502 & 1.645 & 1.407 & 1.344 & 1.225 &  & 6.008 & 6.580 & 6.902 & 6.809 & 6.270 & 5.703 \\ 
		& 0.026 & 0.030 & 0.033 & 0.028 & 0.027 & 0.025 &  & 0.123 & 0.131 & 0.137 & 0.137 & 0.127 & 0.115 \\ 
		&\multicolumn{6}{c}{\underline{Half high-frequency lags}}&&\multicolumn{6}{c}{\underline{Number of covariates $p=50$}}\\
		50 & 2.256 & 2.117 & 2.505 & 1.885 & 1.816 & 1.623 &  &  &  &  & 1.902 & 1.766 & 1.621 \\ 
		& 0.047 & 0.044 & 0.050 & 0.038 & 0.037 & 0.033 &  &  &  &  & 0.038 & 0.035 & 0.032 \\ 
		100 & 1.655 & 1.685 & 2.079 & 1.679 & 1.595 & 1.370 &  & 3.593 & 3.277 & 3.318 & 1.754 & 1.599 & 1.403 \\ 
		& 0.033 & 0.033 & 0.041 & 0.034 & 0.032 & 0.027 &  & 0.075 & 0.068 & 0.068 & 0.035 & 0.032 & 0.028 \\ 
		200 & 1.528 & 1.539 & 2.005 & 1.365 & 1.355 & 1.202 &  & 1.863 & 1.933 & 2.019 & 1.524 & 1.364 & 1.189 \\ 
		& 0.031 & 0.030 & 0.040 & 0.027 & 0.027 & 0.024 &  & 0.038 & 0.039 & 0.039 & 0.030 & 0.027 & 0.024 \\ 
		&\multicolumn{6}{c}{\underline{Baseline scenario, $\rho = 0.7$}}&&\multicolumn{6}{c}{\underline{Number of covariates $p=50$, $\rho = 0.7$}}\\
		50 & 2.411 & 3.019 & 3.471 & 2.786 & 2.298 & 1.720 &  &  &  &  & 4.588 & 3.604 & 2.145 \\ 
		& 0.051 & 0.059 & 0.069 & 0.061 & 0.051 & 0.036 &  &  &  &  & 0.093 & 0.077 & 0.044 \\ 
		100 & 1.717 & 2.423 & 2.943 & 1.710 & 1.501 & 1.331 &  & 5.351 & 5.030 & 4.854 & 2.275 & 1.910 & 1.424 \\ 
		& 0.034 & 0.048 & 0.058 & 0.035 & 0.031 & 0.027 &  & 0.111 & 0.102 & 0.099 & 0.048 & 0.040 & 0.029 \\ 
		200 & 1.564 & 2.135 & 2.657 & 1.340 & 1.269 & 1.222 &  & 2.384 & 2.826 & 3.290 & 1.499 & 1.385 & 1.217 \\ 
		& 0.032 & 0.043 & 0.052 & 0.027 & 0.026 & 0.025 &  & 0.048 & 0.056 & 0.065 & 0.030 & 0.028 & 0.024 \\\hline 
	\end{tabular}
	\caption{Forecasting accuracy results. -- See Table~\ref{appendix:tab:nowcast}	\label{appendix:tab:forecast}} 
\end{table}

\clearpage

\begin{table}[!htbp]
	\centering
	\scriptsize
	\addtolength{\tabcolsep}{-4.5pt}
	\begin{tabular}{r cccccc p{0.1cm} cccccc}
		& \tiny{FLOW} & \tiny{STOCK} & \tiny{MIDDLE} & \tiny{LASSO-U} & \tiny{LASSO-M} & \tiny{SGL-M } && \tiny{FLOW} & \tiny{STOCK} & \tiny{MIDDLE} & \tiny{LASSO-U} & \tiny{LASSO-M} & \tiny{SGL-M }\\\hline
		T&\multicolumn{6}{c}{\underline{Baseline scenario}}&&\multicolumn{6}{c}{\underline{\(\varepsilon_h\sim_{i.i.d.}\text{student-}t(5)\)}}\\
		50 & 1.987 & 2.113 & 2.184 & 1.870 & 1.753 & 1.606 &  & 2.257 & 2.391 & 2.649 & 2.422 & 2.113 & 1.801 \\ 
		& 0.043 & 0.042 & 0.043 & 0.038 & 0.036 & 0.032 &  & 0.046 & 0.054 & 0.057 & 0.052 & 0.046 & 0.038 \\ 
		100 & 1.446 & 1.632 & 1.769 & 1.667 & 1.541 & 1.345 &  & 1.659 & 1.889 & 2.139 & 1.903 & 1.678 & 1.462 \\ 
		& 0.029 & 0.032 & 0.034 & 0.033 & 0.031 & 0.026 &  & 0.033 & 0.038 & 0.043 & 0.038 & 0.033 & 0.029 \\ 
		200 & 1.318 & 1.482 & 1.609 & 1.448 & 1.328 & 1.220 &  & 1.505 & 1.728 & 1.971 & 1.608 & 1.411 & 1.297 \\ 
		& 0.026 & 0.029 & 0.032 & 0.029 & 0.026 & 0.024 &  & 0.030 & 0.035 & 0.041 & 0.033 & 0.028 & 0.026 \\ 
		&\multicolumn{6}{c}{\underline{High-frequency process: VAR(1)}}&&\multicolumn{6}{c}{\underline{Legendre degree $L=5$}}\\
		50 & 2.086 & 2.418 & 2.856 & 2.254 & 1.817 & 1.503 &  & 1.987 & 2.113 & 2.184 & 1.870 & 1.767 & 1.635 \\ 
		& 0.044 & 0.050 & 0.057 & 0.049 & 0.039 & 0.031 &  & 0.043 & 0.042 & 0.043 & 0.038 & 0.036 & 0.033 \\ 
		100 & 1.642 & 1.935 & 2.365 & 1.690 & 1.459 & 1.328 &  & 1.446 & 1.632 & 1.769 & 1.667 & 1.564 & 1.389 \\ 
		& 0.033 & 0.039 & 0.048 & 0.035 & 0.030 & 0.028 &  & 0.029 & 0.032 & 0.034 & 0.033 & 0.031 & 0.027 \\ 
		200 & 1.475 & 1.771 & 2.247 & 1.442 & 1.312 & 1.268 &  & 1.318 & 1.482 & 1.609 & 1.448 & 1.351 & 1.230 \\ 
		& 0.029 & 0.036 & 0.046 & 0.029 & 0.027 & 0.026 &  & 0.026 & 0.029 & 0.032 & 0.029 & 0.027 & 0.024 \\ 
		&\multicolumn{6}{c}{\underline{Legendre degree $L=10$}}&&\multicolumn{6}{c}{\underline{Low frequency noise level $\sigma^2_u$=5}}\\ 
		50 & 1.987 & 2.113 & 2.184 & 1.870 & 1.799 & 1.698 &  & 9.121 & 9.208 & 9.167 & 7.700 & 7.397 & 7.087 \\ 
		& 0.043 & 0.042 & 0.043 & 0.038 & 0.037 & 0.034 &  & 0.193 & 0.184 & 0.181 & 0.155 & 0.150 & 0.143 \\ 
		100 & 1.446 & 1.632 & 1.769 & 1.667 & 1.606 & 1.454 &  & 6.646 & 7.149 & 7.433 & 7.454 & 6.911 & 6.222 \\ 
		& 0.029 & 0.032 & 0.034 & 0.033 & 0.032 & 0.029 &  & 0.135 & 0.141 & 0.144 & 0.149 & 0.138 & 0.123 \\ 
		200 & 1.318 & 1.482 & 1.609 & 1.448 & 1.400 & 1.267 &  & 6.052 & 6.482 & 6.777 & 6.835 & 6.345 & 5.780 \\ 
		& 0.026 & 0.029 & 0.032 & 0.029 & 0.028 & 0.025 &  & 0.122 & 0.127 & 0.134 & 0.137 & 0.127 & 0.114 \\ 
		&\multicolumn{6}{c}{\underline{Half high-frequency lags}}&&\multicolumn{6}{c}{\underline{Number of covariates $p=50$}}\\
		50 & 2.378 & 2.164 & 2.540 & 1.875 & 1.827 & 1.723 &  &  &  &  & 1.912 & 1.767 & 1.611 \\ 
		& 0.049 & 0.044 & 0.049 & 0.038 & 0.037 & 0.035 &  &  &  &  & 0.039 & 0.035 & 0.033 \\ 
		100 & 1.765 & 1.692 & 2.184 & 1.810 & 1.703 & 1.479 &  & 3.703 & 3.162 & 3.179 & 1.762 & 1.622 & 1.441 \\ 
		& 0.035 & 0.033 & 0.042 & 0.036 & 0.033 & 0.029 &  & 0.076 & 0.064 & 0.068 & 0.035 & 0.032 & 0.028 \\ 
		200 & 1.605 & 1.520 & 1.976 & 1.544 & 1.495 & 1.324 &  & 1.912 & 1.871 & 2.017 & 1.546 & 1.428 & 1.260 \\ 
		& 0.031 & 0.029 & 0.039 & 0.031 & 0.029 & 0.026 &  & 0.038 & 0.037 & 0.040 & 0.032 & 0.029 & 0.026 \\ 
		&\multicolumn{6}{c}{\underline{Baseline scenario, $\rho = 0.7$}}&&\multicolumn{6}{c}{\underline{Number of covariates $p=50$, $\rho = 0.7$}}\\ 
		50 & 2.606 & 2.872 & 3.618 & 2.927 & 2.599 & 1.884 &  &  &  &  & 4.606 & 3.816 & 2.242 \\ 
		& 0.055 & 0.058 & 0.073 & 0.063 & 0.054 & 0.039 &  &  &  &  & 0.096 & 0.083 & 0.046 \\ 
		100 & 1.837 & 2.154 & 3.020 & 1.783 & 1.596 & 1.412 &  & 5.154 & 4.373 & 4.764 & 2.373 & 2.161 & 1.520 \\ 
		& 0.037 & 0.043 & 0.059 & 0.037 & 0.032 & 0.028 &  & 0.102 & 0.089 & 0.100 & 0.051 & 0.046 & 0.030 \\ 
		200 & 1.661 & 1.919 & 2.753 & 1.389 & 1.341 & 1.287 &  & 2.622 & 2.555 & 3.364 & 1.563 & 1.500 & 1.315 \\ 
		& 0.033 & 0.038 & 0.056 & 0.027 & 0.027 & 0.026 &  & 0.052 & 0.051 & 0.067 & 0.032 & 0.031 & 0.027 \\\hline
	\end{tabular}
	\caption{ Nowcasting accuracy results. \\ 
		{\scriptsize The table reports simulation results for nowcasting accuracy. The baseline DGP (upper-left block) is with the low-frequency noise level $\sigma^2_u=1$, the degree of Legendre polynomial $L=3$, and Gaussian high-frequency noise. All remaining blocks report results for deviations from the baseline DGP. In the upper-right block, the noise term of high-frequency covariates is student-$t(5)$. Each block reports results for LASSO-U-MIDAS (LASSO-U), LASSO-MIDAS (LASSO-M), and sg-LASSO-MIDAS (SGL-M) (the last three columns). We also report results for aggregated predictive regressions with flow aggregation (FLOW), stock aggregation (STOCK), and taking the middle value (MIDDLE). We vary the sample size $T$ from 50 to 200. Each entry in the odd row is the average mean squared forecast error, while each even row is the simulation standard error.}    \label{appendix:tab:nowcast}} 
\end{table}

\clearpage

\begin{table}[!htbp]
	\centering
	\scriptsize
	\addtolength{\tabcolsep}{-4.5pt}	
	\begin{tabular}{r  ccccccccc}
		& LASSO-U & LASSO-M & SGL-M & LASSO-U & LASSO-M & SGL-M & LASSO-U & LASSO-M & SGL-M \\\hline
		&\multicolumn{3}{c}{T=50}  &\multicolumn{3}{c}{T=100} &\multicolumn{3}{c}{T=200}\\
		&\multicolumn{9}{c}{\underline{Baseline scenario}}\\
		$\mathrm{Beta}(1,3)$ & 2.028 & 1.867 & 1.312 & 2.005 & 1.518 & 0.733 & 1.947 & 0.809 & 0.388\\ 
		& 0.001 & 0.009 & 0.015 & 0.001 & 0.011 & 0.010 & 0.002 & 0.009 & 0.005 \\ 
		$\mathrm{Beta}(2,3)$ & 1.248 & 1.192 & 0.988 & 1.241 & 1.042 & 0.662 & 1.219 & 0.710 & 0.418 \\ 
		& 0.001 & 0.006 & 0.011 & 0.001 & 0.006 & 0.008 & 0.001 & 0.006 & 0.005 \\ 
		$\mathrm{Beta}(2,2)$ & 1.093 & 1.035 & 0.870 & 1.088 & 0.890 & 0.573 & 1.073 & 0.559 & 0.330  \\ 
		& 0.001 & 0.005 & 0.009 & 0.001 & 0.006 & 0.007 & 0.001 & 0.005 & 0.004 \\ 
		&\multicolumn{9}{c}{\underline{\(\varepsilon_h\sim_{i.i.d.}\text{student-}t(5)\)}}\\
		$\mathrm{Beta}(1,3)$ & 2.015 & 1.671 & 1.023 & 1.964 & 1.027 & 0.465 & 1.892 & 0.434 & 0.248 \\ 
		& 0.001 & 0.011 & 0.014 & 0.002 & 0.011 & 0.007 & 0.001 & 0.005 & 0.004 \\ 
		$\mathrm{Beta}(2,3)$ & 1.242 & 1.107 & 0.816 & 1.223 & 0.807 & 0.462 & 1.191 & 0.479 & 0.297 \\ 
		& 0.001 & 0.007 & 0.010 & 0.001 & 0.007 & 0.006 & 0.001 & 0.005 & 0.004 \\ 
		$\mathrm{Beta}(2,2)$ & 1.088 & 0.959 & 0.740 & 1.075 & 0.664 & 0.403 & 1.051 & 0.348 & 0.221 \\ 
		& 0.001 & 0.006 & 0.009 & 0.001 & 0.006 & 0.006 & 0.001 & 0.004 & 0.003 \\ 
		&\multicolumn{9}{c}{\underline{high-frequency process: VAR(1)}}\\
		$\mathrm{Beta}(1,3)$ & 1.944 & 1.353 & 0.960 & 1.909 & 0.905 & 0.657 & 1.871 & 0.562 & 0.485\\ 
		& 0.003 & 0.014 & 0.014 & 0.002 & 0.010 & 0.009 & 0.002 & 0.006 & 0.006 \\ 
		$\mathrm{Beta}(2,3)$ & 1.186 & 0.917 & 0.821 & 1.166 & 0.662 & 0.594 & 1.147 & 0.508 & 0.490 \\ 
		& 0.002 & 0.012 & 0.013 & 0.002 & 0.009 & 0.008 & 0.001 & 0.006 & 0.005 \\ 
		$\mathrm{Beta}(2,2)$ & 1.045 & 0.778 & 0.754 & 1.032 & 0.550 & 0.540 & 1.019 & 0.412 & 0.422 \\ 
		& 0.002 & 0.011 & 0.012 & 0.001 & 0.008 & 0.008 & 0.001 & 0.005 & 0.005 \\ 
		&\multicolumn{9}{c}{\underline{Legendre degree $L=5$}}\\
		$\mathrm{Beta}(1,3)$ & 2.028 & 1.907 & 1.487 & 2.005 & 1.619 & 0.909 & 1.947 & 0.915 & 0.436\\ 
		& 0.001 & 0.009 & 0.016 & 0.001 & 0.010 & 0.012 & 0.002 & 0.009 & 0.006 \\ 
		$\mathrm{Beta}(2,3)$ & 1.248 & 1.211 & 1.090 & 1.241 & 1.091 & 0.783 & 1.219 & 0.772 & 0.462 \\ 
		& 0.001 & 0.005 & 0.012 & 0.001 & 0.006 & 0.009 & 0.001 & 0.006 & 0.005 \\ 
		$\mathrm{Beta}(2,2)$ & 1.093 & 1.055 & 0.962 & 1.088 & 0.938 & 0.672 & 1.073 & 0.619 & 0.356 \\ 
		& 0.001 & 0.005 & 0.010 & 0.001 & 0.005 & 0.008 & 0.001 & 0.005 & 0.005 \\
		&\multicolumn{9}{c}{\underline{Baseline scenario, $\rho=0.7$}}\\
		$\mathrm{Beta}(1,3)$ & 1.901 & 1.035 & 0.526 & 1.839 & 0.388 & 0.243 & 1.805 & 0.196 & 0.166 \\ 
		& 0.003 & 0.012 & 0.009 & 0.003 & 0.005 & 0.004 & 0.002 & 0.002 & 0.002 \\ 
		$\mathrm{Beta}(2,3)$ & 1.174 & 0.742 & 0.492 & 1.139 & 0.428 & 0.301 & 1.117 & 0.310 & 0.252 \\ 
		& 0.002 & 0.009 & 0.008 & 0.002 & 0.005 & 0.004 & 0.002 & 0.003 & 0.003 \\ 
		$\mathrm{Beta}(2,2)$ & 1.031 & 0.594 & 0.396 & 1.002 & 0.291 & 0.212 & 0.983 & 0.190 & 0.153 \\ 
		& 0.002 & 0.007 & 0.006 & 0.002 & 0.003 & 0.003 & 0.002 & 0.002 & 0.002 \\ 
		\hline
	\end{tabular}
	\caption{Shape of weights estimation accuracy I. \\ {\scriptsize  The table reports results for shape of weights estimation accuracy for the first four DGPs of Tables \ref{appendix:tab:forecast}-\ref{appendix:tab:nowcast} using  LASSO-U, LASSO-M and SGL-M estimators for the weight functions $\mathrm{Beta}(1,3)$, $\mathrm{Beta}(2,3)$, and $\mathrm{Beta}(2,2)$ with sample size $T$ = 50, 100 and 200. Entries in odd rows are the average mean integrated squared error and in even rows  the simulation standard error.	\label{appendix:tab:shape_rec1}}}  
\end{table}
\begin{table}[!htbp]
	\centering
	\scriptsize
	\addtolength{\tabcolsep}{-4.5pt}
	\begin{tabular}{r  ccccccccc}
		& LASSO-U & LASSO-M & SGL-M & LASSO-U & LASSO-M & SGL-M & LASSO-U & LASSO-M & SGL-M \\\hline
		&\multicolumn{3}{c}{T=50}  &\multicolumn{3}{c}{T=100} &\multicolumn{3}{c}{T=200}\\
		&\multicolumn{9}{c}{\underline{Legendre degree $L=10$}}\\
		$\mathrm{Beta}(1,3)$ & 2.028 & 1.962 & 1.685 & 2.005 & 1.769 & 1.150 & 1.947 & 1.078 & 0.528 \\ 
		& 0.001 & 0.008 & 0.016 & 0.001 & 0.010 & 0.013 & 0.002 & 0.011 & 0.007 \\ 
		$\mathrm{Beta}(2,3)$ & 1.248 & 1.247 & 1.247 & 1.241 & 1.168 & 0.960 & 1.219 & 0.869 & 0.522 \\ 
		& 0.001 & 0.004 & 0.012 & 0.001 & 0.005 & 0.010 & 0.001 & 0.006 & 0.006 \\ 
		$\mathrm{Beta}(2,2)$ & 1.093 & 1.086 & 1.091 & 1.088 & 1.011 & 0.823 & 1.073 & 0.710 & 0.398 \\ 
		& 0.001 & 0.004 & 0.011 & 0.001 & 0.005 & 0.009 & 0.001 & 0.006 & 0.005 \\ 
		&\multicolumn{9}{c}{\underline{low frequency noise level $\sigma^2_u$=5}}\\
		$\mathrm{Beta}(1,3)$ & 2.038 & 1.941 & 1.588 & 2.025 & 1.816 & 1.109 & 1.983 & 1.436 & 0.563 \\ 
		& 0.001 & 0.009 & 0.019 & 0.001 & 0.009 & 0.014 & 0.002 & 0.010 & 0.009 \\ 
		$\mathrm{Beta}(2,3)$ & 1.252 & 1.215 & 1.144 & 1.246 & 1.160 & 0.878 & 1.230 & 0.996 & 0.529 \\ 
		& 0.001 & 0.006 & 0.015 & 0.001 & 0.005 & 0.010 & 0.001 & 0.006 & 0.007 \\ 
		$\mathrm{Beta}(2,2)$ & 1.096 & 1.065 & 1.022 & 1.092 & 1.007 & 0.773 & 1.080 & 0.845 & 0.460 \\ 
		& 0.001 & 0.006 & 0.013 & 0.001 & 0.005 & 0.009 & 0.001 & 0.005 & 0.007 \\ 
		&\multicolumn{9}{c}{\underline{Half high-frequency lags}}\\
		$\mathrm{Beta}(1,3)$ & 2.028 & 1.826 & 1.219 & 1.990 & 1.504 & 0.825 & 1.924 & 0.964 & 0.611  \\ 
		& 0.001 & 0.009 & 0.012 & 0.001 & 0.010 & 0.008 & 0.001 & 0.007 & 0.004 \\ 
		$\mathrm{Beta}(2,3)$ & 1.252 & 1.206 & 1.072 & 1.243 & 1.133 & 0.925 & 1.224 & 0.968 & 0.779 \\ 
		& 0.000 & 0.004 & 0.008 & 0.001 & 0.004 & 0.006 & 0.001 & 0.005 & 0.005 \\ 
		$\mathrm{Beta}(2,2)$ & 1.096 & 1.060 & 0.991 & 1.090 & 1.007 & 0.878 & 1.076 & 0.890 & 0.783 \\ 
		& 0.000 & 0.004 & 0.008 & 0.000 & 0.004 & 0.006 & 0.000 & 0.004 & 0.004 \\ 
		&\multicolumn{9}{c}{\underline{Number of covariates $p=50$}}\\
		$\mathrm{Beta}(1,3)$ & 2.044 & 1.998 & 1.586 & 2.032 & 1.867 & 1.061 & 1.999 & 1.285 & 0.512 \\ 
		& 0.000 & 0.004 & 0.012 & 0.001 & 0.007 & 0.011 & 0.001 & 0.009 & 0.006 \\ 
		$\mathrm{Beta}(2,3)$ & 1.255 & 1.238 & 1.099 & 1.252 & 1.191 & 0.875 & 1.243 & 0.963 & 0.533 \\ 
		& 0.000 & 0.002 & 0.007 & 0.000 & 0.004 & 0.007 & 0.001 & 0.005 & 0.005 \\ 
		$\mathrm{Beta}(2,2)$ & 1.099 & 1.083 & 0.979 & 1.097 & 1.036 & 0.782 & 1.091 & 0.804 & 0.467 \\ 
		& 0.000 & 0.002 & 0.007 & 0.000 & 0.003 & 0.006 & 0.000 & 0.005 & 0.005 \\ 
		&\multicolumn{9}{c}{\underline{Number of covariates $p=50$, $\rho=0.7$}}\\ 
		$\mathrm{Beta}(1,3)$ & 1.996 & 1.726 & 0.878 & 1.902 & 0.839 & 0.334 & 1.835 & 0.314 & 0.188 \\ 
		& 0.002 & 0.010 & 0.011 & 0.002 & 0.009 & 0.005 & 0.002 & 0.003 & 0.002 \\ 
		$\mathrm{Beta}(2,3)$ & 1.229 & 1.071 & 0.692 & 1.180 & 0.648 & 0.344 & 1.138 & 0.411 & 0.248 \\ 
		& 0.001 & 0.006 & 0.008 & 0.002 & 0.006 & 0.004 & 0.002 & 0.003 & 0.003 \\ 
		$\mathrm{Beta}(2,2)$ & 1.078 & 0.925 & 0.610 & 1.040 & 0.495 & 0.276 & 1.003 & 0.272 & 0.167 \\ 
		& 0.001 & 0.005 & 0.007 & 0.001 & 0.005 & 0.004 & 0.001 & 0.002 & 0.002 \\ 
		\hline
	\end{tabular}
	\caption{Shape of weights estimation accuracy II. -- See Table \ref{appendix:tab:shape_rec1}	\label{appendix:tab:shape_rec2}}  
\end{table}


\begin{figure}[!htbp]
	\centering
	\begin{subfigure}[b]{0.25\linewidth}
		\caption{\scriptsize LASSO-U-MIDAS}
		\centering\includegraphics[width=4cm]{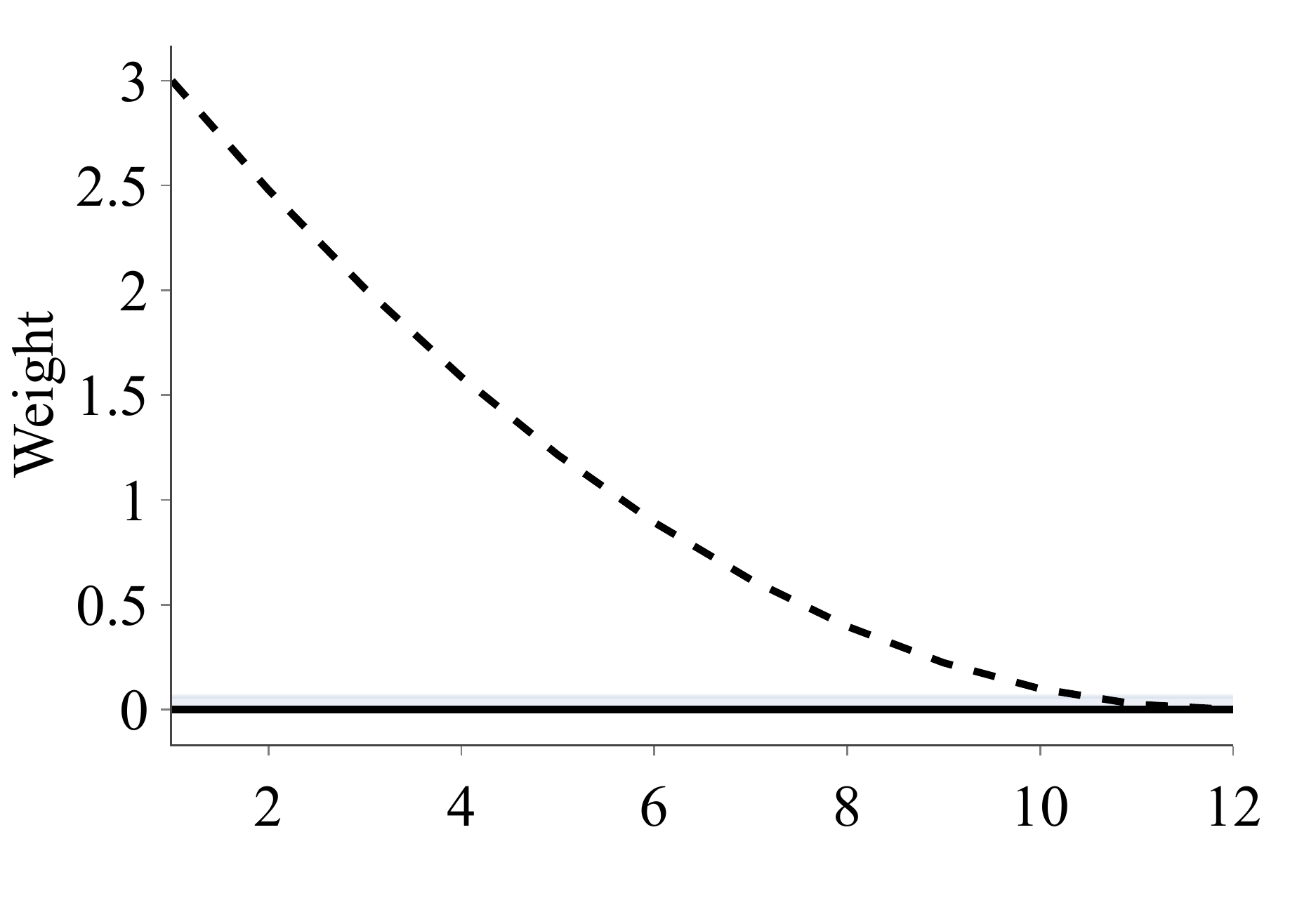}
	\end{subfigure}
	\begin{subfigure}[b]{0.25\linewidth}
		\caption{\scriptsize LASSO-MIDAS}
		\centering\includegraphics[width=4cm]{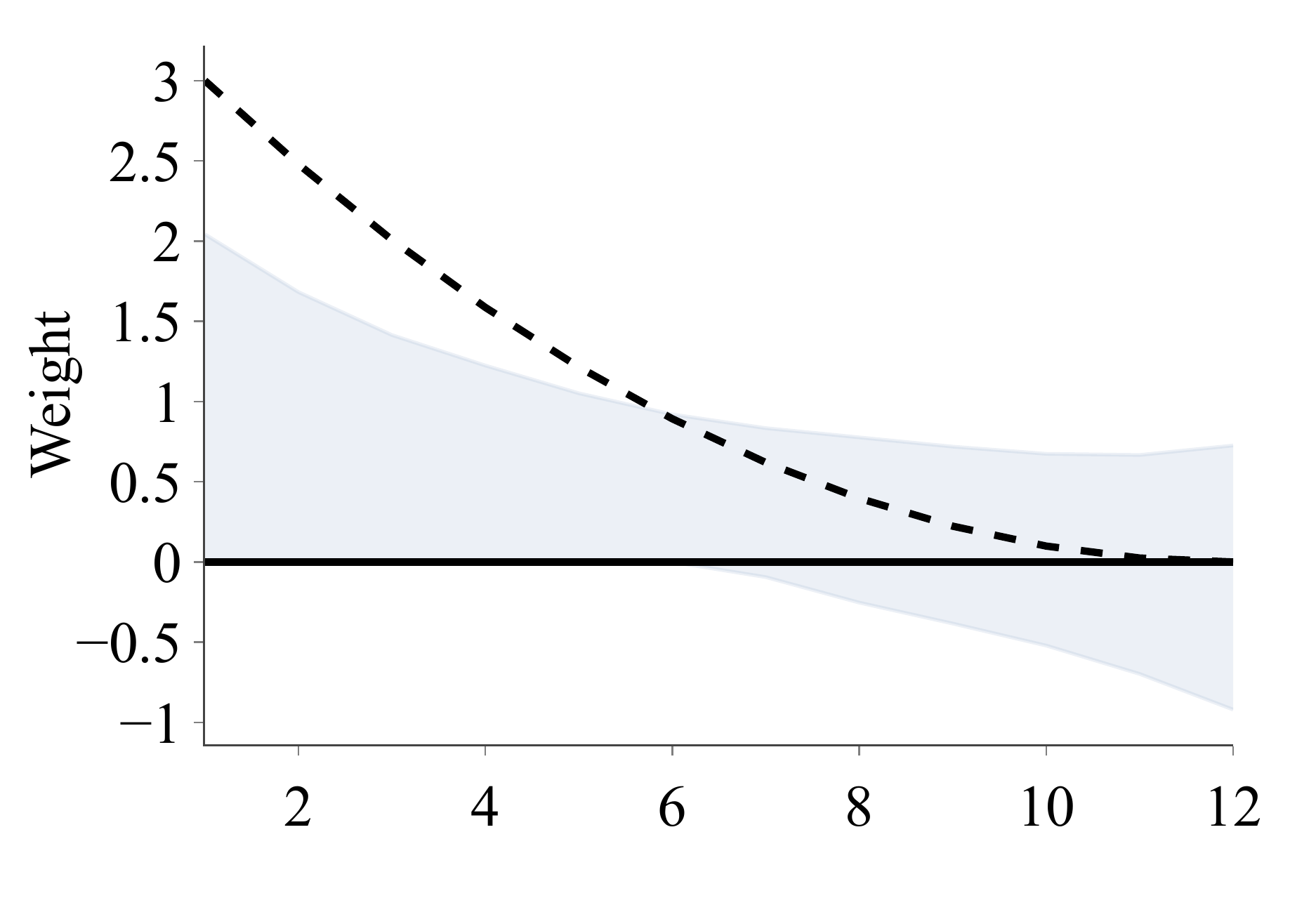}
	\end{subfigure}
	\begin{subfigure}[b]{0.25\linewidth}
		\caption{\scriptsize sg-LASSO-MIDAS}
		\centering\includegraphics[width=4cm]{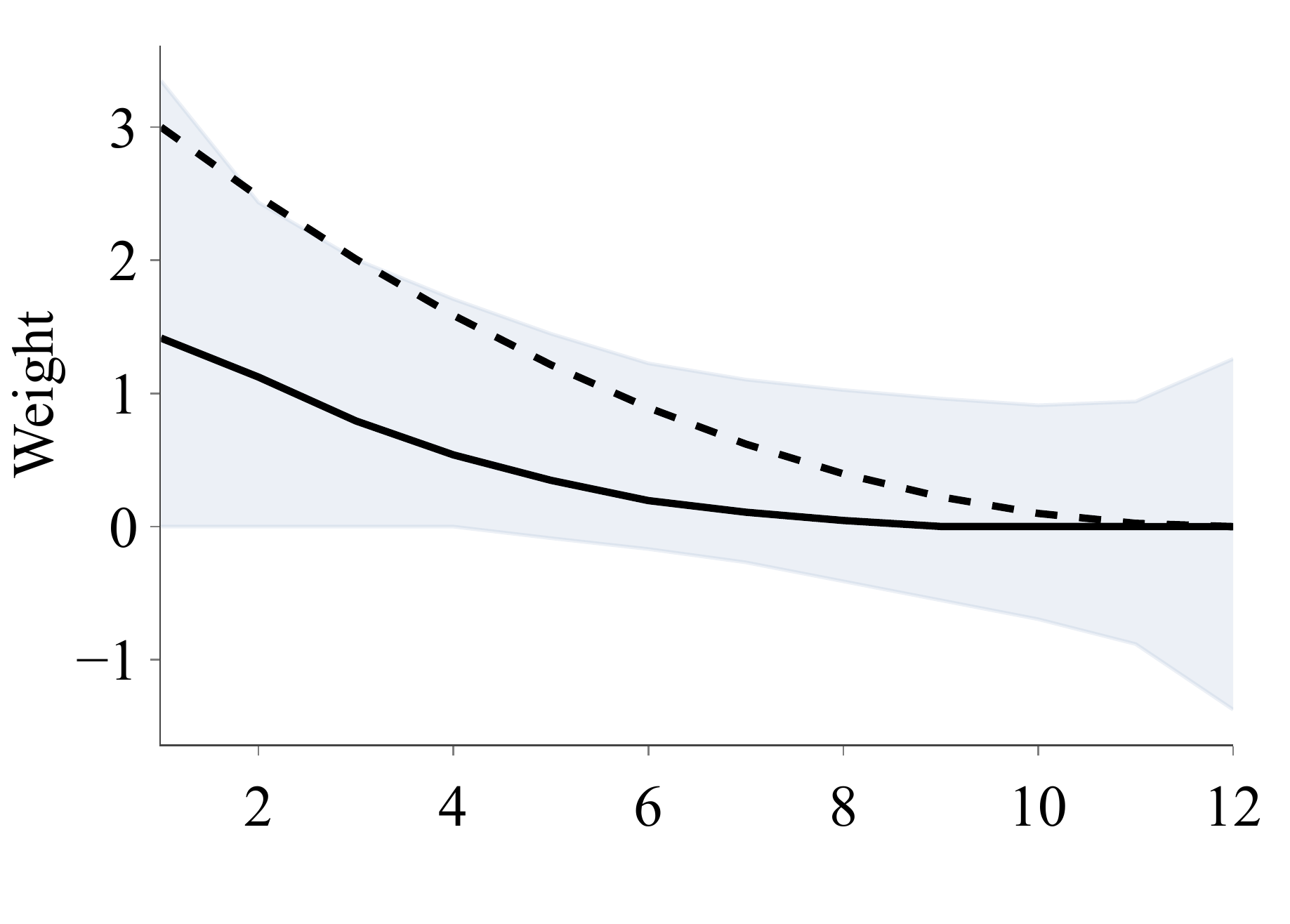}
	\end{subfigure}
	
	\begin{subfigure}[b]{0.25\linewidth}
		\caption{\scriptsize LASSO-U-MIDAS}
		\centering\includegraphics[width=4cm]{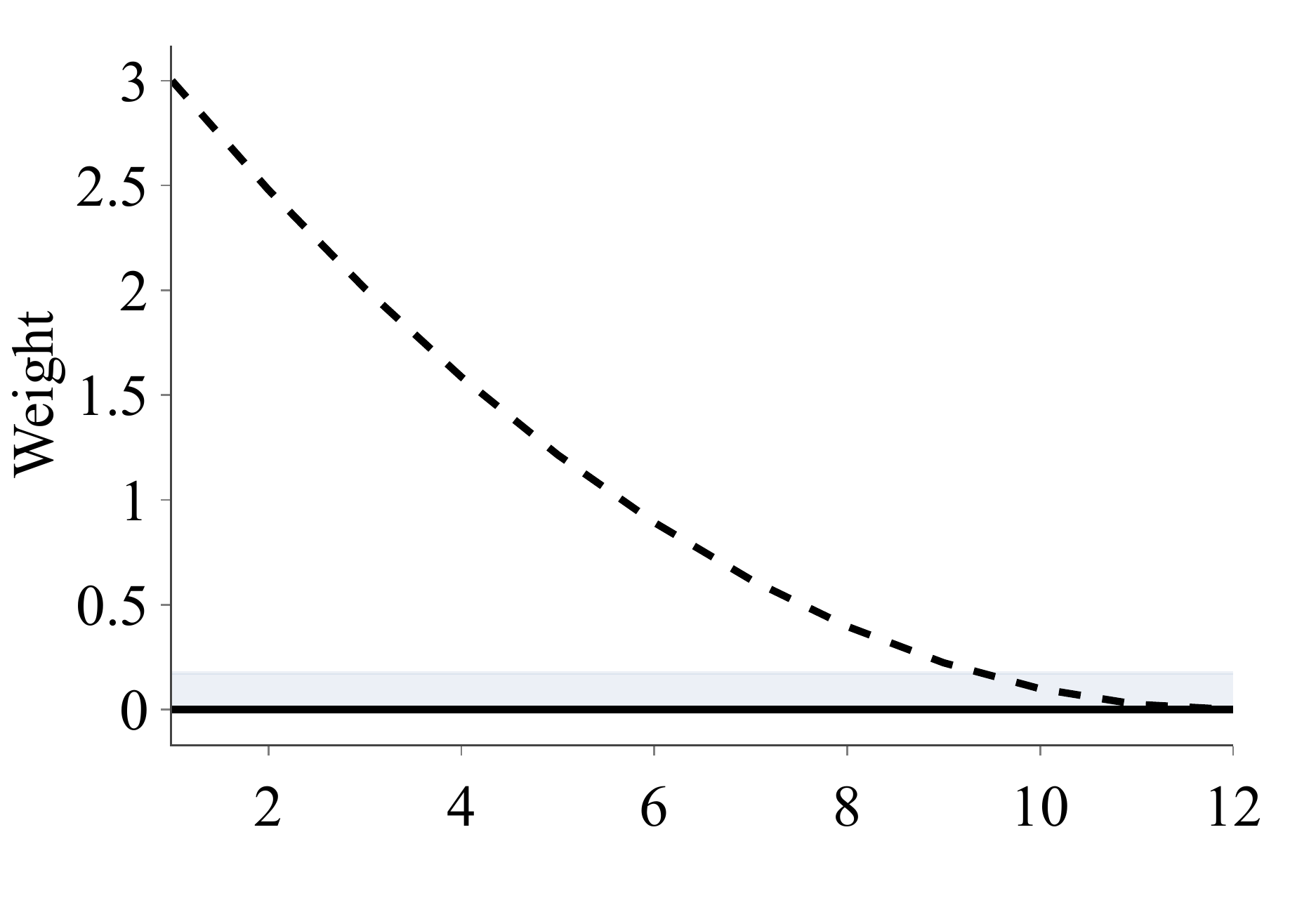}
	\end{subfigure}
	\begin{subfigure}[b]{0.25\linewidth}
		\caption{\scriptsize LASSO-MIDAS}
		\centering\includegraphics[width=4cm]{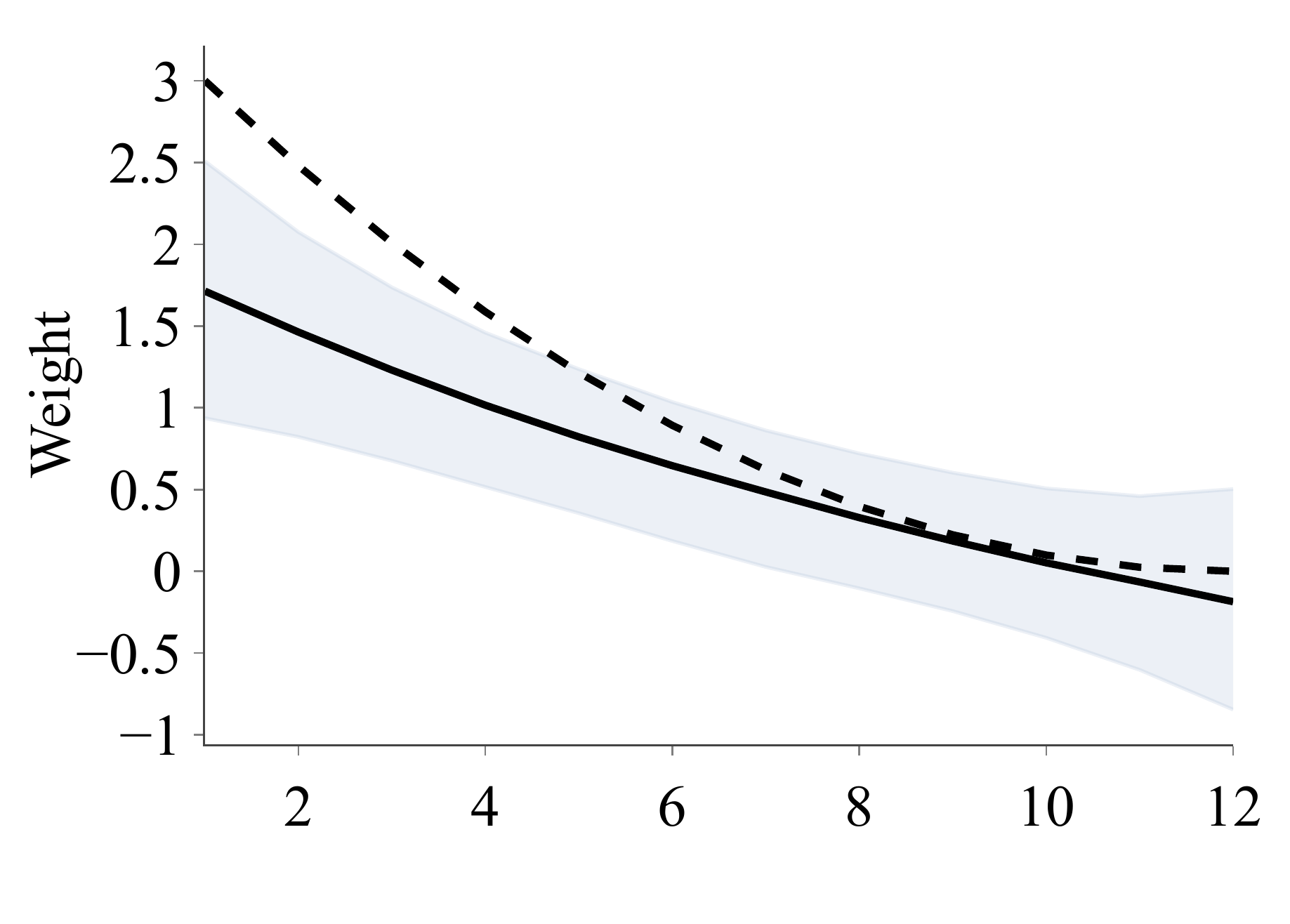}
	\end{subfigure}
	\begin{subfigure}[b]{0.25\linewidth}
		\caption{\scriptsize sg-LASSO-MIDAS}
		\centering\includegraphics[width=4cm]{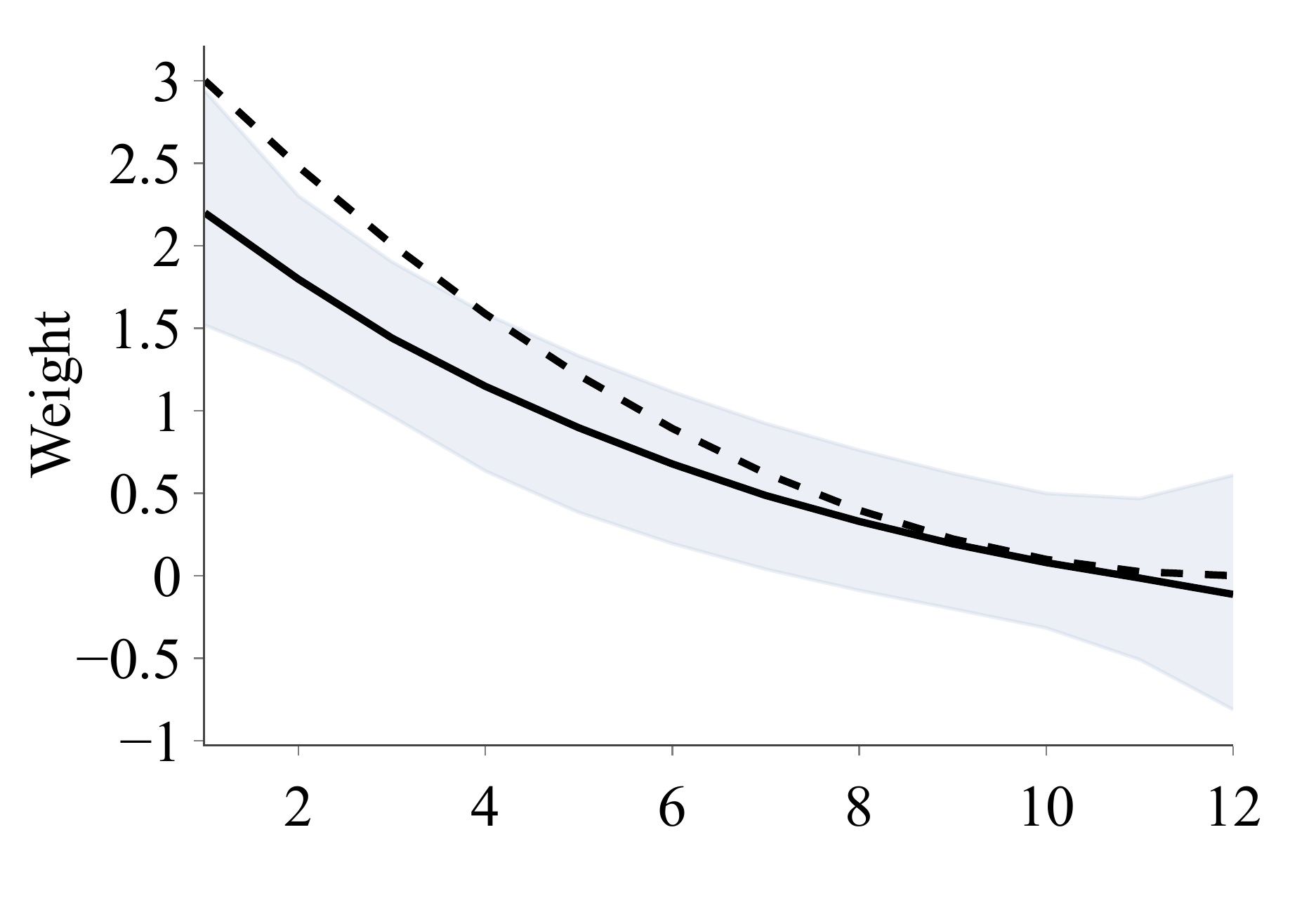}
	\end{subfigure}
	\caption{\footnotesize The figure shows the fitted Beta(1,3) weights. We plot the estimated weights for the LASSO-U-MIDAS, LASSO-MIDAS, and sg-LASSO-MIDAS estimators for the baseline DGP scenario. The first row plots weights for the sample size $T=50$, the second row plots weights for the sample size $T=200$. The black solid line is the median estimate of the weights function, the black dashed line is the population weight function, and the gray area is the 90\% confidence interval.}
	\label{appendix:fig:weights_beta_1}			
\end{figure}

\begin{figure}[!htbp]
	\centering
	\begin{subfigure}[b]{0.25\linewidth}
		\caption{\scriptsize LASSO-U-MIDAS}
		\centering\includegraphics[width=4cm]{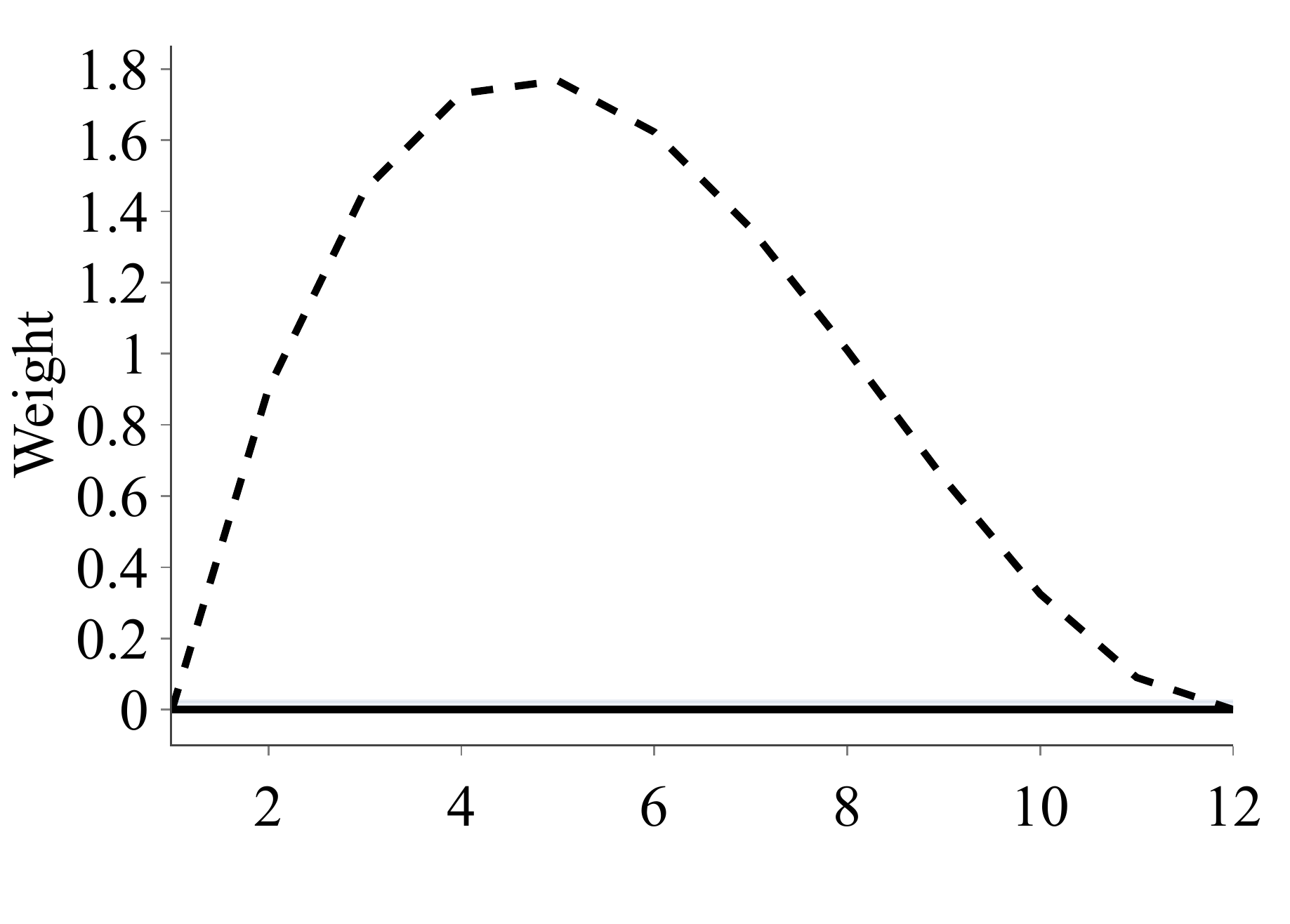}
	\end{subfigure}
	\begin{subfigure}[b]{0.25\linewidth}
		\caption{\scriptsize LASSO-MIDAS}
		\centering\includegraphics[width=4cm]{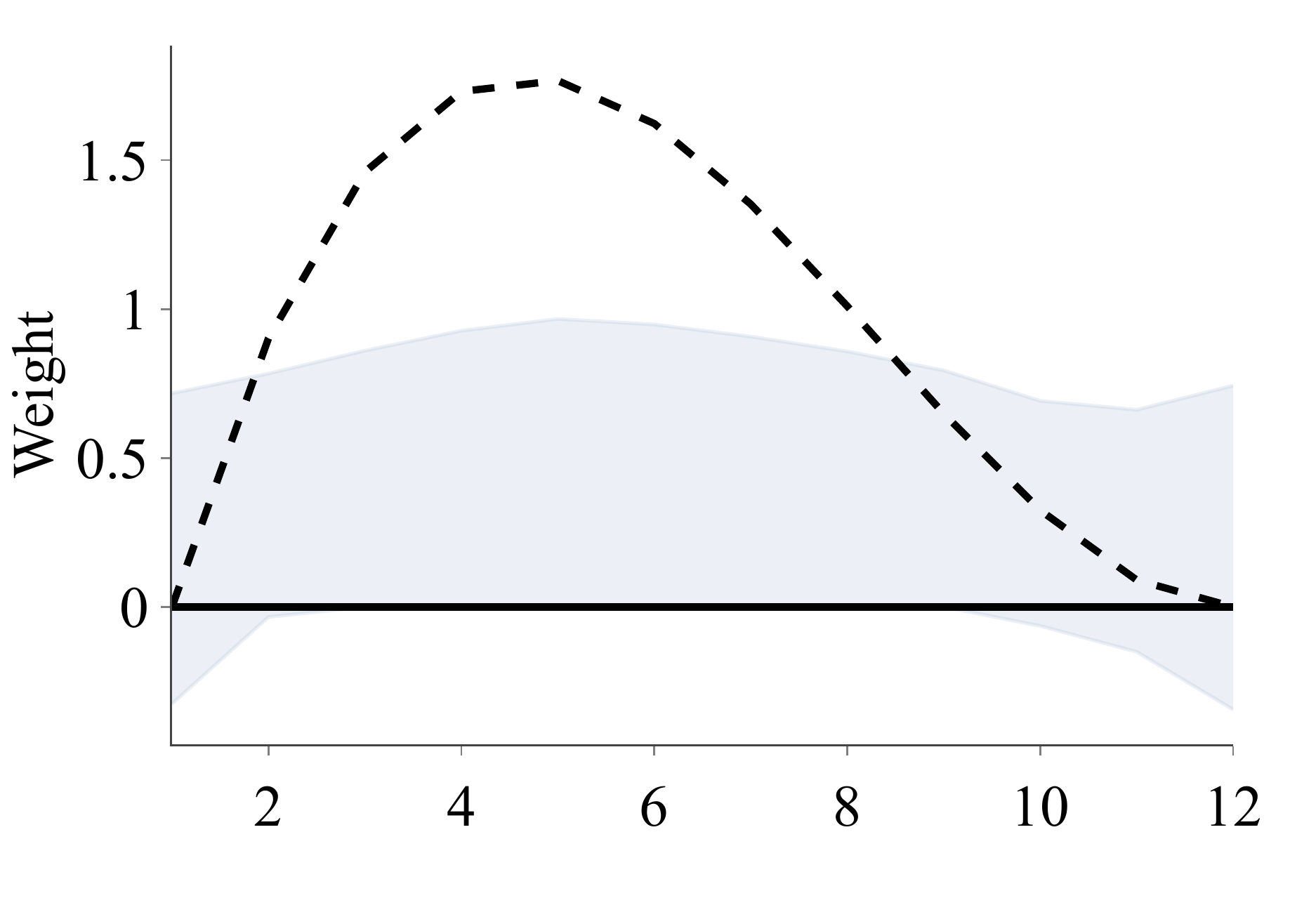}
	\end{subfigure}
	\begin{subfigure}[b]{0.25\linewidth}
		\caption{\scriptsize sg-LASSO-MIDAS}
		\centering\includegraphics[width=4cm]{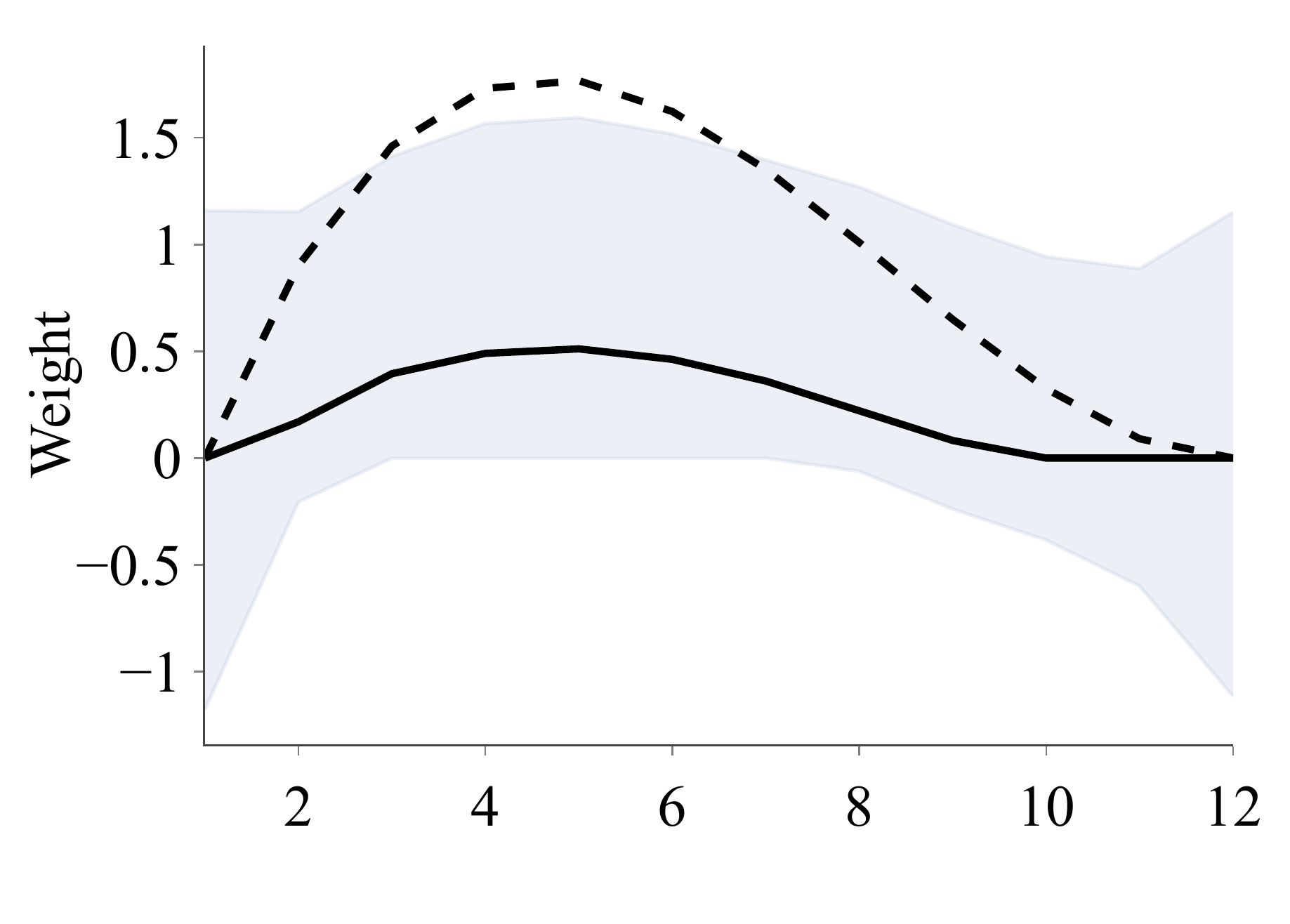}
	\end{subfigure}	
	
	\begin{subfigure}[b]{0.25\linewidth}
		\caption{\scriptsize LASSO-U-MIDAS}
		\centering\includegraphics[width=4cm]{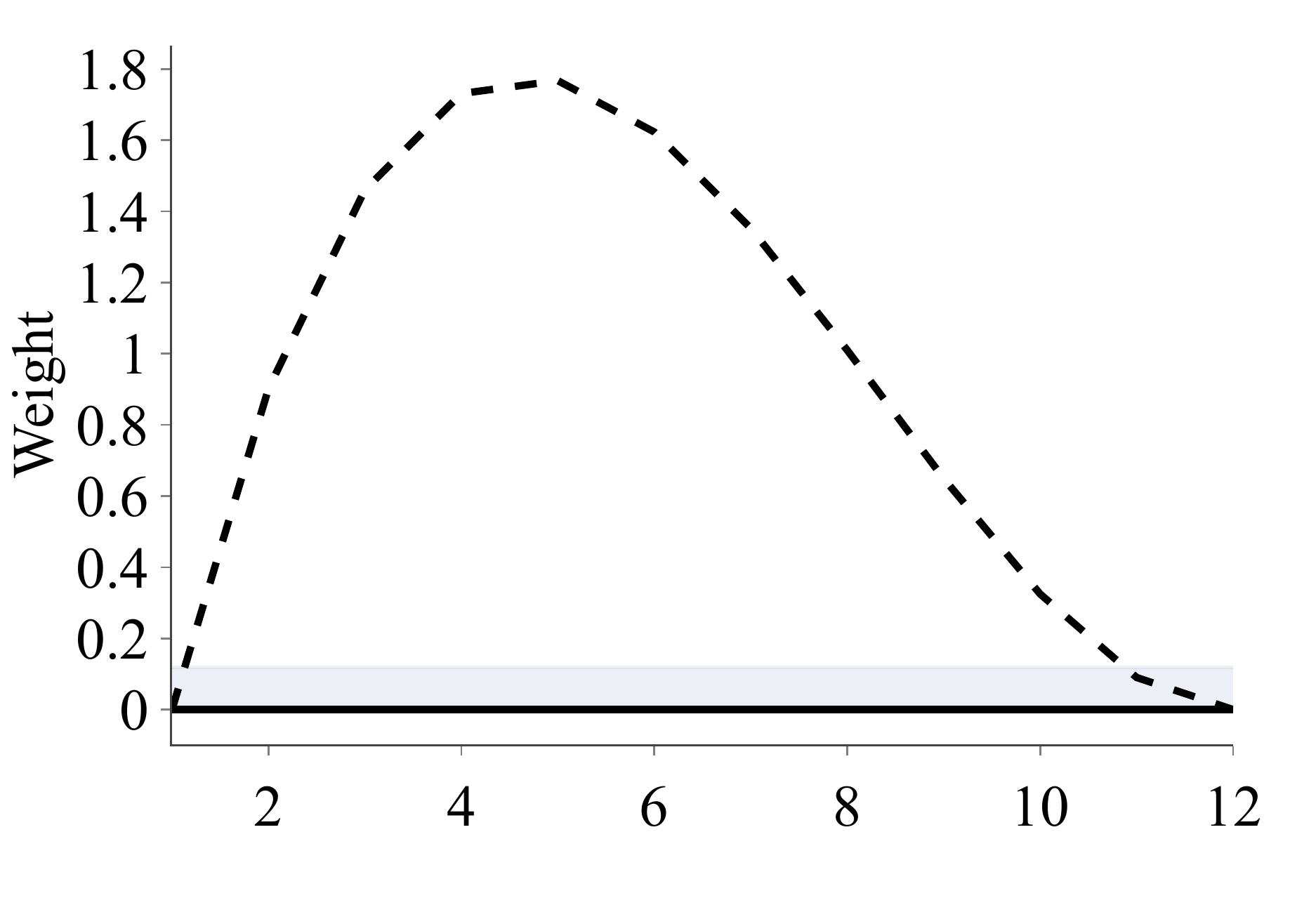}
	\end{subfigure}
	\begin{subfigure}[b]{0.25\linewidth}
		\caption{\scriptsize LASSO-MIDAS}
		\centering\includegraphics[width=4cm]{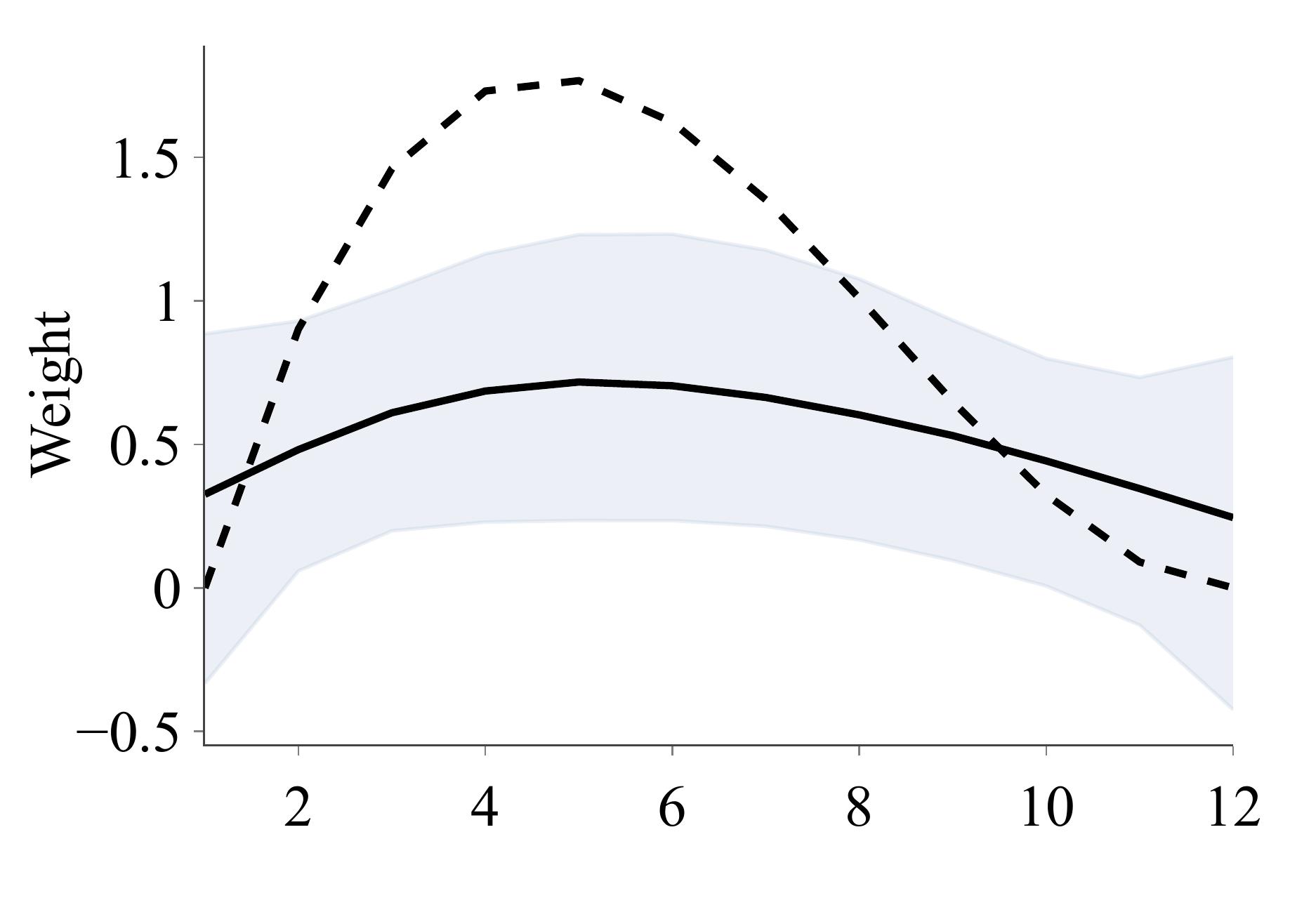}
	\end{subfigure}
	\begin{subfigure}[b]{0.25\linewidth}
		\caption{\scriptsize sg-LASSO-MIDAS}
		\centering\includegraphics[width=4cm]{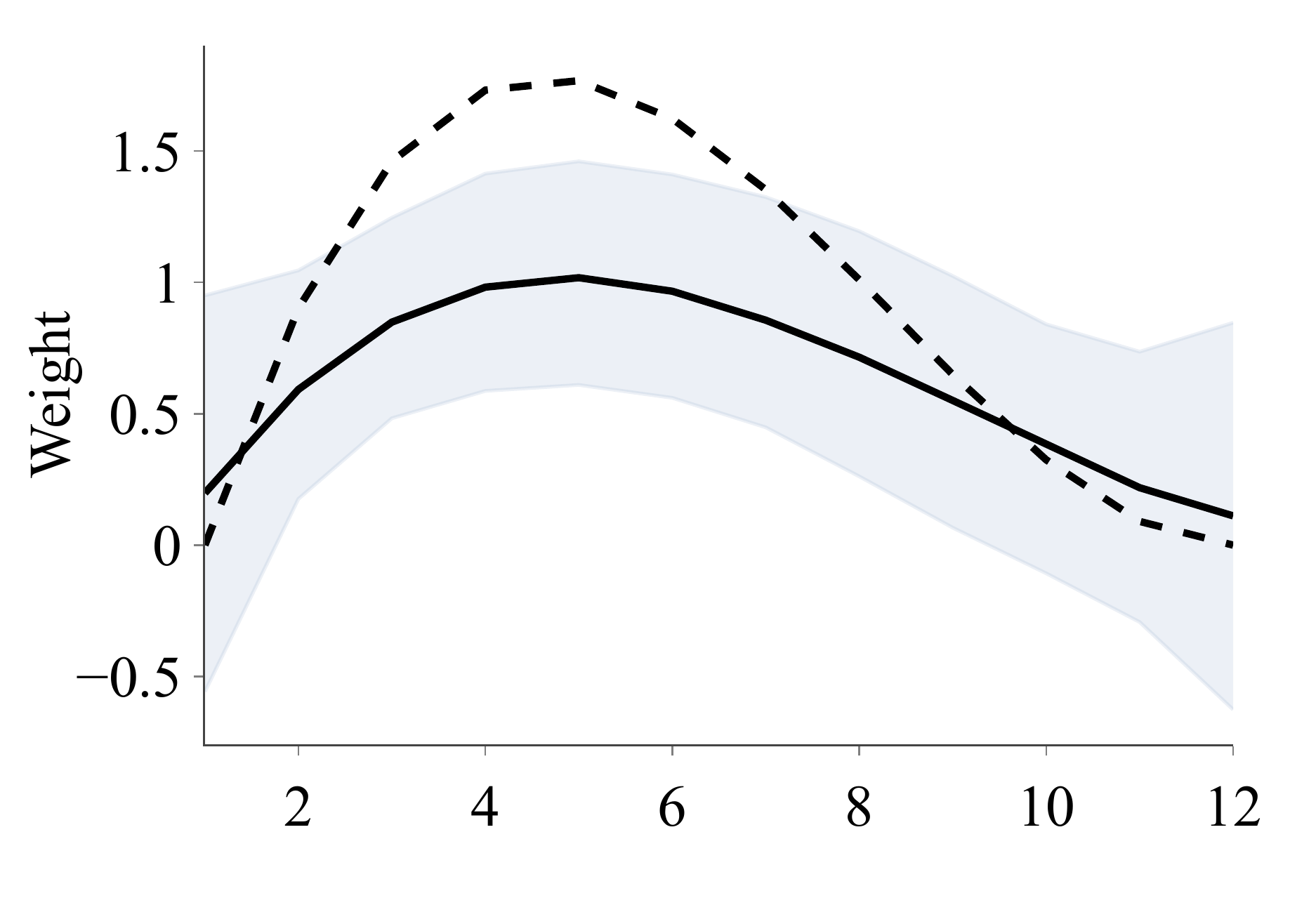}
	\end{subfigure}
	\caption{\footnotesize The figure shows the fitted Beta(2,3) weights. We plot the estimated weights for the LASSO-U-MIDAS, LASSO-MIDAS, and sg-LASSO-MIDAS estimators for the baseline DGP scenario. The first row plots weights for the sample size $T=50$, the second row plots weights for the sample size $T=200$. The black solid line is the median estimate of the weights function, the black dashed line is the population weight function, and the gray area is the 90\% confidence interval.}
	\label{appendix:fig:weights_beta_2}			
\end{figure}

\begin{figure}[!htbp]
	\centering
	\begin{subfigure}[b]{0.25\linewidth}
		\caption{\scriptsize LASSO-U-MIDAS}
		\centering\includegraphics[width=4cm]{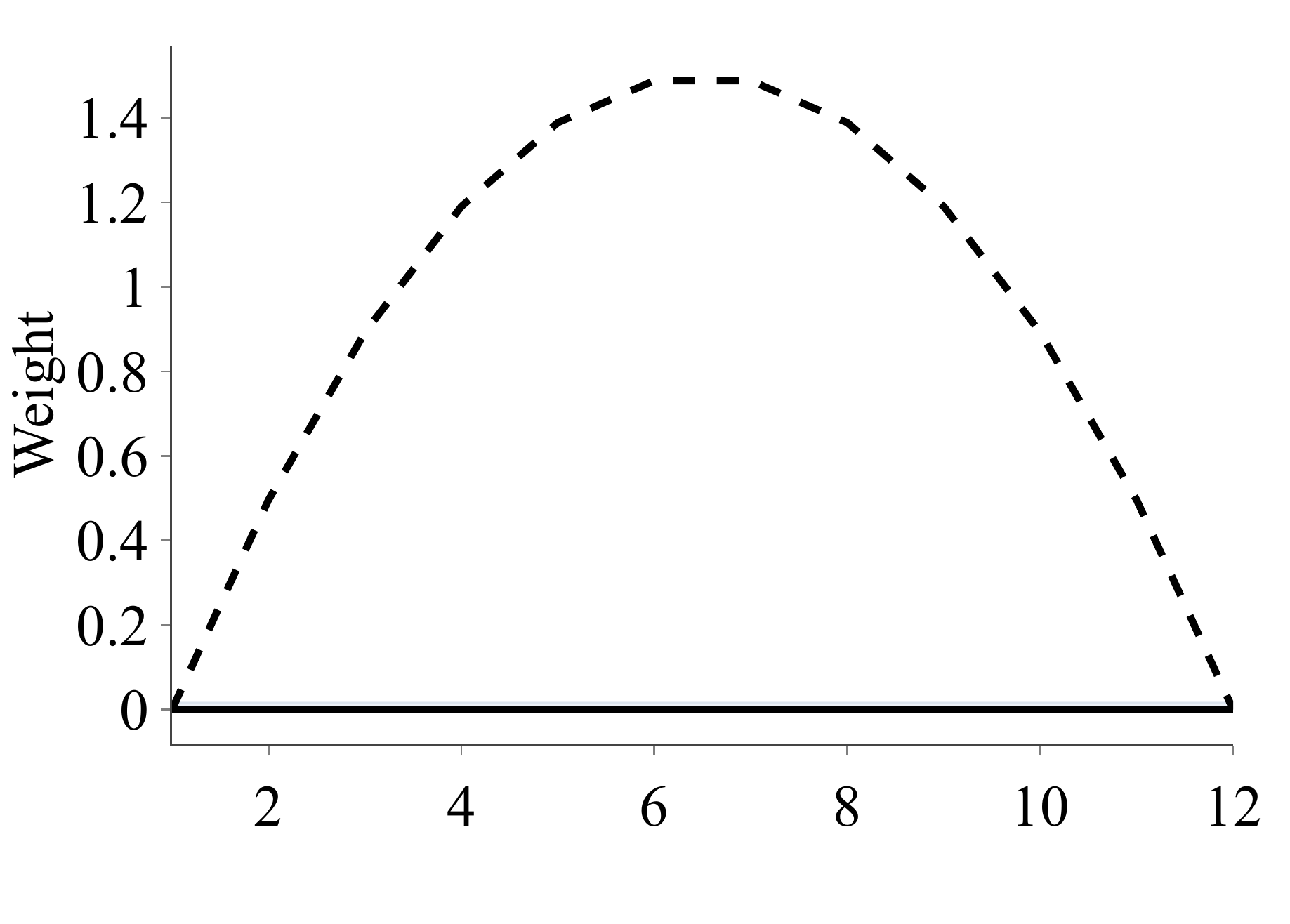}
	\end{subfigure}
	\begin{subfigure}[b]{0.25\linewidth}
		\caption{\scriptsize LASSO-MIDAS}
		\centering\includegraphics[width=4cm]{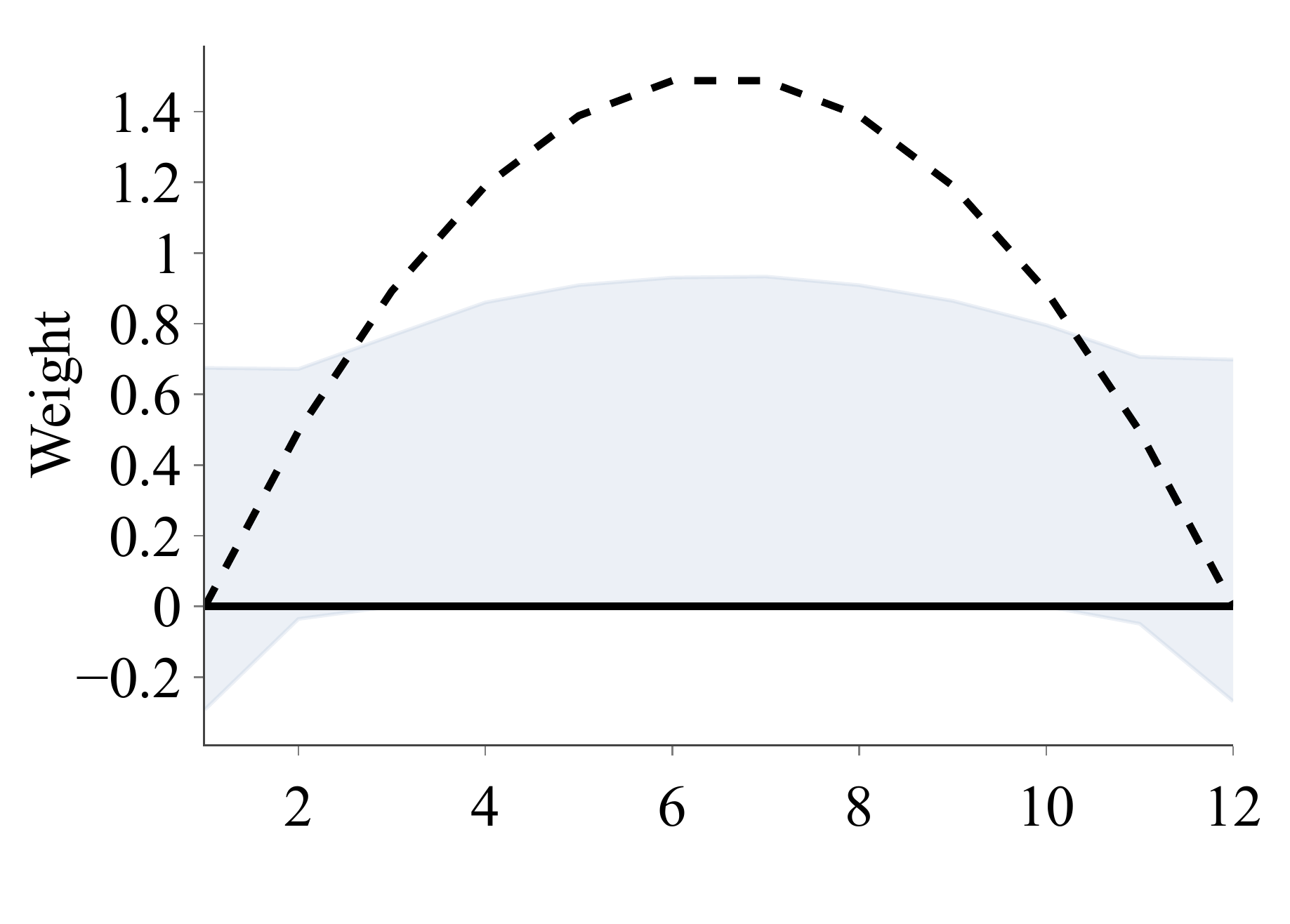}
	\end{subfigure}
	\begin{subfigure}[b]{0.25\linewidth}
		\caption{\scriptsize sg-LASSO-MIDAS}
		\centering\includegraphics[width=4cm]{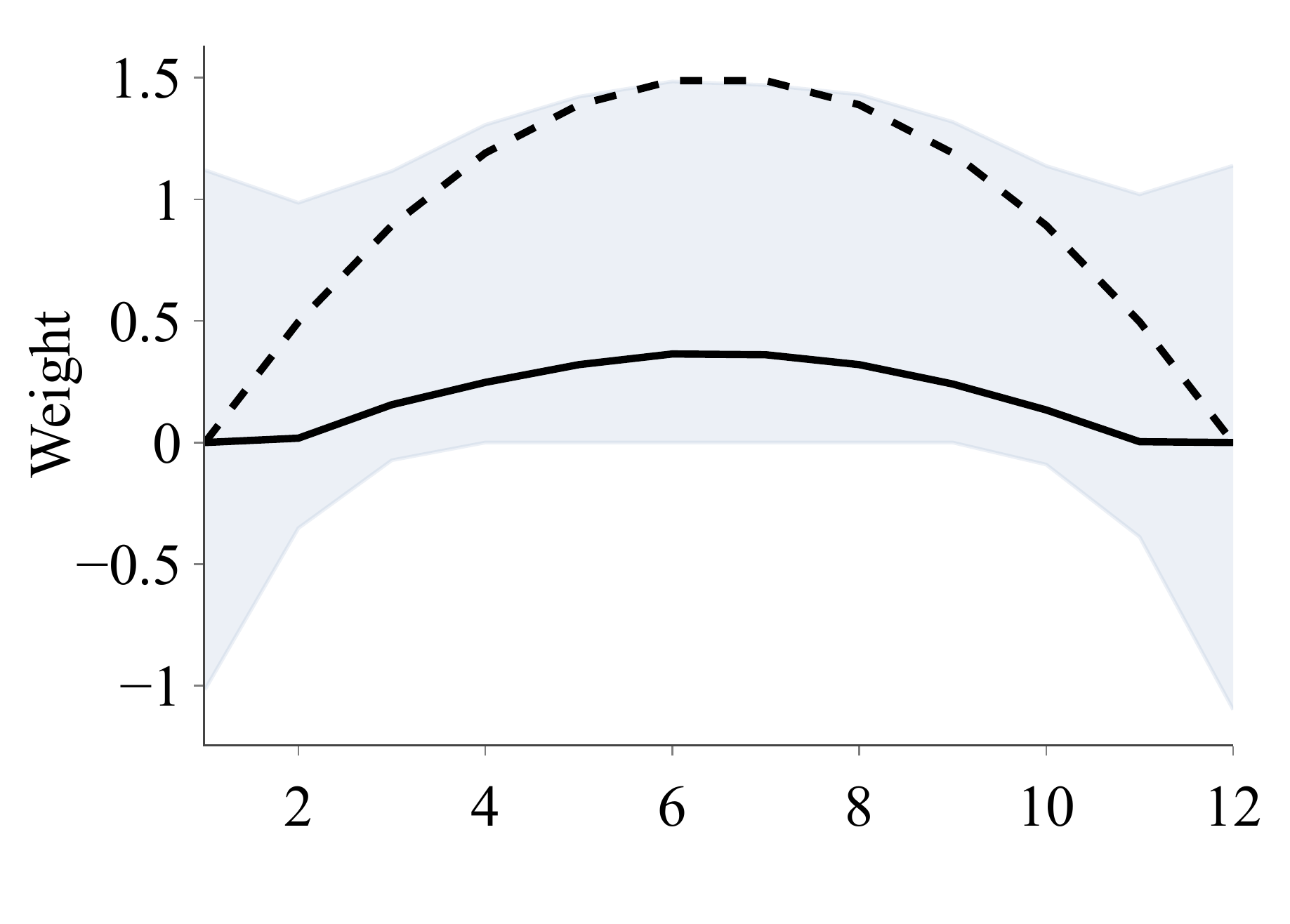}
	\end{subfigure}	
	
	\begin{subfigure}[b]{0.25\linewidth}
		\caption{\scriptsize LASSO-U-MIDAS}
		\centering\includegraphics[width=4cm]{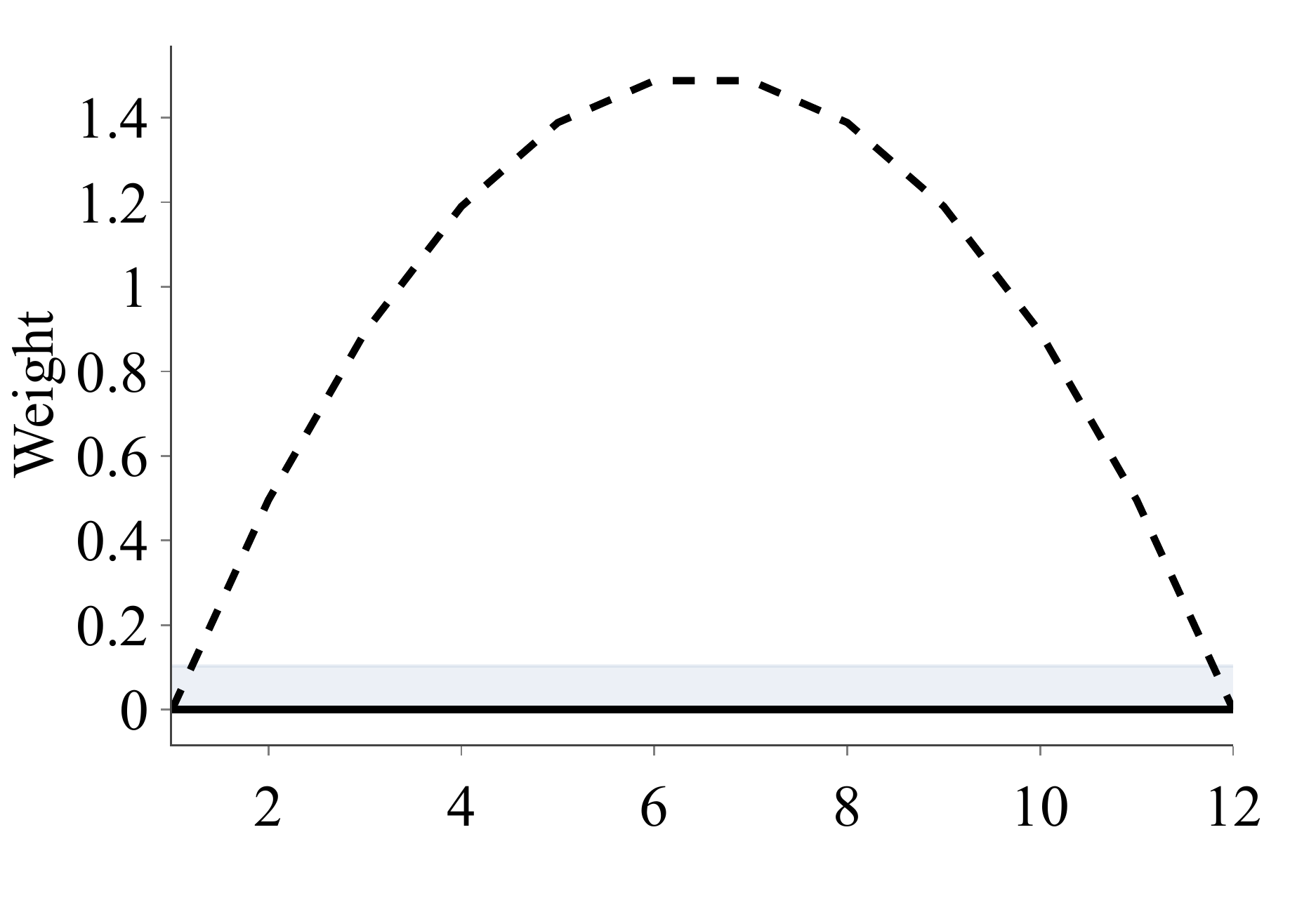}
	\end{subfigure}
	\begin{subfigure}[b]{0.25\linewidth}
		\caption{\scriptsize LASSO-MIDAS}
		\centering\includegraphics[width=4cm]{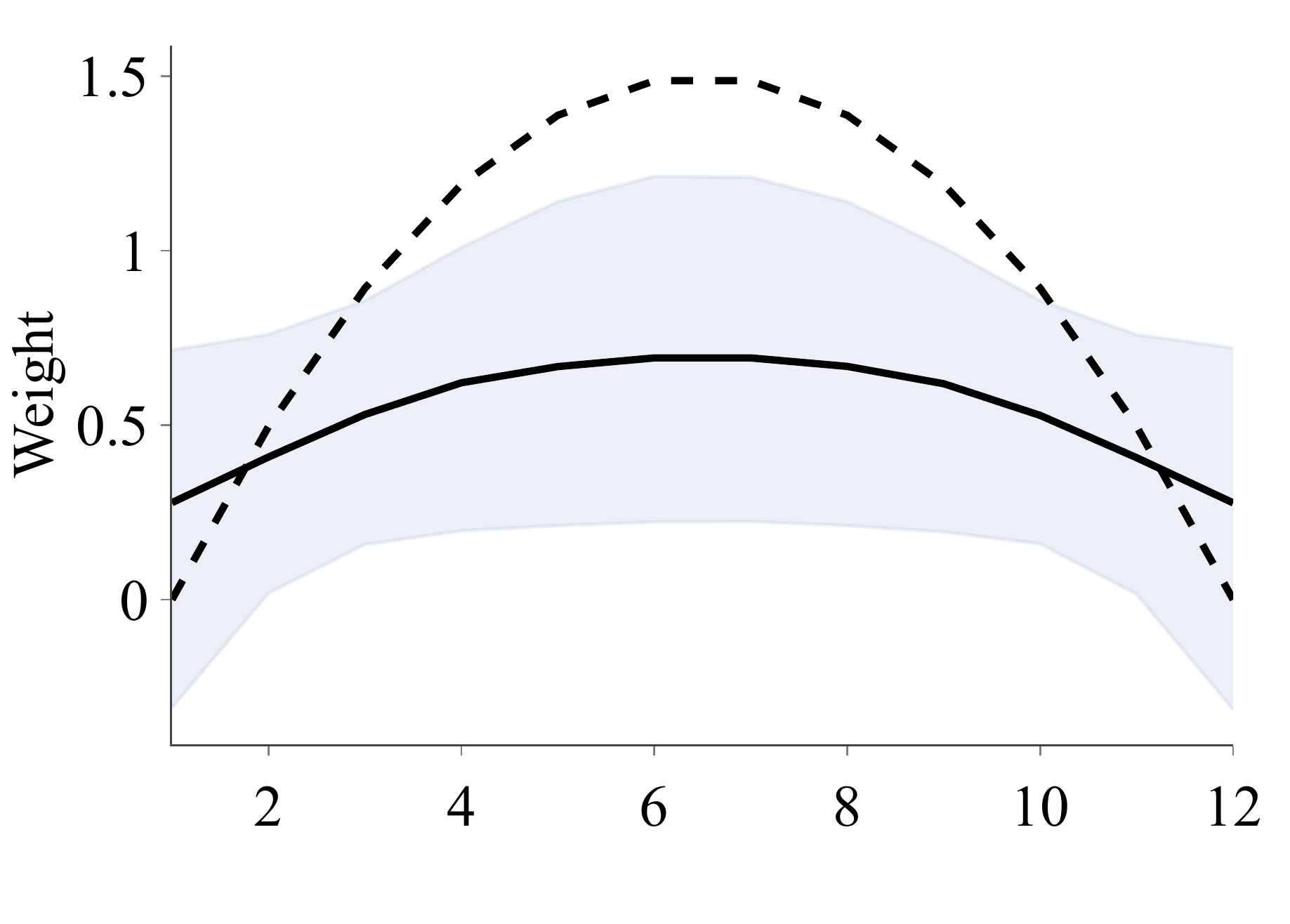}
	\end{subfigure}
	\begin{subfigure}[b]{0.25\linewidth}
		\caption{\scriptsize sg-LASSO-MIDAS}
		\centering\includegraphics[width=4cm]{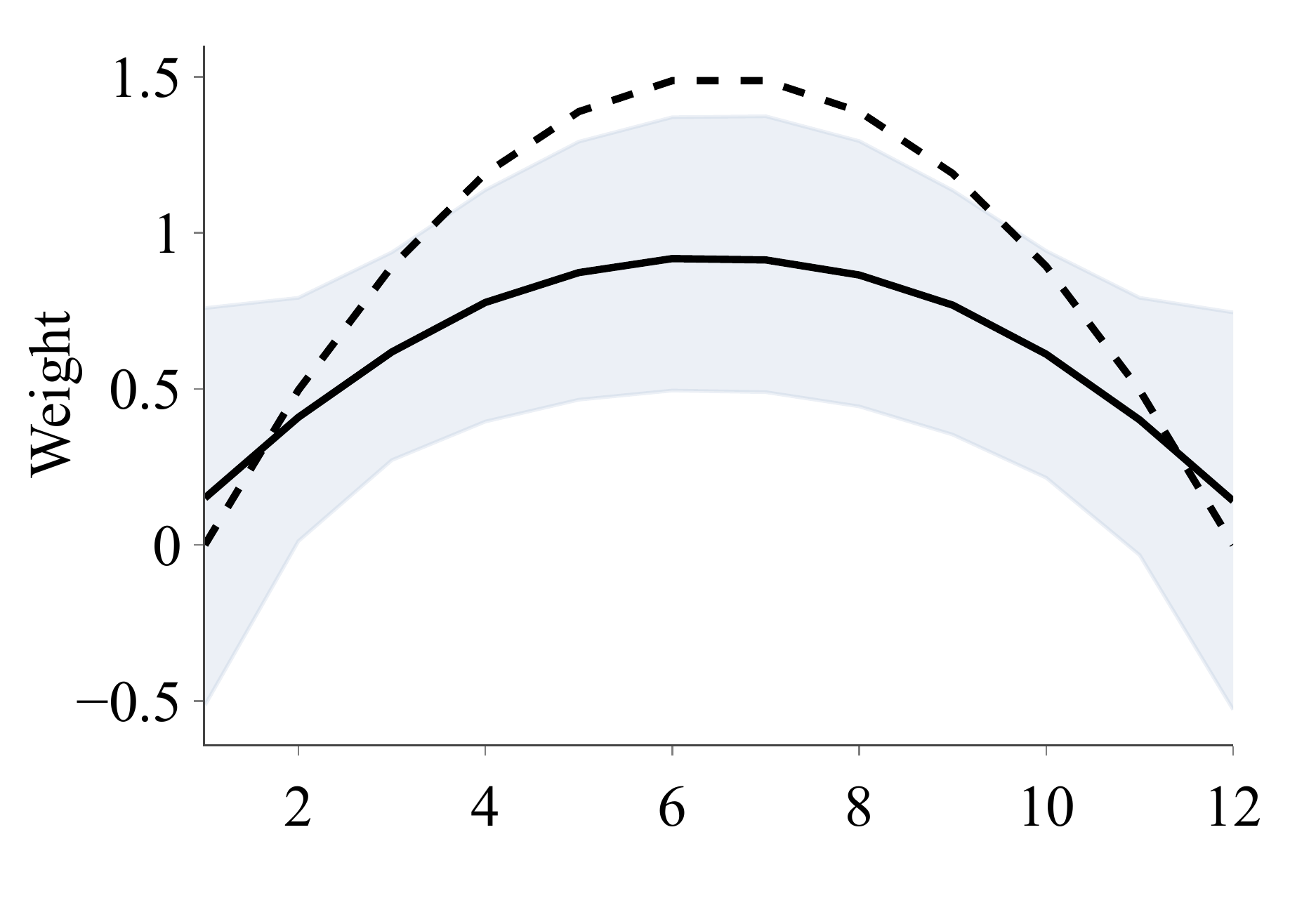}
	\end{subfigure}
	\caption{\footnotesize The figure shows the fitted Beta(2,2) weights. We plot the estimated weights for the LASSO-U-MIDAS, LASSO-MIDAS, and sg-LASSO-MIDAS estimators for the baseline DGP scenario. The first row plots weights for the sample size $T=50$, the second row plots weights for the sample size $T=200$. The black solid line is the median estimate of the weights function, the black dashed line is the population weight function, and the gray area is the 90\% confidence interval.}
	\label{appendix:fig:weights_beta_3}			
\end{figure}

\clearpage

\section{Detailed description of data and models}\label{appendix:detailed_description}

The standard macro variables are collected from {\it Haver Analytics} and {\it ALFRED} databases. ALFRED is a public data source for real-time data made available by the Federal Serve Bank of St. Louis; see the full list of the series with further details in Table \ref{tab:data_macro}. For series that are collected from the Haver Analytics database, we use {\it as reported} data, that is the first release is used for each data point. For the data that we collect from ALFRED, full data vintages are used. All the data is real-time, hence publication delays for each series are taken into consideration and we align each series accordingly. We use twelve monthly and four quarterly lags for each monthly and quarterly series respectively and apply Legendre aggregation with polynomial degree set to three. The groups are defined as lags of each series. 

\smallskip

On top of macro data, we add eight financial series which are collected from {\it FRED} database; the full list of the series appears in Table \ref{tab:data_fin_uncert}. These series are available in real time, hence no publication delays are needed in this case. We use three monthly lags and apply Legendre aggregation with polynomial degree set to two. As for macro, we group all lags of each series.  

\smallskip

Lastly, we add textual analysis covariates which are available at \href{http://structureofnews.com/}{http://structureofnews.com/}. The data is real time, i.e., topic models are estimated for each day and the monthly series are obtained by aggregating daily data; see \cite{bybee2019structure} for further details on the data construction. We use categories of series are potentially closely tied with economic activity, which are {\it Banks}, {\it Economic Growth}, {\it Financial Markets}, {\it Government}, {\it Industry}, {\it International Affairs}, {\it Labor/income}, and {\it Oil \& Mining}. In total, we add 76 news attention series; the full list is available in Table \ref{tab:data_text}. Three lags are used and Legendre aggregation of degree two is applied to each series. In this case, we group variables based on categories.

\smallskip

To make the comparison with the NY Fed nowcasts as close as possible, we use 15 years (60 quarters) of the data and use rolling window estimation. The first nowcast is for the 2002 Q1 (first quarter that NY Fed publishes its historic nowcasts) and the effective sample size starts at 1988 Q1 (taking 15 years of data accounting for lags). We calculate predictions until the sample is exhausted, which is 2017 Q2, the last date for which news attention data is available. Real GDP growth rate data vintages are taken from {\it ALFRED} database. Some macro series start later than 1988 Q1, in which case we impute zero values. Lastly, we use four lags of real GDP growth rate in all models.

\paragraph{Alternative estimators}

We implemented the following alternative machine learning nowcasting methods. The first method is the PCA factor-augmented autoregression, where we estimate the first principal component of the data panel and use it together with four autoregressive lags. We denote this model PCA-OLS. We then consider three alternative penalty functions for the same linear model: ridge, LASSO, and Elastic Net. For these methods, we leave high-frequency lags unrestricted, and thus we call these methods the unrestricted MIDAS (U-MIDAS). As for the sg-LASSO-MIDAS model, we tune one- and two-dimensional regularization parameters via 5-fold cross-validation.

\begin{table}[!htbp]
	\centering
	{\tiny
		\begin{tabular}{r |l ll}
			& Series & Source & Units \\ 
			\hline
			\hline
			1 & ADP nonfarm private payroll employment & Haver & Level change (thousands)\\
			2 & Building permits & ALFRED & Level change (thousands)\\
			3 & Capacity utilization & ALFRED & Ppt. change\\
			4 & Civilian unemployment rate & ALFRED & Ppt. change\\
			5 & CPI-U: all items & ALFRED & MoM \% change\\
			6 & CPI-U: all items less food and energy & ALFRED & MoM \% change\\
			7 & Empire State Mfg. survey: general business conditions & Haver & Index\\
			8 & Exports: goods and services & Haver & MoM \% change\\
			9 & Export price index & Haver & MoM \% change\\
			10 & Housing starts	& ALFRED & MoM \% change\\
			11 & Imports: goods and services & Haver & MoM \% change\\
			12 & Import price index & Haver & MoM \% change\\
			13 & Industrial production index & ALFRED & MoM \% change\\
			14 & Inventories: Total business & ALFRED & MoM \% change\\
			15 & ISM mfg.: PMI composite index & Haver & Index\\
			16 & ISM mfg.: Prices index & Haver & Index\\	
			17 & ISM mfg.: Employment index & Haver & Index\\
			18 & ISM nonmanufacturing: NMI composite index & Haver & Index\\
			19 & JOLTS: Job openings: total & Haver & Level change (thousands)\\
			20 & Manufacturers new orders: durable goods & ALFRED & MoM \% change\\
			21 & Manufacturing payrolls & Haver & Level change (thousands)\\
			22 & Manufacturers shipments: durable goods & Haver & MoM \% change\\
			23 & Manufacturers inventories: durable goods & Haver & MoM \% change\\
			24 & Manufacturers' unfilled orders: total manufacturing & Haver & MoM \% change\\
			25 & New single family houses sold & ALFRED & MoM \% change\\
			26 & Nonfarm business sector: unit labor cost & ALFRED & QoQ \% change (annual rate)\\
			27 & PCE less food and energy: chain price index & ALFRED & MoM \% change\\
			28 & PCE: chain price index & ALFRED & MoM \% change\\
			29 & Philly Fed Mfg. business outlook: current activity & Haver & Index\\	
			30 & Retail sales and food services & ALFRED & MoM \% change\\
			31 & Real personal consumption expenditures  & ALFRED & MoM \% change\\
			32 & Real gross domestic income & Haver & QoQ \% change (annual rate)\\
			33 & Real disposable personal income & ALFRED & MoM \% change\\
			34 & Value of construction put in place & Haver & MoM \% change\\
			\hline 
		\end{tabular}
	}
	\caption{\scriptsize Data description table (macro data)-- The \textit{Series} column gives a time-series name, which is given in the second column \textit{Source}. The column \textit{Units} denotes the data transformation applied to a time-series. \label{tab:data_macro}	}
\end{table}
\begin{table}[!htbp]
	\centering
	{\tiny
		\begin{tabular}{r |l ll}
			& Series & Source & Units \\ 
			\hline
			\hline
			1 & BAA less AAA corporate bond spread & FRED & Level \\
			2 & BAA less 10-year bond spread & FRED & Level \\
			3 & S\&500 & FRED & Log-returns \% \\
			4 & TED spread & FRED & Level \\
			5 & 10-year less 3-month bond spread & FRED & Level \\ 
			6 & VIX & FRED & Level \\
			7 & Economic policy uncertainty index (EPUI) & FRED & Index \\
			8 & Equity market-related economic uncertainty index (EMEUI) & FRED & Index \\
			\hline 
		\end{tabular}
	}
	\caption{\scriptsize Data description table (financial and uncertainty series) -- The \textit{Series} column gives a time-series name, which is given in the second column \textit{Source}. The column \textit{Units} denotes the data transformation applied to a time-series. \label{tab:data_fin_uncert}	}
\end{table}

\clearpage
{\tiny
	\begin{longtable}[!htbp]{r | cc}
		&  Group & Series   \\ 
		\hline
		\hline
		1 & Banks  & Bank loans \\ 
		2 & Banks  & Credit ratings \\ 
		3 & Banks  & Financial crisis \\
		4 & Banks  & Mortgages \\
		5 & Banks  & Nonperforming loans \\
		6 & Banks  & Savings \& loans \\
		7 & Economic Growth & Economic growth \\
		8 & Economic Growth & European sovereign debt \\
		9 & Economic Growth &  Federal Reserve \\
		10 & Economic Growth & Macroeconomic data \\
		11 & Economic Growth & Optimism \\
		12 & Economic Growth & Product prices \\
		13 & Economic Growth & Recession \\
		14 & Economic Growth &  Record high \\
		15 & Financial Markets & Bear/bull market \\
		16 & Financial Markets & Bond yields \\
		17 & Financial Markets & Commodities \\
		18 & Financial Markets & Currencies/metals \\
		19 & Financial Markets & Exchanges/composites \\
		20 & Financial Markets & International exchanges \\
		21 & Financial Markets & IPOs \\
		22 & Financial Markets & Options/VIX \\
		23 & Financial Markets & Share payouts \\
		24 & Financial Markets & Short sales \\
		25 & Financial Markets & Small caps \\
		26 & Financial Markets & Trading activity \\
		27 & Financial Markets & Treasury bonds \\
		28 & Government & Environment \\
		29 & Government & National security \\
		30 & Government & Political contributions \\
		31 & Government & Private/public sector \\
		32 & Government & Regulation \\
		33 & Government & Safety administrations \\
		34 & Government & State politics \\
		35 & Government & Utilities \\
		36 & Government & Watchdogs \\
		37 & Industry & Cable \\
		38 & Industry & Casinos \\
		39 & Industry & Chemicals/paper \\
		40 & Industry & Competition \\
		41 & Industry & Couriers \\
		42 & Industry & Credit cards \\
		43 & Industry & Fast food \\
		44 & Industry & Foods/consumer goods \\
		45 & Industry & Insurance \\
		46 & Industry & Luxury/beverages \\
		47 & Industry & Revenue growth \\
		48 & Industry & Small business \\
		49 & Industry & Soft drinks \\
		50 & Industry & Subsidiaries \\
		51 & Industry & Tobacco \\
		52 & Industry & Venture capital \\
		53 & International Affairs & Canada/South Africa \\
		54 & International Affairs & China \\
		55 & International Affairs & France/Italy \\
		56 & International Affairs & Germany \\
		57 & International Affairs & Japan \\
		58 & International Affairs & Latin America \\
		59 & International Affairs & Russia \\
		60 & International Affairs & Southeast Asia \\
		61 & International Affairs & Trade agreements \\
		62 & International Affairs & UK \\
		63 & Labor/income & Executive pay \\
		64 & Labor/income & Fees \\
		65 & Labor/income & Government budgets \\
		66 & Labor/income & Health insurance \\
		67 & Labor/income & Job cuts \\
		68 & Labor/income & Pensions \\
		69 & Labor/income & Taxes \\
		70 & Labor/income & Unions \\
		71 & Oil \& Mining & Agriculture \\
		72 & Oil \& Mining & Machinery \\
		73 & Oil \& Mining & Mining \\
		74 & Oil \& Mining & Oil drilling \\
		75 & Oil \& Mining & Oil market \\
		76 & Oil \& Mining & Steel \\
		\hline 
		\caption{\scriptsize Data description table (textual data) -- The \textit{Group} column is a group name of individual textual analysis series which appear in the column \textit{Series}. Data is taken in levels. \label{tab:data_text}	}
	\end{longtable}
}

\clearpage

\end{document}